\RequirePackage{fix-cm}
\documentclass[twocolumn]{svjour3}          
\smartqed  
\usepackage{graphicx}
\usepackage{cite}
\usepackage{amsmath,amssymb,amsfonts}
\usepackage{slashed}
\usepackage{booktabs}
\usepackage{algorithm}
\usepackage[noend]{algpseudocode}

\usepackage{epsfig}
\usepackage{subfig}
\usepackage[colorlinks,linkcolor=blue]{hyperref}
\usepackage{caption}
\usepackage{bm}
\usepackage{color}
\definecolor{XFcolor}{RGB}{209,186,116}
\usepackage{textcomp}
\usepackage[outdir=./]{epstopdf}

\begin{document}

\title{Efficiently Finding a Maximal Clique Summary via Effective Sampling
}


 \author{\small Xiaofan Li         \and
        Rui Zhou \and Lu Chen \and Chengfei Liu \and Qiang He \and Yun Yang 
}


\institute{ X.~Li, R.~Zhou, L.~Chen, C.~Liu, Q.~He and Y.~Yang  \at
              Swinburne University of Technology, Australia  \\
              \email{\{xiaofanli,rzhou,luchen,cliu,qhe,yyang\}@swin.edu.au}           
}

\date{Received: date / Accepted: date}

\maketitle

\begin{abstract}
Maximal clique enumeration (MCE) is a fundamental problem in graph theory and is used in many applications, such as social network analysis, bioinformatics, intelligent agent systems, cyber security, etc.  Most existing MCE algorithms focus on improving the efficiency rather than reducing the  output size. The output unfortunately could consist of a large number of maximal cliques. 
In this paper, we study how to report a summary of less overlapping maximal cliques.
The problem was studied before, however, after examining the pioneer approach, we consider it still not satisfactory.
To advance the research along this line, our paper attempts to make four contributions: 
(a) we propose a more effective sampling strategy, which produces a much smaller summary but still ensures that the summary can somehow witness all the maximal cliques and the expectation of each maximal clique witnessed by the summary is above a predefined threshold;  
(b) we prove that the sampling strategy is optimal under certain optimality conditions;  
(c) {we  apply clique-size bounding and 
design new enumeration order 
to approach the optimality conditions;} 
and (d) to verify  experimentally, we test eight real benchmark datasets that have a variety of graph characteristics. 
The results show that our new sampling strategy consistently outperforms the state-of-the-art approach by producing smaller summaries and running faster on all the datasets.

\end{abstract}

\section{Introduction}
\label{sec:introduction}

A clique $C$ is a complete subgraph of an undirected graph~$G(V, E)$, which means that each pair of nodes in $C$ have an edge between them. A maximal clique is a clique which is not a subgraph of any other clique. The procedure of enumerating all maximal cliques in a graph is called Maximal Clique Enumeration (MCE).

MCE has a range of applications in different fields, such as discovering communities in social networks~\cite{lu2018community},
identifying co-expressed genes~\cite{rokhlenko2007similarities}, 
detecting protein-protein interaction complexes~\cite{zhang2008pull}, 
supporting the construction of intelligent agent systems~\cite{Tandon:2018:ASH:3237383.3238081} and recognizing emergent patterns in terrorist networks~\cite{berry2004emergent}.

There are a sufficient number of works \cite{koch_enumerating_2001,tomita_worst-case_2006,schmidt_scalable_2009,xu_distributed_2016,san_segundo_efficiently_2018,naude_refined_2016,eppstein_listing_2013} focusing on improving the efficiency of MCE, which is considered as having  exponential time.
This is probably because the number of cliques in a graph is always very large. 
A graph with less than $10^6$ vertices and $7\times 10^6$ edges can have more than $10^7$ maximal cliques~\cite{wang_redundancy-aware_2013}. 
Counting the number of maximal cliques in a general graph is considered to be \#P-complete~\cite{valiant1979complexity}. 
This means that the output of any MCE procedure is hard to be used by some other post applications. Fortunately, there typically exist a lot of overlaps between different cliques. 
 This motivates us to consider  reporting a summary set of all maximal cliques which has less overlap but can somehow represent all the cliques.

Wang et al.~\cite{wang_redundancy-aware_2013} introduced the concept of $\tau$-visible summary, a set of maximal cliques, which promises that every maximal clique in  graph $G$ can be covered by at least one maximal clique in the summary with a ratio of at least $\tau$. Here, $\tau$ is given by a user and reflects the user's tolerance of overlap. 
For example, a  summary with $\tau=0.8$ ensures that any maximal clique can have at least 80$\%$ nodes covered by some clique in the summary. 
This summary model is interesting, e.g., in the marketing domain, if a certain percentage of users in a clique community has been covered, we  expect that the covered users will spread a message across the community. Consequently, finding fewer communities as targets due to marketing cost while still ensuring a broad final user coverage is very desirable.
The work~\cite{wang_redundancy-aware_2013} modified the depth-first MCE \cite{bron_algorithm_1973} by adding a sampling function that determines whether a new clique enumeration sub-procedure should be entered. It was proved that the \emph{expected visibility} of such a sampled summary is larger than $\tau$.


However, expected $\tau$-visible summaries are not unique. Apparently, as long as a summary is $\tau$-visible, the more concise the summary is, the better the summary is. {Hence three questions} arise naturally in sequence: 
\begin{itemize}
	\item[(1)] Is there any sampling strategy that can find a better (smaller) expected $\tau$-visible summary?
	\item[(2)] What kind of sampling strategy is optimal?
	\item[(3)] If achieving the optimal is difficult or impossible, how can we provide the best effort?
	
\end{itemize}

We will tackle these three questions in this paper.
{For question (1), the answer is yes.  
The state-of-the-art work~\cite{wang_redundancy-aware_2013} used a sampling function to determine whether each maximal clique should be included into the growing summary. 
However, it could include many redundant maximal cliques that 
(a) are visible to the current summary; 
(b) are likely to be visible to the future summary 
(explained by Observation 1 for (a) and by Observation 2  for (b) in Section~\ref{subsec:idea}). 
By identifying and discarding these two types of maximal cliques, we find a novel sampling function that is  superior to the existing work in terms of both output summary size and running time. 
For question (2), we naturally define the optimality as including each maximal clique with the lowest probability while still promising $\tau$-visibility of the summary. 
We find that, 
 when (I) the maximal cliques are enumerated in a proper order 
 and (II) the size of a clique can be estimated sufficiently accurate at an early stage of the depth-first search, our newly proposed sampling function successfully guarantees  optimality, while the existing approach~\cite{wang_redundancy-aware_2013} does not. 
The optimality of our sampling function promises significant improvement $vs.$~\cite{wang_redundancy-aware_2013} in terms both effectiveness and efficiency. 
Although, the above two optimality conditions (I) and (II) are hard to hold ideally, they provide us two important directions to approach the optimal by looking for
(i) better vertex orders to enumerate maximal cliques, and (ii) better maximal clique size estimation techniques.  
Thus for question (3), inspired by the cohesive structure of $k$-truss, we design a novel vertex order based on truss decomposition. 
Following this vertex order, our algorithm can enumerate maximal cliques in such an order that two consecutive generated maximal cliques tend to be contained by the same cohesive $k$-truss  and thus overlap more. 
When one of the two cliques is included into the summary first, the second one has a higher probability to be discarded so  the summary keeps concise. 
Besides, we utilize $k$-truss to estimate the size of a maximal clique contained by a subgraph, 
which provides an upper bound tighter 
than the $k$-core based bound (used by~\cite{wang_redundancy-aware_2013}),  
since a $k$-truss must be a $(k-1)$-core.  
With these two strategies, we provide our best effort to improve the practical performance of our proposed sampling approach. 
}

{Our main contributions} are summarized as follows:
\begin{itemize}
	\item We introduce a new sampling strategy to help to identify an expected $\tau$-visible maximal clique summary. We prove that the new sampling strategy guarantees a better performance than the state-of-the-art method in terms of producing a smaller summary while still meeting the threshold $\tau$. 
	\item We give a theoretical analysis that the sampling can be optimal under certain conditions, which substantiates good performance of the proposed sampling strategy in practice. Future investigations could also be directed by exploring how to approximate the optimal conditions. 
	\item We show that the sampling approach can get close to optimal with clique size bounding and enumeration ordering strategies.  
	{Then we propose truss ordering and truss bound respectively to further improve the performance of our sampling strategy.}  
	\item {We conduct experimental studies to verify the superiority of the new sampling method as well as our newly designed truss order and truss bound in terms of both effectiveness and efficiency  on {eight}  real-world datasets.} 
	
\end{itemize}

The rest of this paper is organized as follows. 
In Section \ref{sec:visible}, we review the definition of $\tau$-visible summary and an existing sampling approach. In Section \ref{sec:new}, we give our motivation, introduce a novel sampling function and prove its superiority. 
The conditions of optimality are analyzed in Section \ref{sec:opt}. 
{We propose truss vertex order and truss bound to practically instantiate optimality conditions in Section~\ref{sec:boundAndorder}.}
Extensive experiments are conducted in Section \ref{sec:exp}. 
Related work and conclusion are in Section \ref{sec:related} and Section \ref{sec:conclusion}.

\section{$\tau$-Visible Summary}\label{sec:visible}

A clique refers to a complete subgraph of an undirected graph $G(V,E)$. A clique $C$ is maximal if it is not contained by any other clique. 
When the context is clear, we also use $C$ to denote the node set of a maximal clique.
Given the set of all maximal cliques in  graph $G$, denoted as $\mathcal{M}(G)$, a summary $\mathcal{S}$ is a subset of $\mathcal{M}(G)$ which means $S\subseteq \mathcal{M}(G)$.
To measure to what extent a summary can witness a clique, \emph{visibility} is defined in~\cite{wang_redundancy-aware_2013},  restated as 
Definitions~\ref{def:visibility} and~\ref{def:exact}. 
We then introduce \emph{expected visibility} in ~Definitions~\ref{def:exp} and~\ref{def:exps}.
\begin{definition}[Visibility]
	Given a summary $\mathcal{S}$, the visibility $\mathcal{V}_{\mathcal{S}}:\mbox{ }\mathcal{M}(G)\rightarrow[0,1]$ of a maximal clique $C$ is defined as:
	\begin{equation}\label{3}
		\mathcal{V}_{\mathcal{S}}(C)=\max\limits_{C'\in \mathcal{S}}\frac{|C\cap C'|}{|C|}
	\end{equation}
\label{def:visibility}	
\end{definition}
\vspace{-12pt}
Note that $C'$ is allowed to be the same as $C$. This means that if $C \in \mathcal{S}$, $C$'s visibility with respect to $\mathcal{S}$ is 1. 
In other words, if $C \in \mathcal{S}$, the summary $\mathcal{S}$ can completely witness $C$.

\begin{definition}[$\tau$-Visible Summary]
	A summary $\mathcal{S}$ is called $\tau$-visible iff $\mbox{ }\forall C \in \mathcal{M}(G)$,
	\begin{equation}
		\mathcal{V}_{\mathcal{S}}(C)\geq \tau
	\end{equation}
\label{def:exact}
\end{definition}
\vspace{-15pt}
Rather than the exact $\tau$-visible summary defined above, our work looks for an \emph{expected $\tau$-visible} summary.
Before we give the formal definition of expected $\tau$-visible summary, we explain what the term \emph{expected} means intuitively. 
Since the number of maximal cliques is likely to be exponential, it is infeasible to firstly compute all the cliques and then decide the summary. Instead, it is more practical to decide on the way while enumerating, i.e., try to make a decision whether to keep/discard a new clique or keep/discard with a probability, when the clique is found. To be more active, a decision can be made on whether to enter each enumeration branch with some probability. This means that 
 each maximal clique has a probability $Pr[C\in \mathcal{S}]$  to be included in $\mathcal{S}$ and a corresponding  probability $Pr[C\notin \mathcal{S}]= 1 - Pr[C\in \mathcal{S}]$ to be discarded. 
For a clique $C$, if it is selected  to be included into $\mathcal{S}$, the visibility of it should be $1$, since it is witnessed by itself; otherwise this value is $\mathcal{V}_{\mathcal{S}}(C)$, which stays unknown before $\mathcal{S}$ is finalized. Given the above discussion of visibility, we can have the mathematical expectation of $\mathcal{V}_{\mathcal{S}}(C)$ in Definition~\ref{def:exp}:
\begin{definition}[Expected Visibility]\label{def:exp}
The expected visibility of a clique $C$ with regards to a summary  $\mathcal{S}$, $\mathbb{E}[\mathcal{V}_{\mathcal{S}}(C)]$, is defined as 
\begin{equation}\label{eq:exp}
    		\mathbb{E}[\mathcal{V}_{\mathcal{S}}(C)]\triangleq 1\cdot Pr[C\in \mathcal{S}]+\mathcal{V}_{\mathcal{S}}(C)\cdot Pr[C\notin \mathcal{S}]
\end{equation}
\end{definition}

One may question that, before $\mathcal{S}$ is finally known, $\mathcal{V}_{\mathcal{S}}(C)$ is unavailable to Formula~(\ref{eq:exp}), since this value relies on a materialization  of $\mathcal{S}$. However, we need to point out that,  such a $\mathcal{V}_{\mathcal{S}}(C)$ does exist albeit it is hard to know its value early. We will see under which conditions $\mathcal{V}_{\mathcal{S}}(C)$ can be calculated   without $\mathcal{S}$ is known in Section~\ref{sec:opt}. Currently, we only need a lower bound of it since 
we want to make sure the lower bound is sufficiently large, so that the expectation of $\mathcal{V}_{\mathcal{S}}(C)$ is larger than a user-given threshold, implying that we want to find a summary with good visibility expectation guarantee.
Definition~\ref{def:exps} defines this case:
\begin{definition}[Expected $\tau$-Visible Summary]\label{def:exps}\\
	A summary $\mathcal{S}$ is 
	{\it expected $\tau$-visible} iff $\mbox{ }\forall C \in \mathcal{M}(G)$,
	\begin{equation}
		\mathbb{E}[\mathcal{V}_{\mathcal{S}}(C)]\geq \tau
\end{equation}
where $\tau\in[0,1]$ is a given threshold. 
\end{definition}
\textbf{In this paper, we focus on developing theories and algorithms for finding a good expected $\tau$-visible summary. }
The key issue that we are going to address is how to keep/discard the enumeration branches to ensure the final found cliques can form a summary which is $\tau$-visible and of a small size.
Note that in an expected $\tau$-visible summary, there may exist a clique which cannot be covered by any other clique in the summary with a factor more than the extent of $\tau$. However, we still aim for expected visibility rather than exact visibility, because (1) visibility itself has already meant that the summary is an approximation, hence there may be less gain to enforce exact visibility; (2) the basic MCE algorithm (BK-MCE) which we will introduce in Section~\ref{subsec:mce} is a depth-first search approach. For expected visibility semantics, the great pruning power of a sampling approach can terminate search subtrees as early as possible so that the exponential search space can be reduced significantly and the summary is promised to be concise with a sufficient quality guarantee. 
An algorithm serving for exact visibility 
has to decide whether a search subtree can be discarded at a relatively late stage, thus  slows down the running time.

Next, we start with introducing an existing  depth-first MCE procedure \cite{bron_algorithm_1973} for maximal clique enumeration in Section~\ref{subsec:mce}, and then explain how it can be modified to find an expected $\tau$-visible summary by the state-of-the-art work~\cite{wang_redundancy-aware_2013} in Section~\ref{subsec:sample}.
Important notations are listed in Table~\ref{notations}.

\subsection{Maximal Clique Enumeration}
\label{subsec:mce}
\begin{table}
\renewcommand\arraystretch{1.2}
\small
\centering
\caption{Notations}
\begin{tabular}{p{40pt}p{185pt}}
\toprule
Notation  & Meaning \\
\midrule
 $G(V,E)$ & the graph $G$ with vertex set $V$ and edge set $E$  \\
 $G_T$&the induced graph of vertex set $T$ on graph $G$\\
  $C$  & a  maximal clique  \\
   $\mathcal{M}(G)$ &  the set of all maximal cliques in  graph $G$ \\
  $\mathcal{S}$  & a summary, which is a subset of $\mathcal{M}(G)$  \\
  $\mathcal{V}_{\mathcal{S}}(C)$  & the visibility of  maximal clique $C$ w.r.t. summary $\mathcal{S}$  \\
  $\mathbb{E}[\mathcal{V}_{\mathcal{S}}(C)]$ &   the expectation of visibility $\mathcal{V}_{\mathcal{S}}(C)$\\
      $\tau$&the user-specified threshold   \\

   $\mathcal{N}(v)$ &  the neighbor node set of node $v$ \\
   $T$ &  the candidate set in BK-MCE algorithm \\
   $D$ &  the candidate set whose elements should not be touched in BK-MCE algorithm\\
   $v_p$ &  the pivot in BK-MCE algorithm\\
   $r$ & the local visibility, refers to Formula~(\ref{eq:r})  \\
   $\overline{l}$ & the upper bound of the size of a maximal clique  \\
   $\underline{r}$ &  the lower bound of local visibility \\
   
   $s(r)$ & the sampling function used in~\cite{wang_redundancy-aware_2013}  \\
   $s_{opt}(r)$ & the conditional optimal sampling function  \\
   $\mathcal{T}$ &  a search subtree \\
   

\bottomrule
\end{tabular}
\label{datasets}
\label{notations}
\vspace{-15pt}
\end{table}

BK-MCE algorithm~\cite{bron_algorithm_1973} (Algorithm~\ref{algo:BK}) is a backtracking approach, which recursively calls  procedure $ProcMCE$ to grow the current partial clique by adding a new node from the candidate set until a maximal clique is found.
Here we denote all the neighbor nodes of node $v$ by $\mathcal{N}(v)$. $C$ is the current partial clique or \emph{configuration}, which is still growing. 
$T$ and $D$ are candidate sets whose elements are common neighbors of $C$, while $D$ only contains nodes which have been contained by some earlier output maximal cliques grown from the current $C$.
Algorithm~\ref{algo:BK} takes  graph $G(V,E)$ as input and outputs all the maximal cliques in $G$. Initially it calls  procedure $ProcMCE(\emptyset,V,\emptyset)$ (line 1). Then $ProcMCE$ will be called recursively (line 10) until $\mathcal{M}(G)$ is generated. 
At every recursive stage $ProcMCE$ will first check whether $T=\emptyset$ and $D=\emptyset$ (line 3). If so, it means that there is no candidate node left, and therefore the current $C$ is output as a maximal clique (line 4). 
If not, generally speaking, it will remove an arbitrary node $v$ from $T$ and add it into $C$.
Then it recursively calls  procedure $ProcMCE(C\cup \{v\},T\cap \mathcal{N}(v),D\cap \mathcal{N}(v))$. 
Here, $T\cap \mathcal{N}(v)$ is the refined $T$ by deleting all nodes which are not neighbors of $v$, and the same for $D\cap \mathcal{N}(v)$. It ensures that every node in $T$ or $D$ is a common neighbor of the current $C$. 
Finally, since $v$ is sure to be contained by some future cliques grown from $C$, $v$ is added into $D$ (lines 8-12). 
Note that a pivot $v_p$ is chosen for avoiding some branches which will generate the same maximal clique (line 6). This is because, from the current configuration, a maximal clique containing $v$ which is a neighbor of $v_p$, can be grown either from $v_p$ or $u$. $u$ is a neighbor of $v$ but not of $v_p$.

\begin{algorithm}[tttt]
        \caption{BK-Maximal Clique Enumeration}
        \begin{algorithmic}[1] 
            \Require Graph $G(V,E)$;
            \Ensure $\mathcal{M}(G)$;
            \State Call $ProcMCE(\emptyset,V,\emptyset)$
            \Statex
\Procedure {$ProcMCE(C,T,D)$}{}

\If {{$T=\emptyset$ and $D=\emptyset$}}
\State Output $C$ as a maximal clique;
\Return
\EndIf
\State Choose a pivot vertex $v_p$ from $(T\cup D)$;
\State $T'\leftarrow T\backslash \mathcal{N}(v_p)$;
\For{each $v\in T'$}
\State Call $ProcMCE(C\cup \{v\},T\cap \mathcal{N}(v),D\cap \mathcal{N}(v))$;
\State $T\leftarrow T\backslash \{v\}$;
\State $D\leftarrow D\cup \{v\}$;
\EndFor
\EndProcedure

        \end{algorithmic}
    \label{algo:BK} 
   
    \end{algorithm}

    \begin{algorithm}[t]
     \caption{Summarization by Sampling}
        \begin{algorithmic}[1] 
            \Require Graph $G(V,E)$, threshold $\tau$;
            \Ensure An expected $\tau$-visible summary $\mathcal{S}$;
\State $\mathcal{S}\leftarrow \emptyset$, $C'\leftarrow \emptyset$;
\State Call $ProcRMCE(\emptyset,V,\emptyset)$. 
\Statex
\Procedure {$ProcRMCE(C,T,D)$}{}

\If {{$T=\emptyset$ and $D=\emptyset$}}
\State include $C$ in $\mathcal{S}$; $C'\leftarrow C$;
\Return
\EndIf
\State Calculate $\overline{l}$ and $\underline{r}$;
\State Keep the branch $\mathcal{T}$ with probability $\sqrt[\overline{l}]{s(\underline{r})}$; 
\If {the branch $\mathcal{T}$ is kept}
\State Choose a pivot vertex $v_p$ from $(T\cup D)$;
\State $T'\leftarrow T\backslash \mathcal{N}(v_p)$;
\For{each $v\in T'$}
\State {\small Call $ProcRMCE(C\cup \{v\},T\cap \mathcal{N}(v),D\cap \mathcal{N}(v))$;}
\State $T\leftarrow T\backslash \{v\}$;
\State $D\leftarrow D\cup \{v\}$;
\EndFor
\EndIf
\EndProcedure

\end{algorithmic}
\label{al2}
\end{algorithm}
\subsection{Summarization by Sampling}
\label{subsec:sample}

Let us first ignore sampling and consider a deterministic enumeration that can find a $\tau$-visible summary: recall that BK-MCE is a depth-first algorithm, it outputs $\mathcal{M}(G)$ in such an order that two maximal cliques share a large portion of common nodes if they are produced next to each other. We denote this property as \emph{locality}. Let $C'$ be the last generated maximal clique which has been added into  summary $\mathcal{S}$, when a new clique $C$ is generated, we can compare it with $C'$, rather than with every clique in $\mathcal{S}$, to compute a \emph{local} visibility $r$ (Formula~(\ref{eq:r})). If $r \geq \tau$, discard $C$; otherwise, keep $C$. Such a deterministic strategy will guarantee to produce a $\tau$-visible summary. 
\begin{equation}\label{eq:r}
	r=\frac{|C\cap C'|}{|C|}
\end{equation}
However, it will be desirable if we can discard a whole search branch with good confidence when we find that the branch has significant overlap with the last found clique $C'$. This leads to the idea of deliberately pruning some recursive sub-procedures with some probability - let us call it sampling. Meanwhile, we must guarantee that the summary should have the expected visibility $\mathbb{E}[\mathcal{V}_{\mathcal{S}}(C)]\geq \tau, \mbox{ }\forall C\in \mathcal{M}(G)$.

Details about invoking a sampling method to give an expected $\tau$-visible summary are shown in Algorithm 2. The key idea is to execute a sampling operation (line 8) to determine whether this current new branch $\mathcal{T}$ should be grown or not before entering a new procedure $ProcMCE(C,T,D)$ (line 13). In line 7, $\overline{l}$ denotes an upper bound of the size of the next maximal clique $C$ and $\underline{r}$ denotes a lower bound of the local visibility $r$. As we have not found $C$, i.e., $l(=|C|)$ and $r$ are unknown, we can only estimate $\overline{l}$ and $\underline{r}$. The sampling probability function $\sqrt[\overline{l}]{s(\underline{r})}$ is designed to be a function of $\overline{l}$ and $\underline{r}$.
The work in~\cite{wang_redundancy-aware_2013} chose the probability function $s()$ to be:
\begin{equation}\label{eq:s_r}
	s(r)=\frac{(1-r)(2-\tau)}{(2-r-\tau)}
\end{equation}
and proved that applying $s(r)$ in Algorithm 2 can produce a summary with the expected visibility $\mathbb{E}[\mathcal{V}_{\mathcal{S}}(C)]\geq \tau, \mbox{ }\forall C\in \mathcal{M}(G)$.
Due to the space limit, we briefly introduce the rationale of $\sqrt[\overline{l}]{s(\underline{r})}$. From Formula~(\ref{eq:s_r}), $s(r)$ is a decreasing function with range [0,1]. This means when we find the estimated $\underline{r}$ becoming larger, the probability of keeping the current search branch becomes smaller. When $\overline{l}$ is estimated larger, this implies that we will call the recursion more times, hence the probability of keeping the current search branch is made larger. 
{Algorithm~\ref{algo:BK} and Algorithm~\ref{al2} follow the clique enumeration paradigm, 
so the time complexities of them are both bounded by $O(3^{|V|/3})$
because a $|V|$-vertex graph has at most $3^{|V|/3}$ maximal cliques\cite{moon_cliques_1965}. 
Algorithm~\ref{al2} should be practically faster due to early prunings, but has the same complexity in the worst case when $\tau=1$.}

\section{A new sampling function}
\label{sec:new}

Expected $\tau$-visible summaries are not unique. Apparently, the more concise (smaller) a summary is, the better the summary is. Three questions arise naturally: 
\begin{itemize}
	\item [(1)]Are there any better sampling strategies?
	\item [(2)]What kind of sampling strategy is optimal?
	\item [(3)]If finding the optimal is difficult, how can we provide the best effort?
\end{itemize}
We will address question (1) in this section and discuss questions (2) and (3) in Section~\ref{sec:opt} and Section~\ref{sec:boundAndorder} respectively.
In Section~\ref{subsec:idea}, we give our idea why we consider there should exist a better sampling function, then we introduce the new sampling function and prove its superiority in Section~\ref{subsec:function}. 

\subsection{Intuition}
\label{subsec:idea}

Our new sampling strategy is based on the following two observations:

\textbf{Observation 1:} In Formula~(\ref{eq:s_r}), when $r\in (\tau, 1)$, we always have $s(r)>0$. This means even if we know the newly 
generated clique is $\tau$-visible with respect to the current $\mathcal{S}$, there is still a positive probability to add it into the summary. Thus $\mathcal{S}$ will be more redundant because of these unnecessary cliques. A better strategy is to set $s(r)=0$ in such cases, which means not to add these cliques at all.

\textbf{Observation 2:} In Formula~(\ref{eq:s_r}), when $r=0$, we have $s(r)=1$. This means once we find a maximal clique whose nodes are totally new to the current summary, we add it into $\mathcal{S}$ without hesitation. This seems reasonable, however, there is still some possibility for this brand new clique to be covered by some future cliques. 
Moreover, considering we are looking for an \emph{expected $\tau$-visible} summary, which means that we have the option not to include a brand new clique as long as the final summary is expected $\tau$-visible. In other words, it is safe to add the brand new clique with certain probabilities. Cases where $r$ is in $(0,\tau)$ are similar. 

%
\subsection{Sampling Function $s_{opt}(r)$}
\label{subsec:function}
Following the observations in Section~\ref{subsec:idea}, we give a new sampling function $s_{opt}(r)$ in Formula~(\ref{eq:s_opt}). 
\begin{equation}\label{eq:s_opt}
    s_{opt}(r)=
 \begin{cases}
    \mbox{ }\frac{\tau-r}{1-r}   &  \text{, if $r\in [0,\tau)$.} \\
    \mbox{ }0        &  \text{, if $r\in[\tau,1].$}
 \end{cases} 
\end{equation}
The sampling function $s_{opt}(r)$ implies: if $r\in[\tau,1]$, discard the current search branch; otherwise, keep the current search branch with probability $\frac{\tau-r}{1-r}$. The rationale of setting $\frac{\tau-r}{1-r}$ will be shown in Theorem~\ref{theo:correct}. 



Next, we prove that compared with $s(r)$, $s_{opt}(r)$ is a better function. This means that we need to prove: (1) $s_{opt}(r)$ samples with a lower probability (in Theorem~\ref{theo:lower}); and (2) $s_{opt}(r)$ can produce an expected $\tau$-visible summary (in Theorem~\ref{theo:correct}).


\begin{theorem}\label{theo:lower}
$s_{opt}(r)$ samples with a low probability than s(r), i.e.
	\begin{equation}
		s_{opt}(r)\leq s(r),\mbox{ } \forall r\in [0,1]
	\end{equation}
The equation holds iff $r=1$.
\end{theorem}
\begin{proof}

We show that this inequality holds when $r\in [0,\tau)$ and $r\in[\tau,1]$ separately:
\begin{itemize}
	\item [-] if $r\in[\tau,1]$, we have $s_{opt}(r)=0$ and $s(r)\geq 0$, it is clear that this inequality is satisfied, and the equation holds only when $r=1$ (corresponding to $s(r)=0$).
	\item [-] if $r\in [0,\tau)$, since $\tau\in [0,1]$ and $r\neq 1$, we have
\end{itemize}
\begin{equation}\label{eq:7}
\begin{aligned}
s(r)-s_{opt}(r)&=\frac{(1-r)(2-\tau)}{(2-r-\tau)}-\frac{\tau-r}{1-r}\\
&=\frac{(1-r)^2(2-\tau)-(2-r-\tau)(\tau-r)}{(2-r-\tau)(1-r)}\\
&=\frac{(1-\tau)((1-r)^2+(1-\tau))}{(2-r-\tau)(1-r)}\\
&>0
\end{aligned}
\end{equation}

Combining these two cases, we complete this proof.  
\end{proof}

 {While Theorem~\ref{theo:lower}} promises us a more concise summary $\mathcal{S}$, we {have} to prove that this $\mathcal{S}$ 
is indeed expected $\tau$-visible:

\begin{theorem}\label{theo:correct}
	Algorithm 2, with sampling function $s_{opt}(r)$, can produce an expected $\tau$-visible summary .
\end{theorem}
\begin{proof}\label{pr2}
	First, we give the probability for a maximal clique $C$ being added into $\mathcal{S}$. Then we calculate the expected visibility $\mathbb{E}[\mathcal{V}_{\mathcal{S}}(C)]$ and show it is no less than $\tau$.
	
	Recall that every time before Algorithm 2 starts a new search subtree $\mathcal{T}_i$, line 7 will compute a new pair of $\overline{l}$ and $\underline{r}$. We denote them by $\overline{l_i}$ and $\underline{r_i}$, where $1\leq i\leq k$ and $k=|C|$ is the size of this maximal clique to be grown. Since every $\overline{l_i}$ is an upper bound of $k$ and every $\underline{r_i}$ is a lower bound of $r$, together with the monotonicity of $s_{opt}(r)$, we have
\begin{equation}\label{eq8}
\begin{aligned}
Pr[C\in \mathcal{S}]&=\mathop{\Pi}\limits_{1\leq i\leq k}Pr[\mathcal{T}_i \mbox{ }is\mbox{ }kept]\\
&=\mathop{\Pi}\limits_{1\leq i\leq k}\sqrt[\overline{l_i}]{s_{opt}(\underline{r_i})}\\
&\geq \mathop{\Pi}\limits_{1\leq i\leq k}\sqrt[k]{s_{opt}(r)}\\
&=s_{opt}(r)
\end{aligned}
\end{equation}

If $C$ is not included in $\mathcal{S}$, its visibility $\mathcal{V}_{\mathcal{S}}(C)$ should be no less than the local visibility $r$; if $C$ is included in $\mathcal{S}$,  $\mathcal{V}_{\mathcal{S}}(C)=1$. Now we can calculate the expectation of $\mathcal{V}_{\mathcal{S}}(C)$:

\begin{equation}\label{eq9}
\begin{aligned}
\mathbb{E}[\mathcal{V}_{\mathcal{S}}(C)]&\geq1\cdot Pr[C\in \mathcal{S}]+r\cdot Pr[C\notin \mathcal{S}] \\
&\geq s_{opt}(r)+r\cdot (1-s_{opt}(r))
\end{aligned}
\end{equation}

We show two cases where $r\in [0,\tau)$ and $r\in[\tau,1]$ separately:
\begin{itemize}
\item [-] if $r\in [\tau,1]$, $\mathbb{E}[\mathcal{V}_{\mathcal{S}}(C)]\geq 0+r\cdot (1-0)=r\geq \tau$.
\item [-] if $r\in [0,\tau)$, $\mathbb{E}[\mathcal{V}_{\mathcal{S}}(C)]\geq \frac{\tau-r}{1-r}+r\cdot (1-\frac{\tau-r}{1-r})=\tau$.

\end{itemize}

Combining these two cases, we complete this proof. 
\end{proof}

{\bf{Summary: }} Theorem 1 and Theorem 2 jointly show that $s_{opt}(r)$ is a valid sampling function and is better than $s(r)$. 

\section{Optimality}
\label{sec:opt}

In this section, for the purpose of analyzing the optimality of the sampling function, we show what kinds of conditions should be satisfied. We prove the optimality of $s_{opt}(r)$ under such conditions and further explain why the performance of $s_{opt}(r)$ is good even without the conditions being  fully satisfied.

\subsection{Conditions for Optimality Analysis}
\label{subsec41}
Since we can find a better sampling function $s_{opt}(r)$, another question comes out naturally: with the restriction of expected $\tau$-visibility, does an optimal sampling function (even better than $s_{opt}(r)$) with the smallest probability exist? 

In the proof of Theorem 2, we can only prove the expectation $\mathbb{E}[\mathcal{V}_{\mathcal{S}}(C)]\geq \tau$. 
Intuitively, the smaller the sampling probability is, the smaller the expectation is. 
It is hard to determine whether a sampling function is optimal because of lacking information on how loose the inequality is. If we intend to analyze the optimality of any function, we need to tighten this inequality to be an equation first.
Now we show  under what conditions  the theoretical analysis of optimality can be feasible.

In the proof of Theorem~\ref{theo:correct}, we amplify $\mathbb{E}[\mathcal{V}_{\mathcal{S}}(C)]$ two times:

The \textbf{first inequality sign} of Formula~(\ref{eq9}) is derived from Formula~(\ref{eq8}). Algorithm 2 implements the sampling operation in each recursive procedure
	using the probability $\sqrt[\overline{l_i}]{s_{opt}(\underline{r_i})}$. Since  $\overline{l_i}$ and $\underline{r_i}$ are upper bound and lower bound of $k$ and $r$ respectively, we have $\sqrt[\overline{l_i}]{s_{opt}(\underline{r_i})}\geq \sqrt[k]{s_{opt}(r)}$, thus $Pr[C\in \mathcal{S}]\geq s_{opt}(r)$. Now we have two approaches to eliminate this inequality. 
		 \emph{The first approach} is that if we want to analyze the property of the sampling function itself, we need to set the other factors ideal. This means if we do not care about the details of how to calculate $\overline{l_i}$, we can assume that this upper bound is ideal, so we have $\overline{l_i}=l$, and the same for $\underline{r_i}$. Note this hypothesis is made only for the purpose of analyzing the function theoretically, not for implementing Algorithm 2 in practice. With this assumption, we have $Pr[C\in \mathcal{S}]= s_{opt}(r)$. 
		 \emph{The second approach} is to modify the sampling procedure. Now let us sample using the probability $s_{opt}(r)$ only after a complete maximal clique is generated, rather than sample each time a new node is grown. If so, it is obvious that $Pr[C\in \mathcal{S}]= s_{opt}(r)$. One may argue that it is meaningless to do sampling once a maximal clique is found. It is true that if we add every clique whose $r$ is no more than $\tau$ into summary, this summary is strictly $\tau$-visible. However, as we explained in Section~\ref{sec:new}, this would introduce more redundancy to $\mathcal{S}$. In some applications, we only need this summary to be expected $\tau$-visible, so this one-step sampling procedure is significant to give such a concise $\mathcal{S}$.
	
	 The \textbf{second inequality sign} of Formula~(\ref{eq9}) is from the definition of $r$. 
	Due to the locality of Algorithm~\ref{al2}, 
	we use $r$ to replace the real visibility which should be no less than $r$. Now we need to make the assumption that such locality is sufficiently strong (by which we mean that two similar cliques should be produced consecutively), so that $r$ is indeed the visibility defined in Formula~(\ref{3}). In practice, we do not need to enforce such strong locality to implement Algorithm~\ref{al2}. We introduce this hypothesis only for the theoretical consideration, which means, we only need this assumption to construct a framework under which we can analyze the optimality of sampling functions. Once this hypothesis is made, the second inequality becomes an equation.

Now we can modify Formula (\ref{eq9}) to be an equation:
\begin{equation}
	\mathbb{E}[\mathcal{V}_{\mathcal{S}}(C)]= s_{opt}(r)+r\cdot (1-s_{opt}(r))
\end{equation}
if the following two conditions are satisfied:
\begin{itemize}
	\item [-] The bounds of $l$ and $r$ are ideal, or we only do sampling each time when a full maximal clique is generated. 
	\item [-] The property of locality is strong for $r$ to be the real visibility. 
\end{itemize}


\subsection{Optimality of $s_{opt}(r)$}


{ 
One may expect that 
the optimality should be defined as minimizing the expected cardinality of the summary. 
However, we  notice  that such a strong version of optimality requires the sampling function $s(r)$ to be a function of the distribution of $r$: $\rho (r)$, rather than simply a function of $r$. 
Accordingly, we have to report that it is not easy to give a proper definition for $\rho(r)$. 
This is due to such a  fact: 
{$\rho (r)$ 
could not be known
until the summary has been found.}
This fact holds because  the $r$ value of maximal clique $C$ depends on which maximal clique $C'$ it is compared with:  
once $C'$ changes, the value of $r$ will change accordingly, and  distribution $\rho (r)$ will also be disturbed thereafter. 
Since 
a nondeterministic (sampling) algorithm 
 can hardly know the exact predecessor $C'$ of each maximal clique unless  the algorithm is finalized, 
 it may be impossible to define $\rho(r)$ unless the summary is fully determined. 
{Thus $\rho (r)$  not only is data-dependent, but also  relies on the sampling function $s(r)$. 
Since  $s(r)$ should  inversely rely on $\rho (r)$, 
it is likely to be hard  to properly  give $\rho (r)$ an independent definition.}}

{For the above reasons, we choose not to seek a distribution related  sampling function with the strong version of optimality, but rather define the {optimality} as: \textbf{given the $r$ value of maximal clique $C$,} \textbf{sampling $C$ with the lowest probability while still promising $\tau$-visibility of the summary}. 
We consider this definition being more manipulatable theoretically and applicable practically. 
And this optimality successfully shows its effectiveness in the experimental studies (see Section~\ref{sec:exp}).} 
Now we can analyze the optimality of $s_{opt}(r)$ in the framework introduced in {Section} \ref{subsec41}. 

\begin{theorem}
	If the two conditions in Section \ref{subsec41} are satisfied, $s_{opt}(r)$ is optimal for Algorithm~\ref{al2}.
\label{theo:3}
\end{theorem}

\begin{proof}
	We show that if there exists a sampling function $s'(r)$, such that $\forall r\in [0,1]$, $s'(r)\leq s_{opt}(r)$ and for at least one point $r_0$ there is $s'(r_0)< s_{opt}(r_0)$, such a function cannot be used to generate an expected $\tau$-visible summary. 
	
	Note $s_{opt}(r)=0$ when $r\in [\tau,1]$, $r_0$ cannot be in this range, since a valid probability should be nonnegative. So we have $r_0\in [0,\tau)$, and
	\begin{equation}
	\begin{aligned}
				\mathbb{E}[\mathcal{V}_{\mathcal{S}}(C)]|_{r_0}&= s'(r_0)+r_0\cdot (1-s'(r_0))\\
				&<s_{opt}(r_0)+r_0\cdot (1-s_{opt}(r_0))\\
				&=\frac{\tau-r_0}{1-r_0}+r_0\cdot (1-\frac{\tau-r_0}{1-r_0})\\
				&=\tau
	\end{aligned}
	\end{equation}
	This means that the summary generated by $s'(r)$ cannot be expected $\tau$-visible. 
\end{proof}

Theorem~\ref{theo:3} is a conditional theoretical guarantee for the optimality of $s_{opt}(r)$.
 Note even if in general cases  these two strong conditions are not fully satisfied, Theorem 3 is still useful for the practical implementation of Algorithm~\ref{al2}. 
That means if the bounds $\overline{l}$ and $\underline{r}$ are well estimated and the property of locality is strong, the inequalities in Formula~(\ref{eq9}) can be very tight. In such cases, $s_{opt}(r)$ can still show good performance.

{One may concern that it is not clear to what extent we can achieve the good performance of $s_{opt}(r)$ in practice by 
 (1) tightening bounds;
 (2) strengthening the locality. 
 In the next section, we address the first concern by reviewing two existing bounds and proposing a new one which outperforms the other two by large margins. 
 For the second concern, we show that 
 stronger locality can be achieved by reordering  vertices carefully. We will review an existing vertex order and design a new better one.}

\section{Bounds and Locality}
\label{sec:boundAndorder}

In this section, we show how to approach good performance of the new sampling strategy by tightening bounds  (Section~\ref{BA}) and by reordering vertices  (Section~\ref{LA}). 

\subsection{Bound Analysis}\label{BA}
{
The first inequality of (\ref{eq9}) is derived from bound estimation. 
Note that Formula~(\ref{eq:14})  we use to calculate the lower bound $\underline{r}$ is the same as that introduced in~\cite{wang_redundancy-aware_2013}:
\begin{equation}\label{eq:14}
	\underline{r}=\min\limits_{1\leq t\leq \overline{d}}\frac{|C\cap C'|+\max \{t-\overline{y_t},0\}}{|C|+t}
\end{equation}
where $C'$ is the previous maximal clique added into $\mathcal{S}$; 
$t$ is the number of vertices to be used for growing the partial configuration $C$ into a full maximal clique; 
$\overline{d}$, which satisfies $\overline{l}=|C|+\overline{d}$, is the upper bound of $t$;
and given $y_t$ out of the $t$ vertices are not covered by $C'$, $\overline{y_t}$ is an upper bound of $y_t$.  
Formula~(\ref{eq:14}) can be understood in this way: 
suppose we know that the current partial configuration $C$ still needs $t$ vertices to grow into a full maximal clique $P$, then the dominator $|C|+t$ is the size of $P$. 
Since $\overline{y_t}$ means that at most $\overline{y_t}$ out of $t$ vertices in $P\backslash C$ are not contained by $C'$, 
this  means that at least $t-\overline{y_t}$ are covered by $C'$. 
Thus $\max \{t-\overline{y_t},0\}$ is  a lower bound of the size of $(P\backslash C)\cap C'$. (The \emph{max} operator is inserted here because depending on the estimation method of $\overline{y_t}$, $t-\overline{y_t}$ may be negative.) Now we see that the two parts of the numerator are  $|C\cap C'|$ and a lower bound of $|(P\backslash C)\cap C'|$ respectively, thus the sum is a lower bound of $|P\cap C'|$. Combining the discussions above, the whole fraction is the very lower bound of $r\equiv |P\cap C'|/ |P|$. Since we lack information of the exact value of $t$, we have to enumerate all possible $t$ in $[0, \overline{d}]$ and choose the minimum  as the lower bound.  
$\overline{y_t}$ can be estimated  as   $|T\backslash C'|$, or simply the value of $t$, or the number of vertices in $T\backslash C'$ whose degrees are at least $t-1$
(because these $y_t$ vertices should be contained by a $t$-clique). 
 We see here the upper bound $\overline{d}$ 
( $=\overline{l}- |C|$, where $|C|$ is known)
 is used to estimate $\underline{r}$, 
 and the fraction after the $\min_{1\leq t\leq \overline{d}}$ operator has nothing related to $\overline{d}$ (because it is calculated after $t$ is given), 
 so the quality of $\underline{r}$ is  determined by the tightness of $\overline{d}$ (or $\overline{l}$). Thus in the following, we focus on estimating $\overline{d}$.

One valid and tight bound of $d$ is the size of the maximum clique in candidate set $T$, however, finding such a maximum clique itself is a clique enumeration problem which is of exponential time.
As a result, we should consider  realistic bounds instead. 
In the following, we review two bounds that were  discussed in the previous work~\cite{wang_redundancy-aware_2013},  
then we propose a new one to further improve the effectiveness of $s_{opt}(r)$. 
{ 

%

 Let $G_T$ be the induced graph of the candidate set $T$ on graph $G$, then two existing upper bounds of $d$ are: 
\begin{itemize} 
%
	\item [-]$H$ bound, denoted by $\overline{d}_h$, is the maximum $h$ so that there exist at least $h$ vertices in $G_T$ whose degrees are no less than $h-1$.
	The maximum clique size can be bounded by $h$ because if there exists a $k$-clique, there should also exist at least $k$ vertices in $G_T$ whose degrees are no less than $k-1$. Therefore $h\geq k$ holds for all possible $k$-cliques, including the maximum clique. 
	\item [-]Core bound, denoted by $\overline{d}_{core}= Core(G_T)+ 1$, where $ Core(G_T)$ denotes the maximum core number in $G_T$. 
	We now review the definition of  $k$-core~\cite{seidman1983network} and core number first. 
\end{itemize}
\begin{definition}[$k$-core]
	The $k$-core of a graph $G$ is the largest induced subgraph in which the degree of each vertex is at least $k$.
\end{definition}

\begin{definition}[Core Number]
The core number of  graph $G$, denoted as $Core(G)$, is the largest $k$ such that a $k$-core is contained in $G$.
\end{definition} 
The core number can serve as an upper bound because $k$-core is weaker than $k$-clique: a $k$-clique must be a $(k$-$1)$-core, while a ($k$-$1$)-core may not be a $k$-clique. Thus $Core(G_T)+ 1$ will be no less than the maximum clique size in $G_T$. 
Now we define our newly proposed bound.    
\begin{definition}[Truss bound]
	Truss bound,  denoted by $\overline{d}_{truss}$,  is the maximum truss number $ Truss(G_T)$ in $G_T$.
\end{definition}
Now we review the definition of $k$-truss and truss number, and then  explain  why the maximum truss number $Truss(G_T)$ is valid to be an upper bound.
\begin{definition}[$k$-truss]
	The $k$-truss of a graph $G$ is the largest induced subgraph in which each edge must be part of $k-2$ triangles in this subgraph. 
\end{definition}

\begin{definition}[Truss Number]
	The truss number of  graph $G$, denoted as $Truss(G)$, is the largest $k$ such that a $k$-truss is contained in $G$.
\end{definition}
$\overline{d}_{truss}$ is a  upper bound of the size of maximum clique. 
This is  because a $k$-clique with the maximum $k$ is also a $k$-truss since each edge in a $k$-clique is strictly contained by $k-2$ triangles. 
Thus the truss number cannot be less than the maximum clique size $k$. 

These three bounds satisfy the following inequality:
\begin{equation}
    \overline{d}_h \geq \overline{d}_{core} \geq \overline{d}_{truss}
\end{equation}
The first inequality holds because $H$ bound 
does not enforce the $h$ vertices to be connected, while core bound does.
The second inequality  comes from the fact that
a $k$-truss must be a $(k-1)$-core. 
This is because the endpoints of each edge $e$ should be incident to no less than $k-1$ edges (including $e$ itself) since $e$ is guaranteed to be involved in at least $k-2$ triangles. 


The cost of evaluating these bounds are:
\begin{equation}\label{eq:16}
		 \overline{d}_h: O(V_T);\mbox{   }\mbox{   }
		 \overline{d}_{core}: O(E_T);\mbox{   }\mbox{   }
		\overline{d}_{truss}: O(E_{T}^{1.5})
	\end{equation}
where $V_T$ and $E_T$ are  vertex set and edge set of the induced graph $G_T$ respectively. 
The induced graph $G_T$ can be constructed when selecting the pivot thus its construction does not incur an extra cost. 
 For $\overline{d}_h$, when constructing $G_T$, we can maintain a   $V_T$-length array to record 
the number of vertices at
 each degree value. 
 This can be done in $O(V_T)$. 
 Then the $H$-value can be found by scanning this array from tail  (where the numbers of vertices with higher degree values are stored) to head  (where the numbers of vertices with lower degree values are stored) until $h$ vertices whose degrees are no less than $h-1$ are found.  
 This step is also in $O(V_T)$.  
 For ${d}_{core}$, an $O(E_T)$ core decomposition~\cite{Khaouid:2015:KDL:2850469.2850471} is needed after $G_T$ is found. 
  For ${d}_{truss}$, the truss decomposition takes $O(E_{T}^{1.5})$ to find the maximum truss number~\cite{wang2012truss}. 
 }  

We see that the truss bound is the tightest one among all of the three, and therefore it  promises the best performance in terms of  effectiveness. 
The intrinsic is the fact that the structure of truss is more compact (or cohesive) than the other two. 
(This property of compactness can also be used to design vertex orders to enhance the locality. 
We will give a detailed discussion soon in Section~\ref{LA}.) 
Users may have their own preferences to balance the running time and summary size. 
Thus which bound to select depends on to what extent the effectiveness can be improved by sacrificing the efficiency. 
In Section~\ref{sec:exp}, we conduct experimental studies to compare the practical performance of different bounds in terms of both effectiveness and efficiency.

\subsection{Locality Analysis}\label{LA}

Strong locality implies that two similar cliques should be produced consecutively. 
This means, for a new clique $C$, the local visibility computed with the previous output clique $C'$ should be close to the global visibility which is computed with the most similar clique to $C$ in the summary. 
However, such a condition is difficult to meet. 
Reflected in practice, one typical implementation is the vertex order we should follow to grow the current partial clique.
An effective vertex order with strong locality should have such a property that each candidate set $T$ of the current configuration $C$ has a {sufficiently compact} structure. 
Here, by {\it compact (or cohesive)} we mean that the nodes of a candidate set  are well-connected with each other such that  cliques in this set have a higher probability to overlap. 

One question arises: in the outer recursion level of BK-MCE,  since the neighbor set $\mathcal{N}(v)$ of the only vertex $v$ in the current partial clique $C$=$\{v\}$ is uniquely determined by the graph $G(V,E)$, why we still expect a particular structure in the candidate set of $\{v\}$?
The answer is that if we implement a fixed vertex order to grow cliques, when we include $v$ into $C$, all the neighbors of $v$ which precede it in the order can be safely moved into set $D$. The key point is that the difference between $\mathcal{N}(v)$ and the candidate set of $\{v\}$ is determined by the particular order we choose, thus leaves us the very opportunity to reshape the structure of the candidate set. The same holds for each level of the recursion. 


{
Now we see that 
 strong locality can be achieved by reordering vertices. 
 In the following, we explain  why degeneracy order can be employed to achieve this goal even if the initial purpose of it is to bound time complexity of the BK-MCE~\cite{eppstein_listing_2013}. 
 Then we propose a novel truss order based on truss decomposition to further enhance locality.  
  Now we begin with the definition of degeneracy. 
\begin{definition}[Degeneracy]
	Given a graph $G(V,E)$, the degeneracy of $G$ is the smallest value $d$, such that every subgraph of $G$ contains a vertex whose degree is no more than $d$. 
\end{definition}
Degeneracy is naturally related to a special vertex order below. 
\begin{definition}[Degeneracy Order]\label{degeneracy}
	The vertices of a d-degeneracy graph have a degeneracy order, in which each vertex $v$ has only $d$ or fewer neighbors after itself.
\end{definition}
}
\begin{table*}
\renewcommand\arraystretch{1.2}
\small
\centering
\caption{Statistics of datasets}
\begin{tabular}{l l l l c c}
\toprule
Name  & $\lvert V\rvert$ & $\lvert E\rvert$    & Cliques & $\tau$-RMCE-TU [$\tau=0.5/0.9$] & $\tau$-R$^+$MCE-TU [$\tau=0.5/0.9$]   \\
\midrule
soc-Epinions1   & 75,879  & 508,837    & 1,775,065 &18.1\% / 77.9\% &2.5\% / 31.3\% \\
loc-Gowalla    & 196,591  & 950,327  &960,916  &33.9\% / 85.3\% &3.9\% / 30.3\%\\
amazon0302     & 262,111  & 1,234,877   &403,360   &68.8\% / 95.6\% &15.8\% / 37.8\%\\
email-EuAll   & 265,214  & 420,045  & 377,750  &71.3\% / 93.6\% &3.7\% / 14.0\%\\
NotreDame   & 325,729  & 1,497,134  & 495,947   &69.2\% / 93.7\% &5.1\% / 17.0\%\\
com-youtube  & 1,134,890  & 2,987,624   & 3,265,951  &62.8\% / 93.8\% &6.0\% / 23.5\%\\
soc-pokec  & 1,632,803  & 30,622,564   & 19,376,873  &61.1\% / 93.3\% &6.0\% / 27.3\%\\ 
cit-Patents   & 3,774,768  & 16,518,948   & 14,787,031  & 86.4\% / 96.4\%  &7.6\% / 20.2\% \\

\bottomrule
\end{tabular}
\label{datasets}
\vspace{-10pt}
\end{table*}
Degeneracy order can be formed by repeatedly deleting the minimum degree vertex with  all its edges  on the current subgraph. 
Note this actually is the core decomposition procedure~\cite{Khaouid:2015:KDL:2850469.2850471}, 
thus this order sorts vertices by core number from low to high. 

The reason why this order can be used to enhance  locality is straightforward. 
We explain it by focusing on this particular scene during MCE procedure that a vertex $v$ is being moved from candidate set $T$ to partial clique $C$.  
This $v$ and all vertices of $G$ that are reordered after $v$ in the degeneracy order induce a subgraph $G'$. 
By the construction of degeneracy order, we know $v$ is the minimum degree vertex in $G'$, which is denoted by $d(v)$,  
thus $G'$ is a $d(v)$-core. 
Since the candidate set $T$ is a subset of $G'$, we reach the conclusion that $T$ is contained by a $d(v)$-core, which is our desirable compact structure with strong locality. 
Although existing works~\cite{eppstein_listing_2013}~\cite{san_segundo_efficiently_2018} studied using degeneracy to speed up MCE in the aspect of running time, 
to our best knowledge, 
our work is the first to exploit degeneracy order to strengthen the locality for the purpose of reducing overlapping cliques. 
To further enhance the locality,  
we notice that the key of  locality is to guarantee the candidate set $T$ to be contained by a compact structure, e.g., $k$-core. 
Hence if we can find a novel vertex order that has a stronger guarantee, 
e.g., $T$ is covered by a $k$-truss, 
then we can foresee that the performance in terms of effectiveness will outperform that of degeneracy order. 
Following this intuition, we carefully inspect the relationship between core decomposition and degeneracy order, and find that such a relationship also applies to the truss decomposition and a new vertex order (truss order). 
\begin{definition}[Truss Order]\label{df:to}
Vertices sorted by truss order satisfy such a property: 
if $k$ is the maximum value that there exists a $k$-truss containing vertex $v$, then all the  vertices reordered after $v$ should also be contained by the same $k$-truss. 
\end{definition}

Truss order can be formed during the procedure of truss decomposition.  
We firstly delete the edge $(u,v)$ which is contained by the least number of triangles (this number is denoted by the  \textit{support} of an edge). 
After $(u,v)$ is removed, the supports of all edges whose endpoints contain $u$ or $v$ decrease by 1. 
The procedure repeats until all the edges are removed. 
Then the order that vertices are peeled off from $G$ is a valid truss order. 
This is because the order sorts each vertex by the maximum value of $k$ that there exists a $k$-truss containing it. 
The same analysis  why degeneracy order enhances locality applies to truss order: by Definition~\ref{df:to}, the candidate set $T$ is guaranteed to be contained by a $k$-truss. 

Since what we desire is a compact structure of candidate set, $k$-truss is apparently more favorable than $k$-core. 
{Note that the concept of locality is  more goal-driven and the extent of locality is output-determined. 
Hence, instead of giving a formal theoretical analysis which we found difficult,  we decide to use  extensive experiments to illustrate the effect of three types of vertex orders on output size.}
We will report experimental results in Section~\ref{sec:exp} to compare the performance of these two vertex orders with random order as a baseline in terms of both effectiveness and efficiency. 



\section{Experimental Evaluation}
\label{sec:exp}

\begin{table}[b]
\renewcommand\arraystretch{1.2}
\small
\centering
\caption{Notations}
\begin{tabular}{p{30pt}p{180pt}}
\toprule
Setting  & Meaning \\
\midrule
  $T,H,C$  &  Truss bound, H bound, Core bound \\
   $U,I,R$ & Truss order, Degeneracy order, Random order  \\

\bottomrule
\end{tabular}
\label{table3}
\end{table}

\begin{figure*}[ht]
\vspace{-12pt}
	\centering

	\subfloat[soc-Epinions1 \label{1a}]{\includegraphics[width=4.2cm, height=3.36cm]{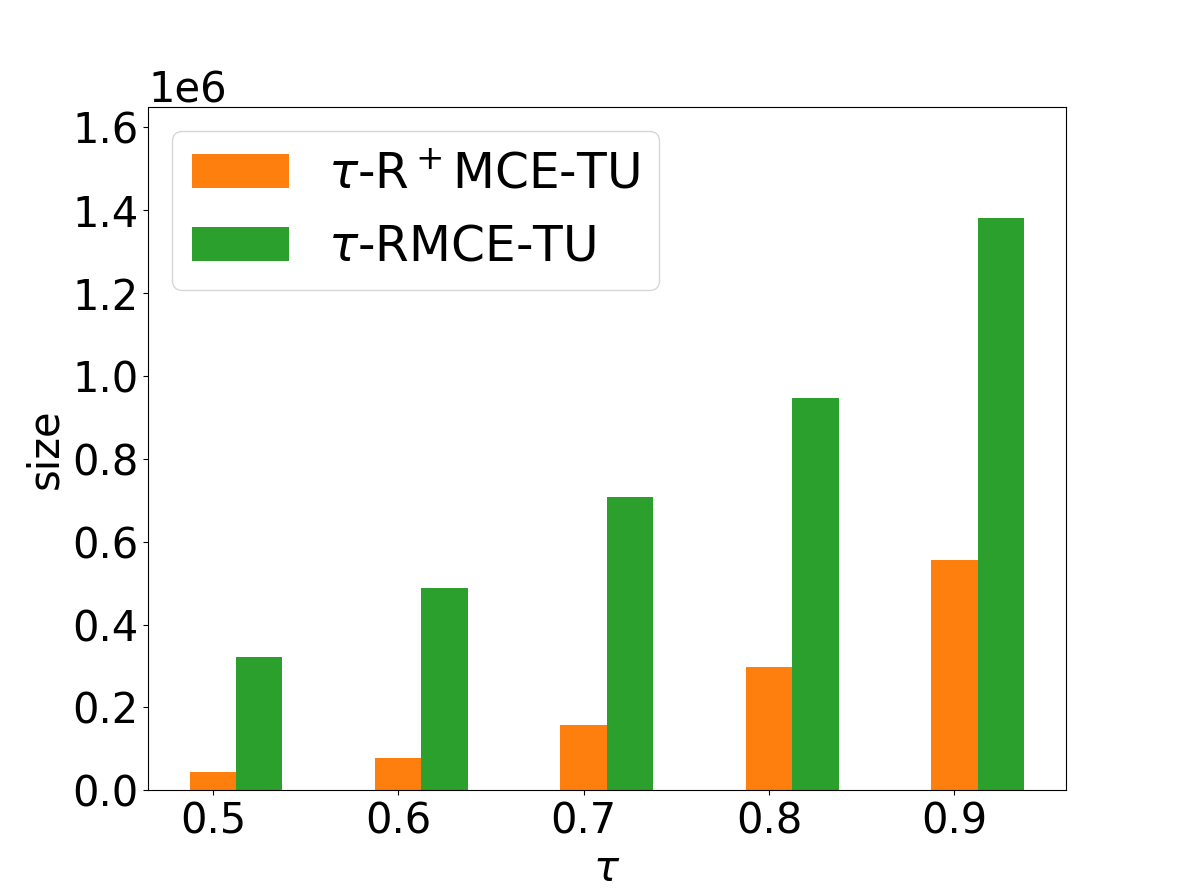}}
	\subfloat[loc-Gowalla \label{1b}]{\includegraphics[width=4.2cm, height=3.36cm]{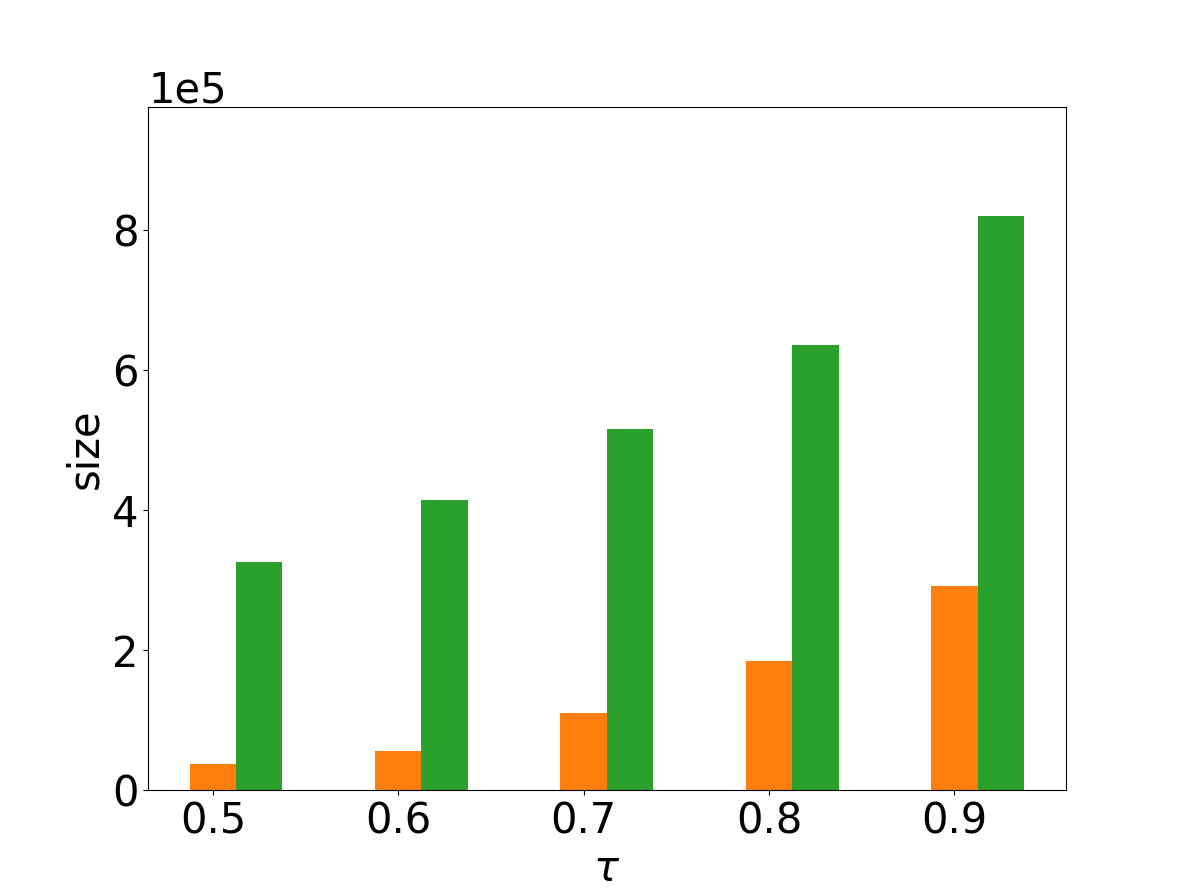}}
\vspace{-12pt}
	\subfloat[amazon0302 \label{1c}]{\includegraphics[width=4.2cm, height=3.36cm]{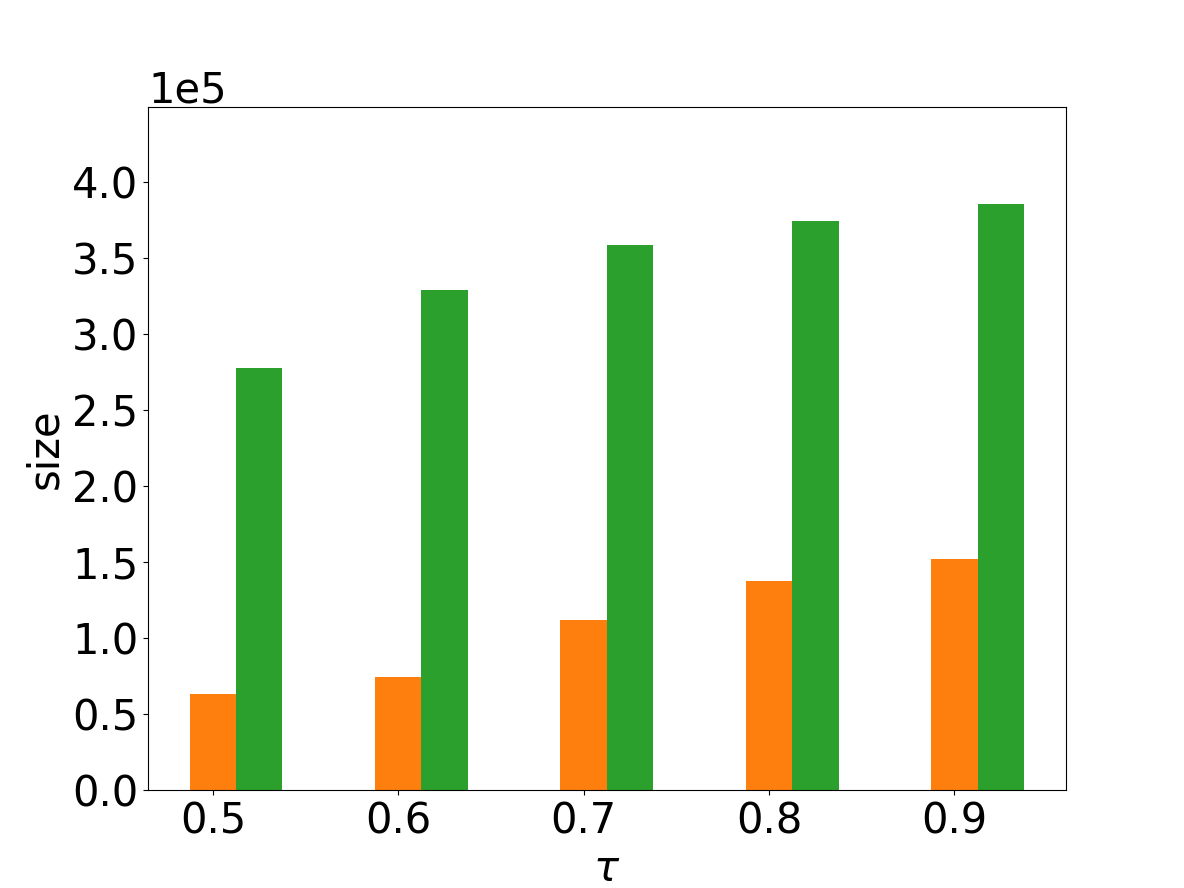}}
\subfloat[email-EuAll  \label{1d}]{\includegraphics[width=4.2cm, height=3.36cm]{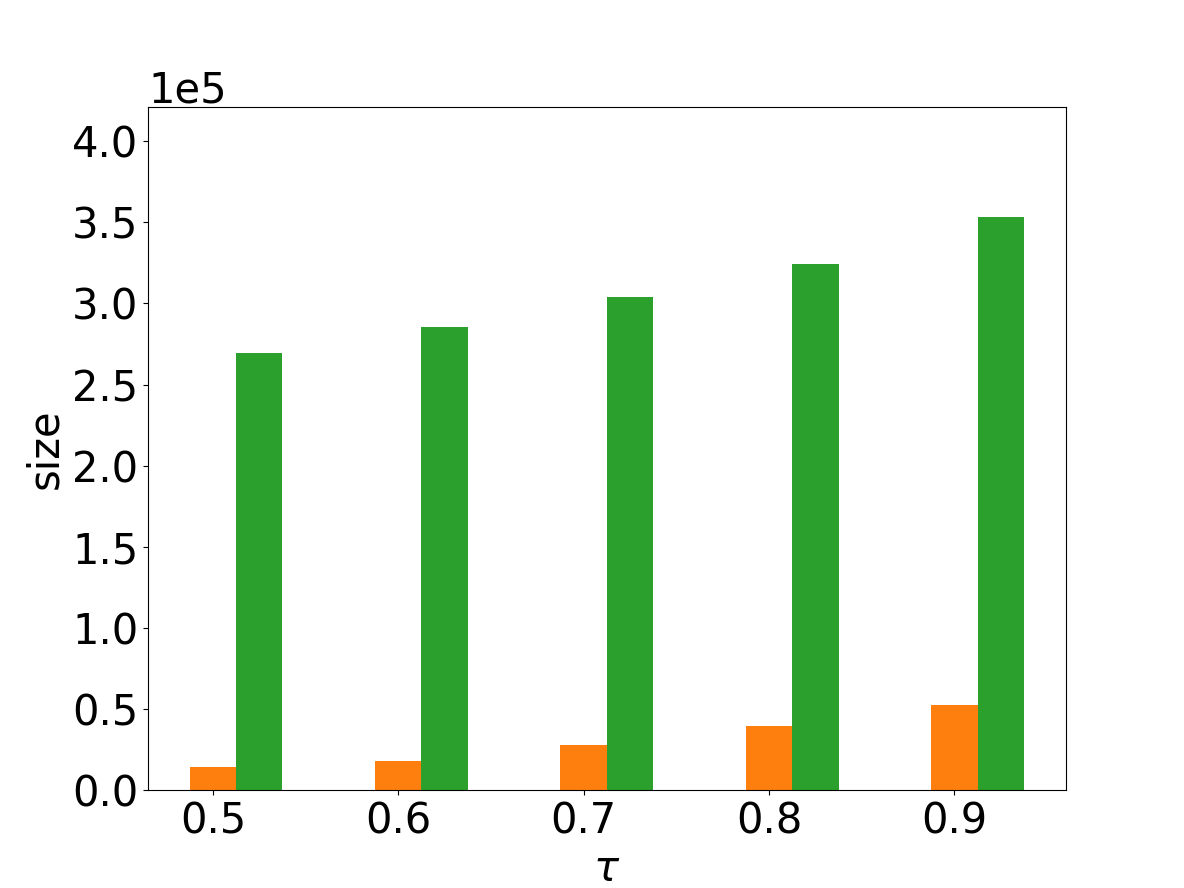}}

\subfloat[web-NotreDame \label{1f}]{\includegraphics[width=4.2cm, height=3.36cm]{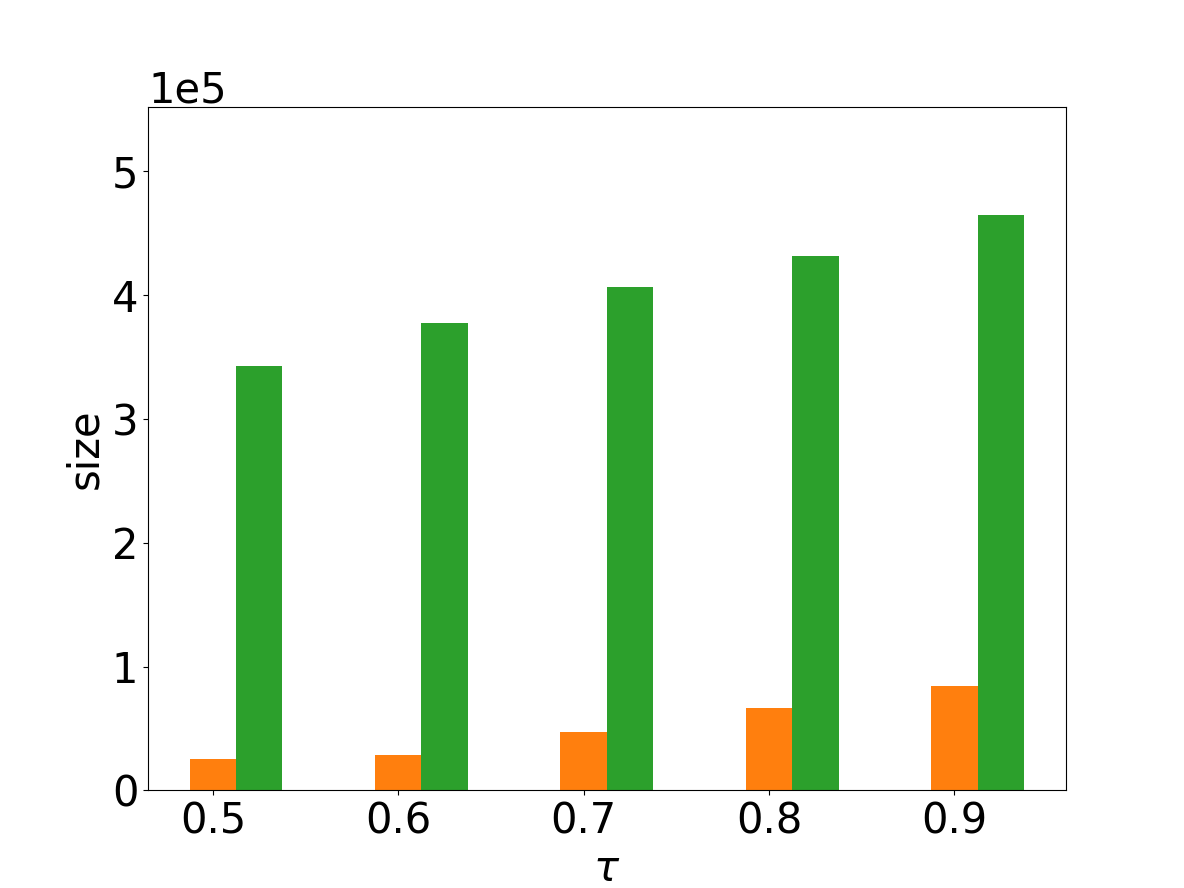}}
\subfloat[com-youtube \label{1g}]{\includegraphics[width=4.2cm, height=3.36cm]{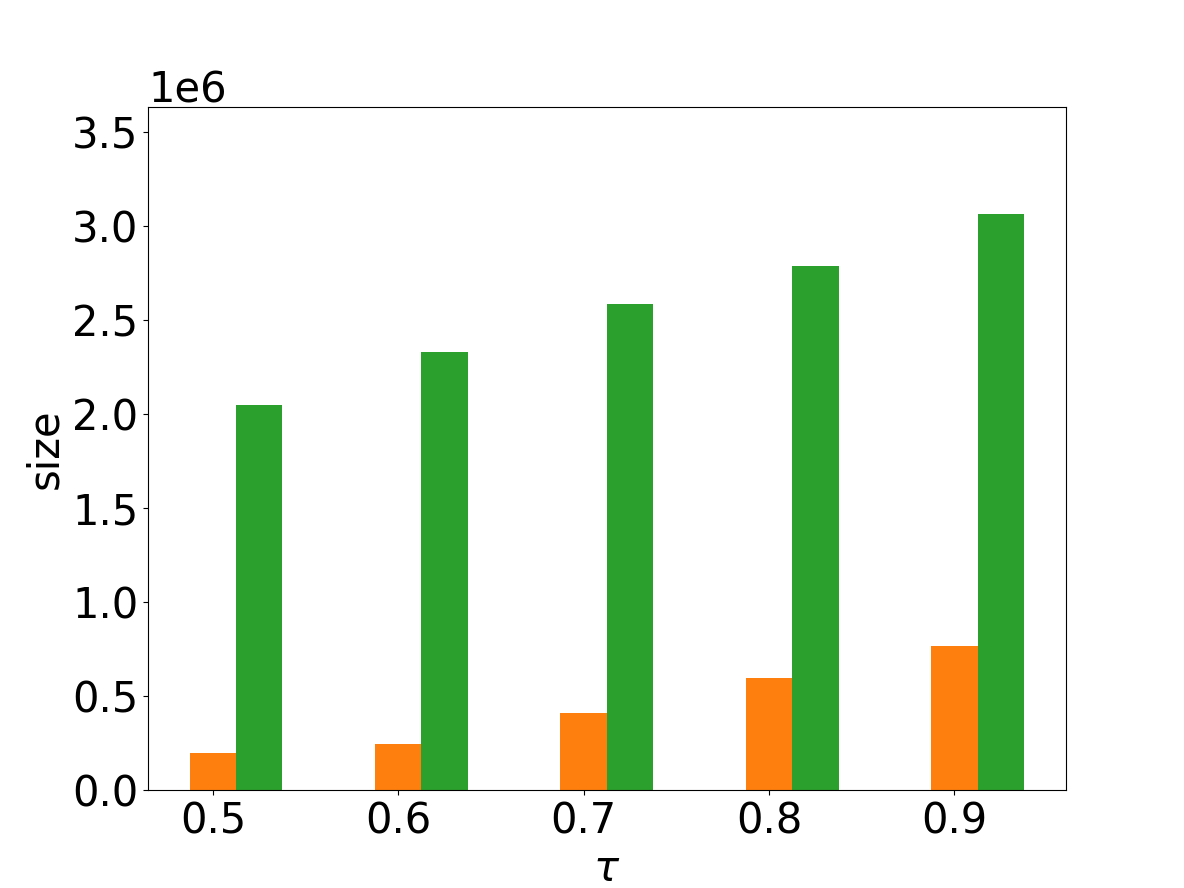}}
\subfloat[soc-pokec \label{1h}]{\includegraphics[width=4.2cm, height=3.36cm]{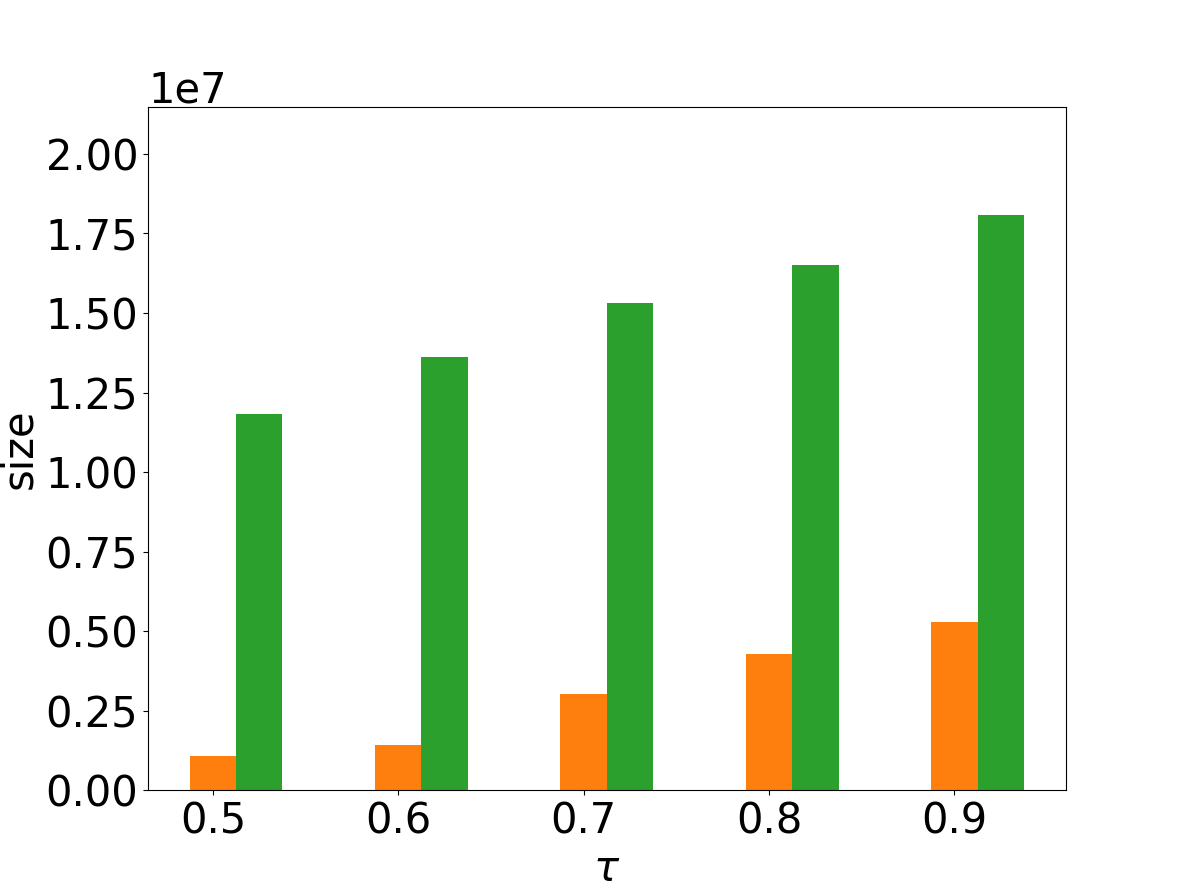}}
\subfloat[cit-Patents \label{1i}]{\includegraphics[width=4.2cm, height=3.36cm]{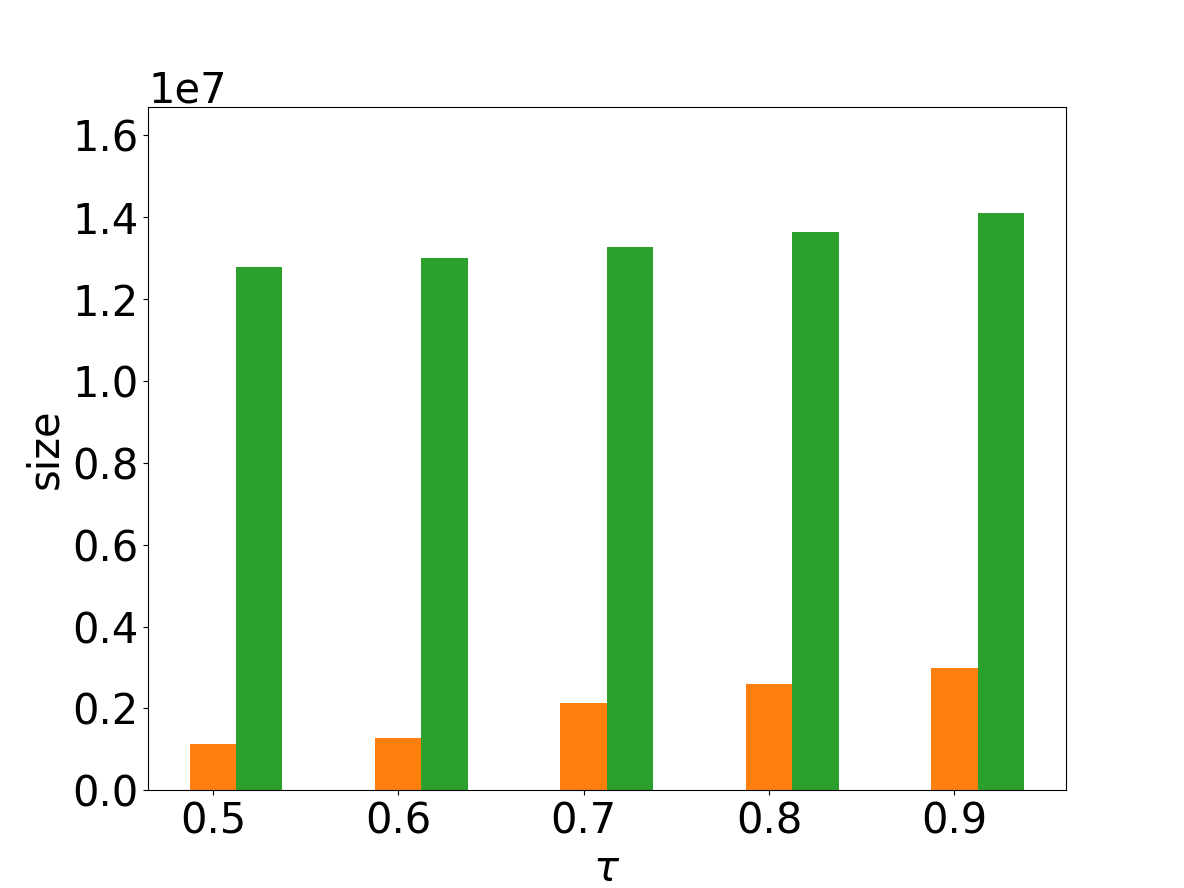}}


	\captionsetup{justification=centering}
	\caption{{\small Summary size of $\tau$-R$^+$MCE and $\tau$-RMCE on eight datasets with $\tau$ varied from 0.5 to 0.9, T  and U  as default}}\label{fig:stotal}
	\vspace{-10pt}
\end{figure*}

In this section, we look into three research questions by experiments.  
(1) To what extent the summary size and running time can be reduced by $\tau$-R$^+$MCE vs. $\tau$-RMCE? 
(2) To what extent the effectiveness of $\tau$-R$^+$MCE can be further improved by our newly proposed truss order and truss bound? 
(3) To what extent our newly designed truss order and bound affect the efficiency (both running time and memory requirement)? 
For short, we denote the the $\tau$-visible MCE algorithm \cite{wang_redundancy-aware_2013} by $\tau$-RMCE and ours by $\tau$-R$^+$MCE.
%
All algorithms are implemented in C++ and tested on a MacBook Pro with 16GB memory and Intel Core i7 2.6GHz CPU 64. 
We evaluated both effectiveness (in terms of summary size) and efficiency (in terms of first-result time, total running time and total memory requirement) with $\tau$ varying  from $0.5$ to $0.9$.
$\tau$-R$^+$MCE and $\tau$-RMCE were implemented  with three types of bounds (truss bound (T), core bound (C), H bound (H)) as well as three vertex orders (truss order (U), degeneracy order (I), random order (R)). Details are shown in Table~\ref{table3}. 
All results  were reported by an average of five runs. 
 
{\bf Datasets} 
{We use eight real-world datasets from different domains with various data properties to show the robustness of our algorithms. 
To provide a more comprehensive comparison with $\tau$-RMCE, we also test the algorithms on the dataset {\it cit-Patents} that contains the largest number of vertices used by the work~\cite{wang_redundancy-aware_2013}.} 
Details are shown in Table~\ref{datasets}. 
For each dataset, 
we denote by $\lvert V\rvert$ the number of vertices, 
by $\lvert E\rvert$ the number of edges and
by Cliques the total number of maximal cliques as reference. 
The 5th and 6th column denotes the fraction 
$summary\mbox{ }size/total\mbox{ }number\mbox{ }of\mbox{ }maximal\mbox{ }cliques$ for $\tau$-RMCE-TU and $\tau$-R$^+$MCE-TU with the best configuration (Truss bound (T) and Truss order (U)) respectively. 
The percentage before and after / is the value at $\tau = 0.5$ and $\tau = 0.9$ respectively. 
For example, for soc-Epinions1 at the 5th column, $18.1\%/77.9\%$ means that the sizes of summaries produced by $\tau$-RMCE occupy $18.1\%$ and $77.9\%$ of the total number of maximal cliques at $\tau = 0.5$ and $\tau = 0.9$. 
All datasets used in this paper can be found in {Stanford Large Network Dataset Collection}\footnote{Available at http://snap.stanford.edu/data/index.html}.

\subsection{Effectiveness}
\label{sec:size}

To evaluate the effectiveness of our algorithm, 
 we compare the size of the summaries generated by $\tau$-RMCE and $\tau$-R$^+$MCE in Section~\ref{sec:ss} (both with T bound and U order as default). 
 To see to what extent our proposed truss bound and truss order benefit effectiveness, 
 we implemented  $\tau$-RMCE and $\tau$-R$^+$MCE with three orders (U, I, R, bound T as default) in Section~\ref{sc:6.1.2}, and with three bounds (T, C, H, order U as default) in Section~\ref{sc:6.1.3}. 

\begin{figure*}[ht]
\vspace{-12pt}
	\centering

	\subfloat[soc-Epinions1 \label{3a}]{\includegraphics[width=4.2cm, height=3.36cm]{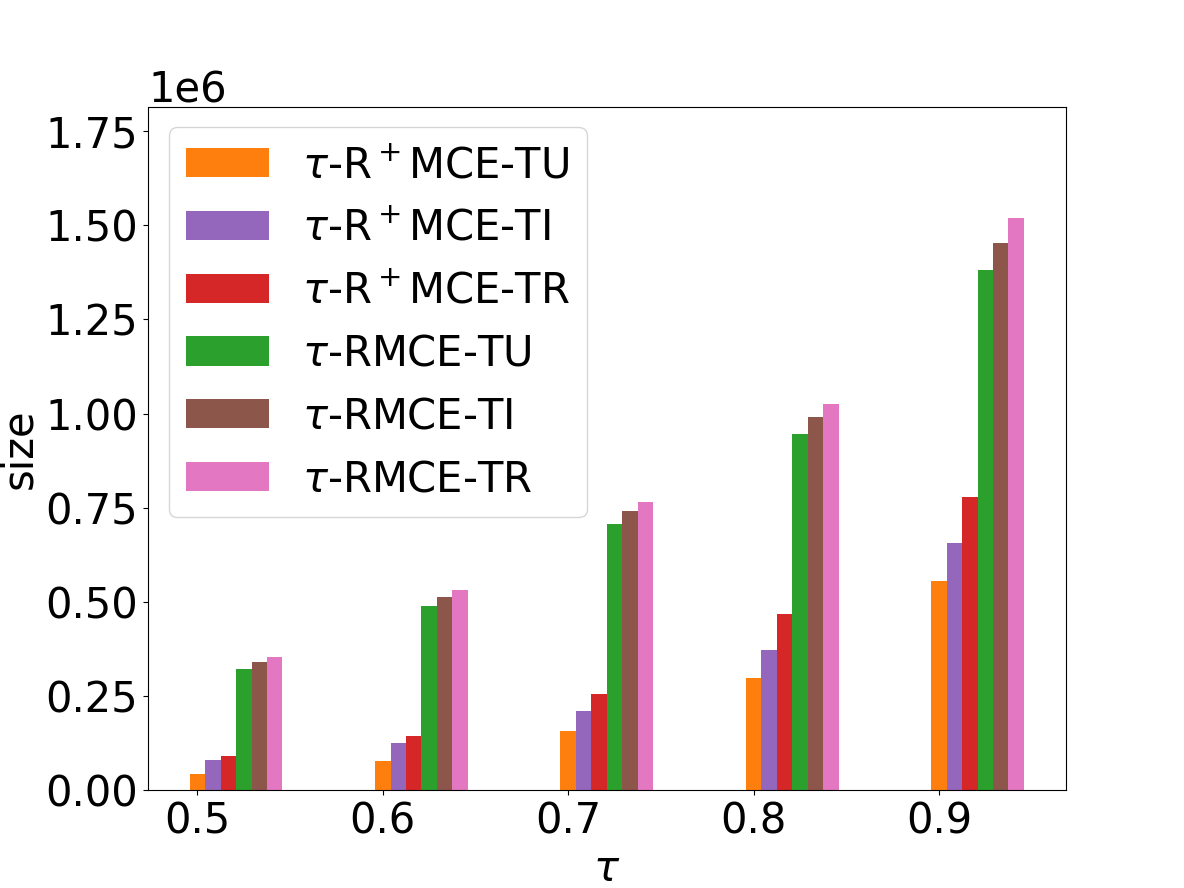}}
	\subfloat[loc-Gowalla \label{3b}]{\includegraphics[width=4.2cm, height=3.36cm]{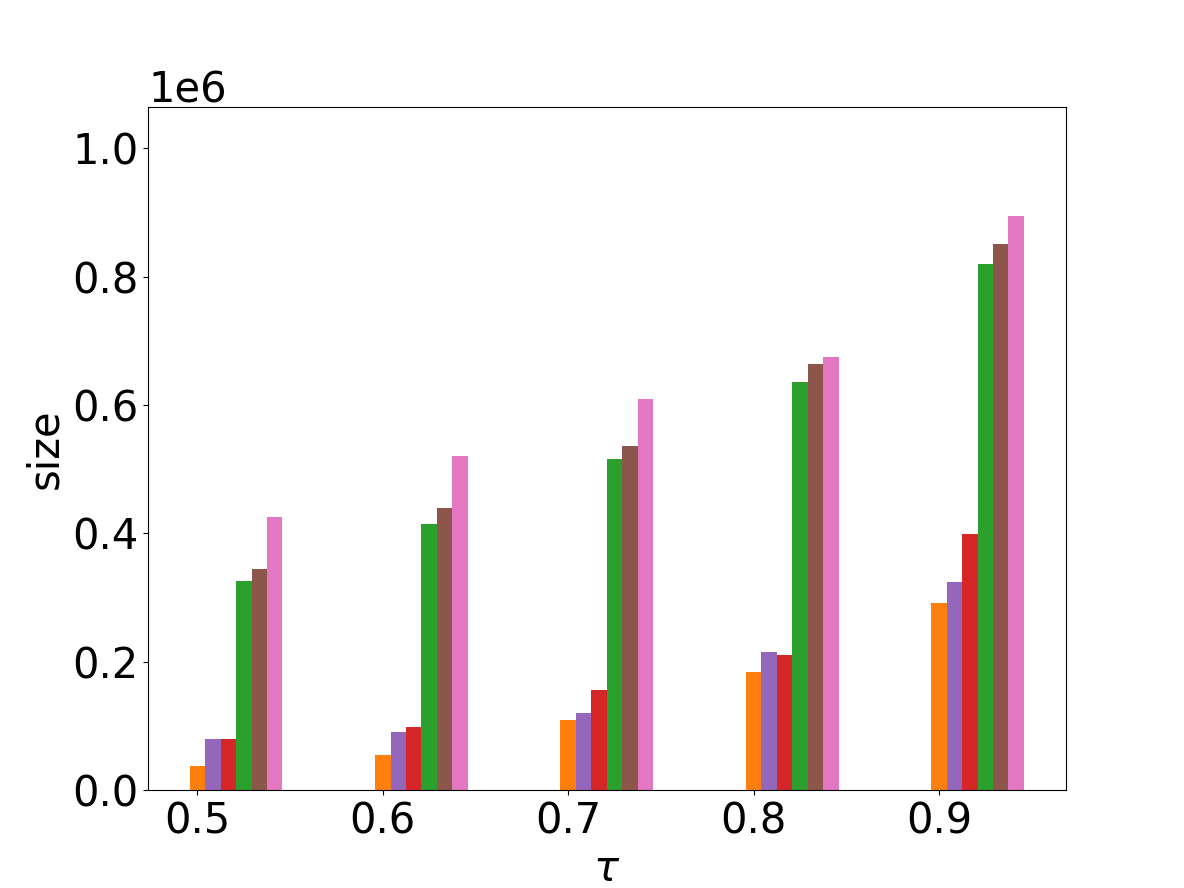}}
	\subfloat[amazon0302 \label{3c}]{\includegraphics[width=4.2cm, height=3.36cm]{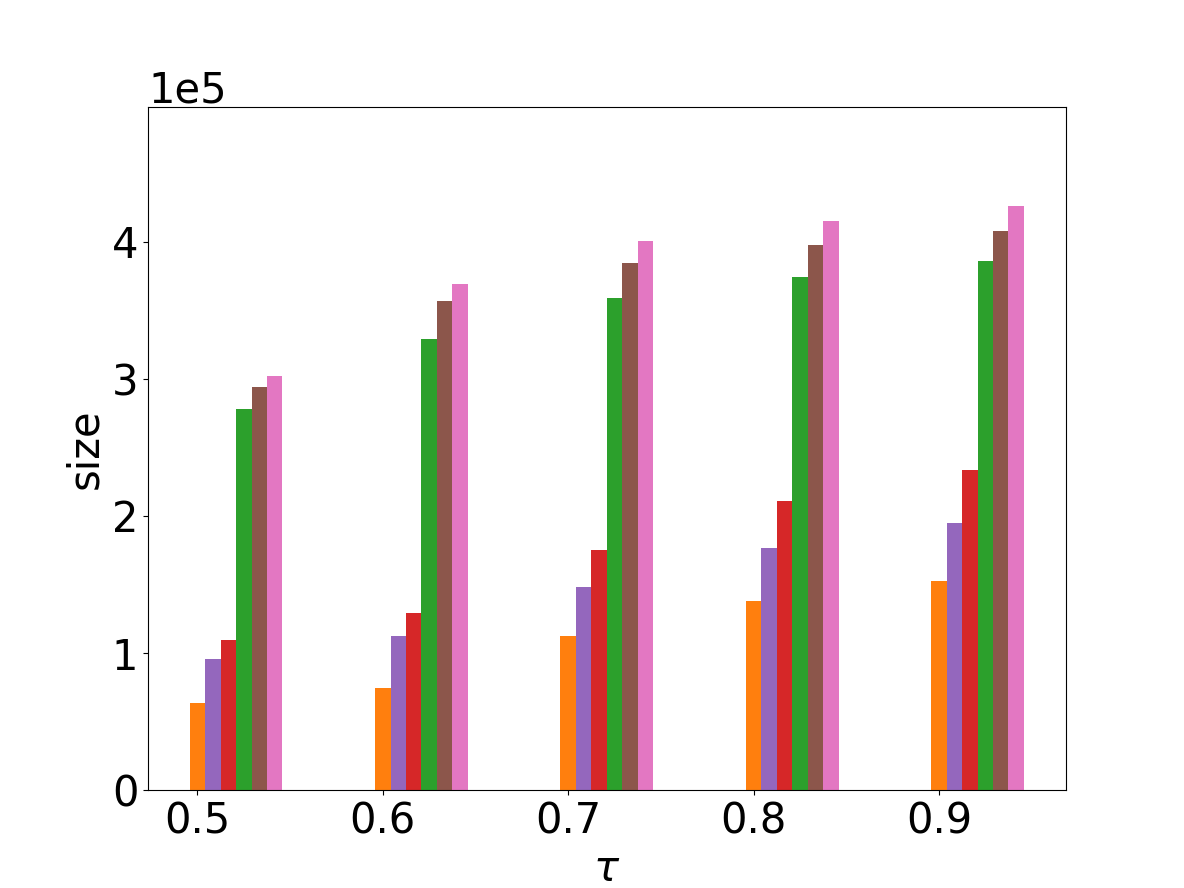}}
\subfloat[email-EuAll  \label{3d}]{\includegraphics[width=4.2cm, height=3.36cm]{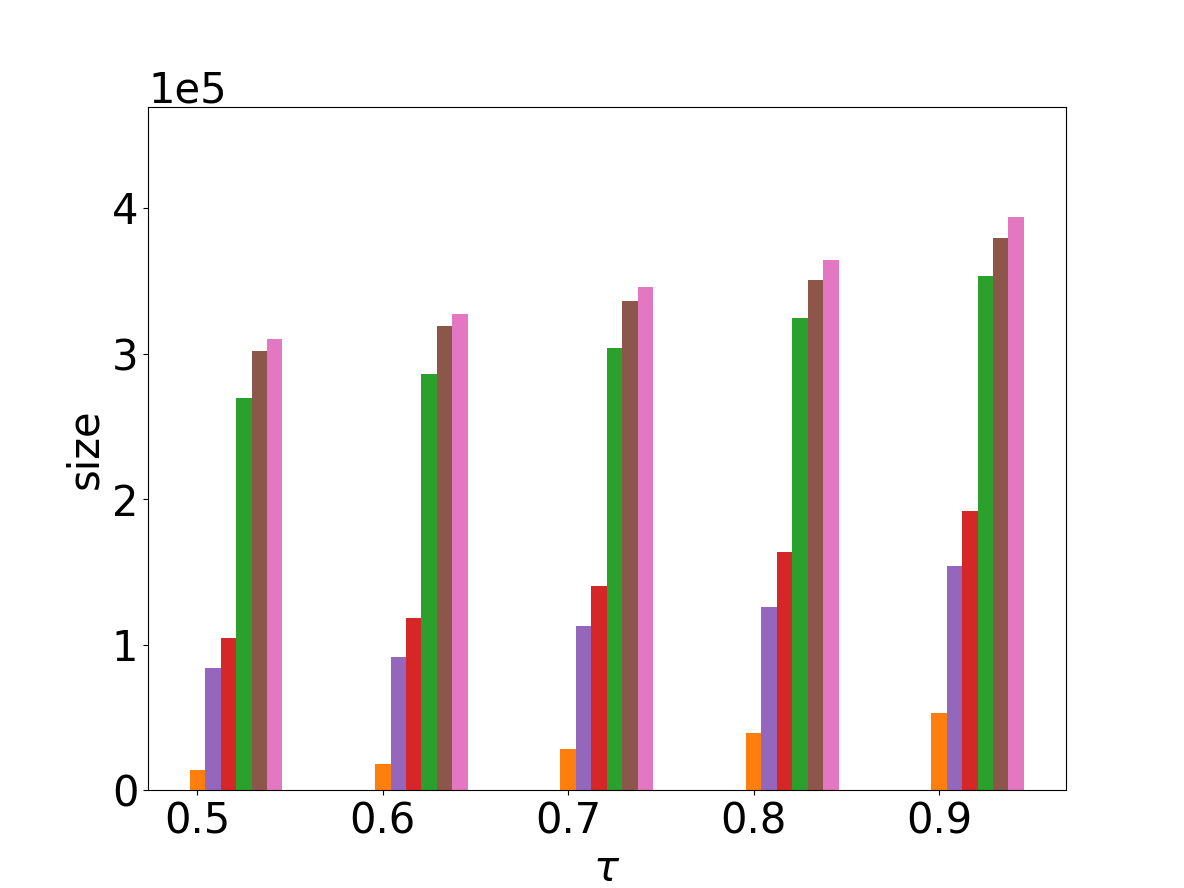}}	

\vspace{-11pt}

\subfloat[web-NotreDame \label{3f}]{\includegraphics[width=4.2cm, height=3.36cm]{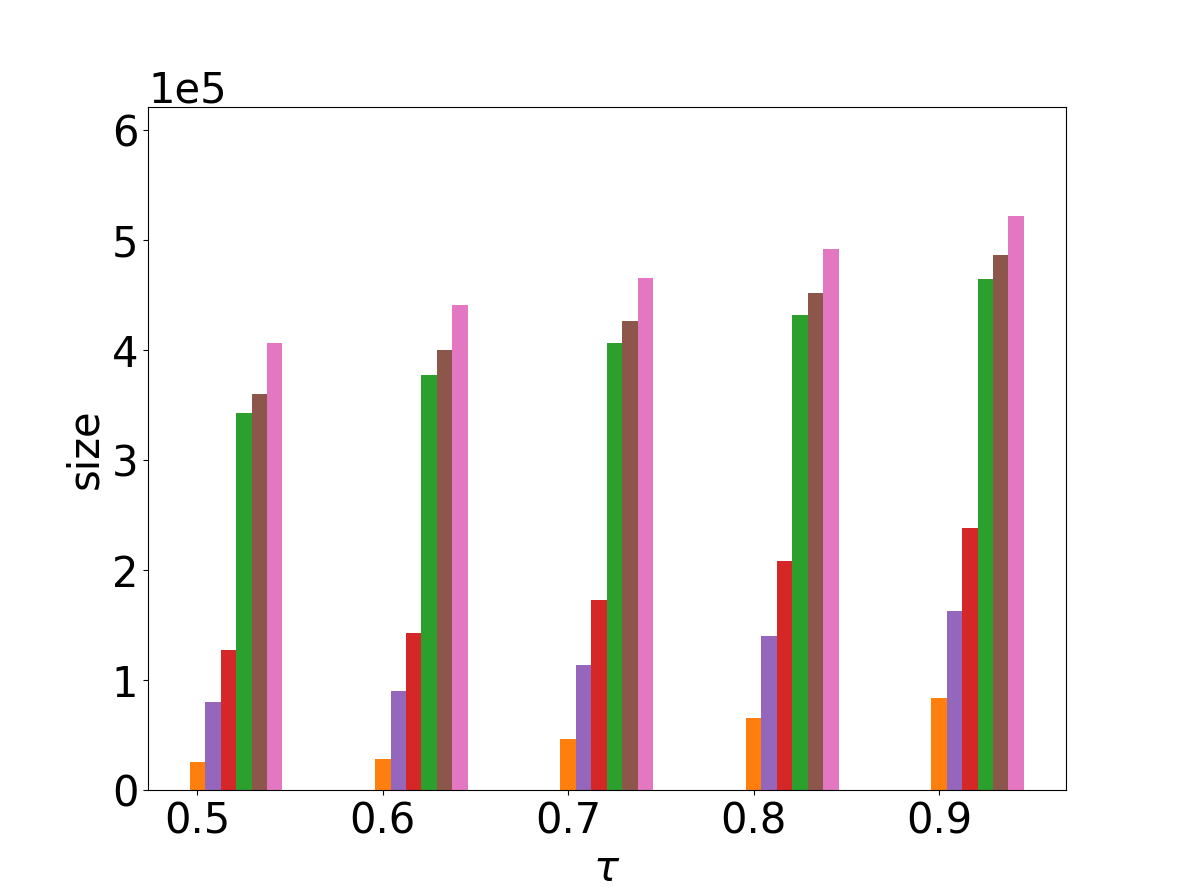}}
\subfloat[com-youtube \label{3g}]{\includegraphics[width=4.2cm, height=3.36cm]{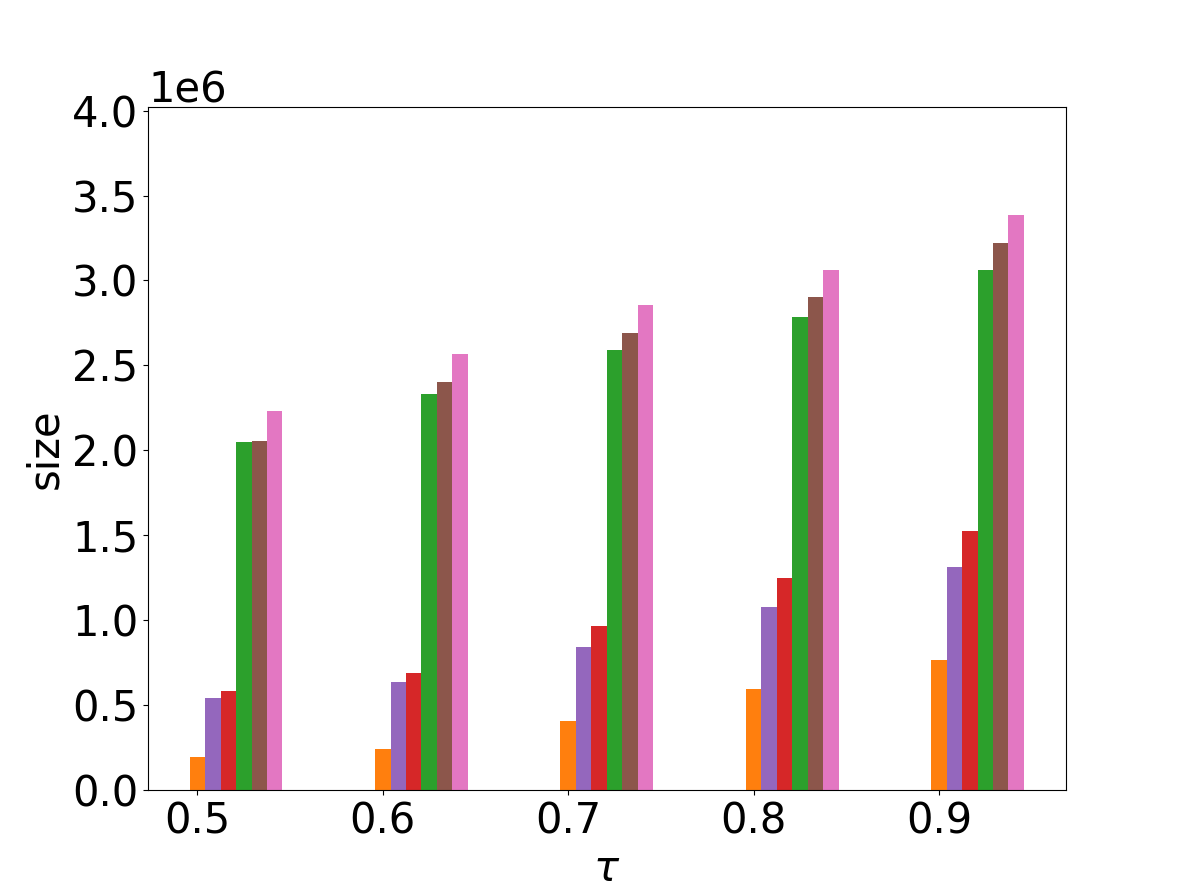}}
\subfloat[soc-pokec \label{3h}]{\includegraphics[width=4.2cm, height=3.36cm]{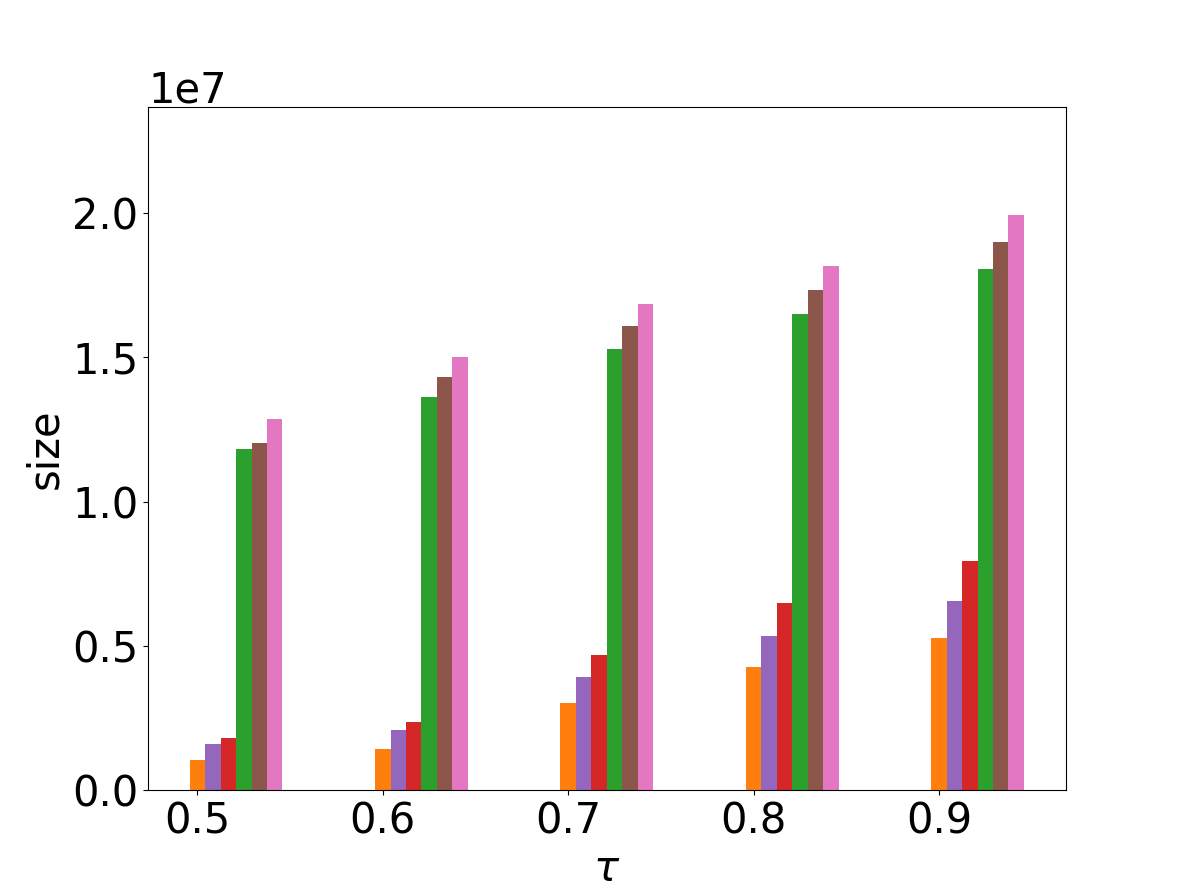}}
\subfloat[cit-Patents \label{3i}]{\includegraphics[width=4.2cm, height=3.36cm]{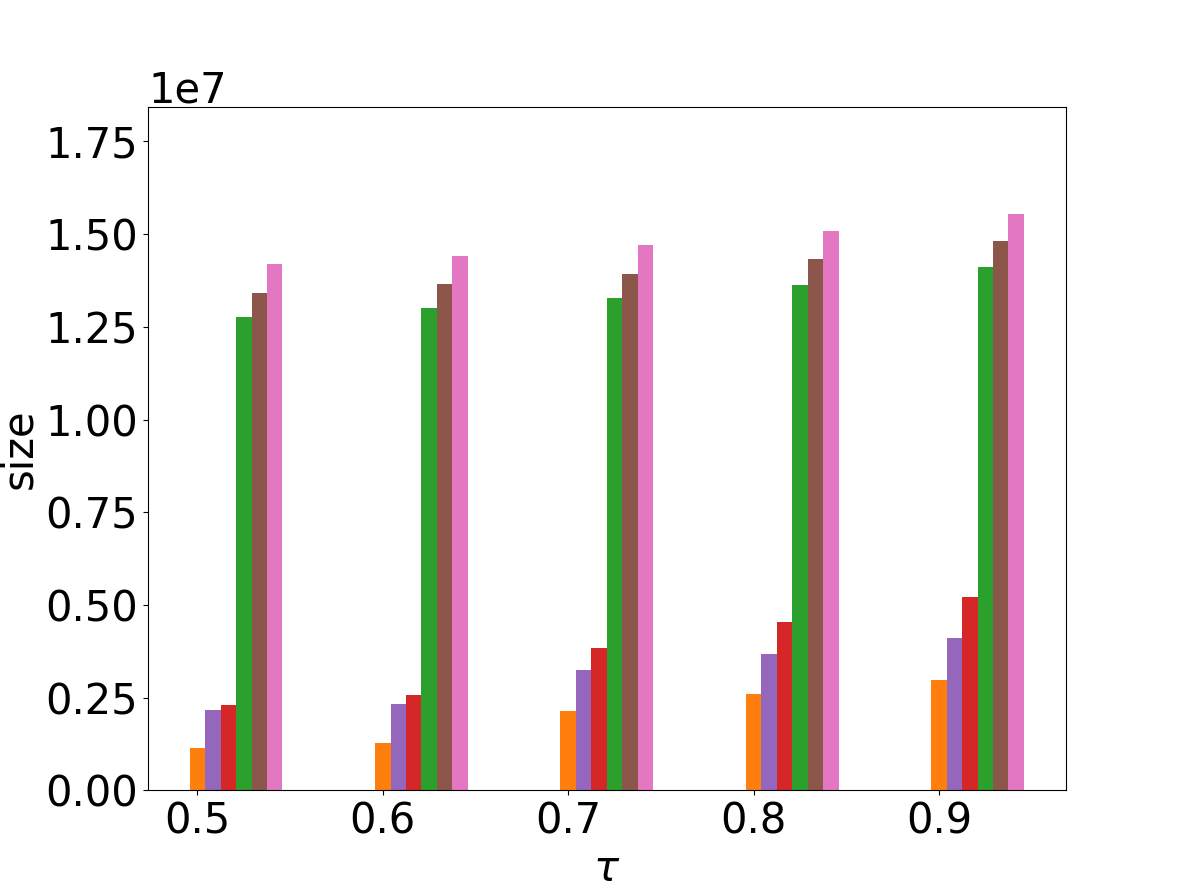}}


	\captionsetup{justification=centering}
	\caption{Summary size of $\tau$-R$^+$MCE and $\tau$-RMCE on eight datasets with different orders, $\tau$ varies from 0.5 to 0.9, T bound as default}\label{fig:sorder}
	\vspace{-18pt}
\end{figure*}

\begin{figure*}[ht]
	\centering

	\subfloat[soc-Epinions1 \label{3a}]{\includegraphics[width=4.2cm, height=3.36cm]{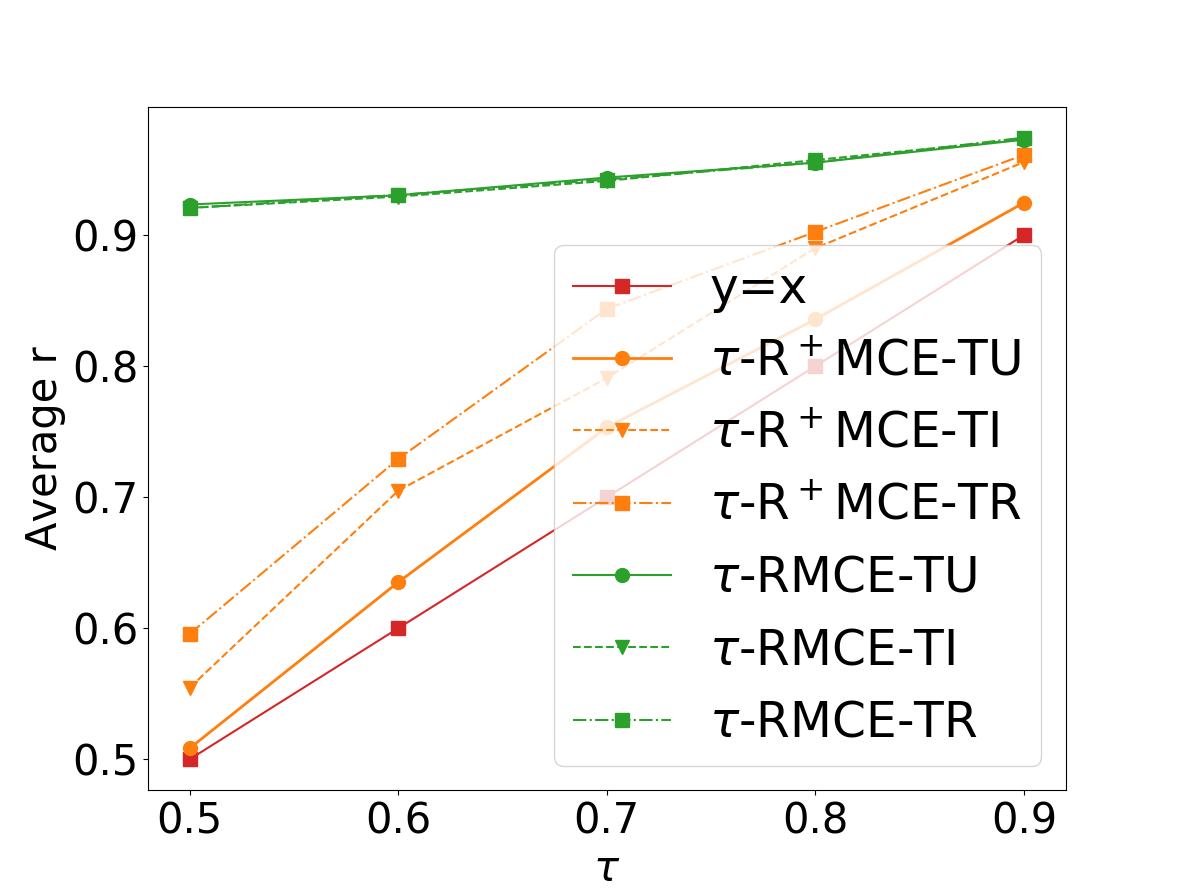}}
	\subfloat[loc-Gowalla \label{3b}]{\includegraphics[width=4.2cm, height=3.36cm]{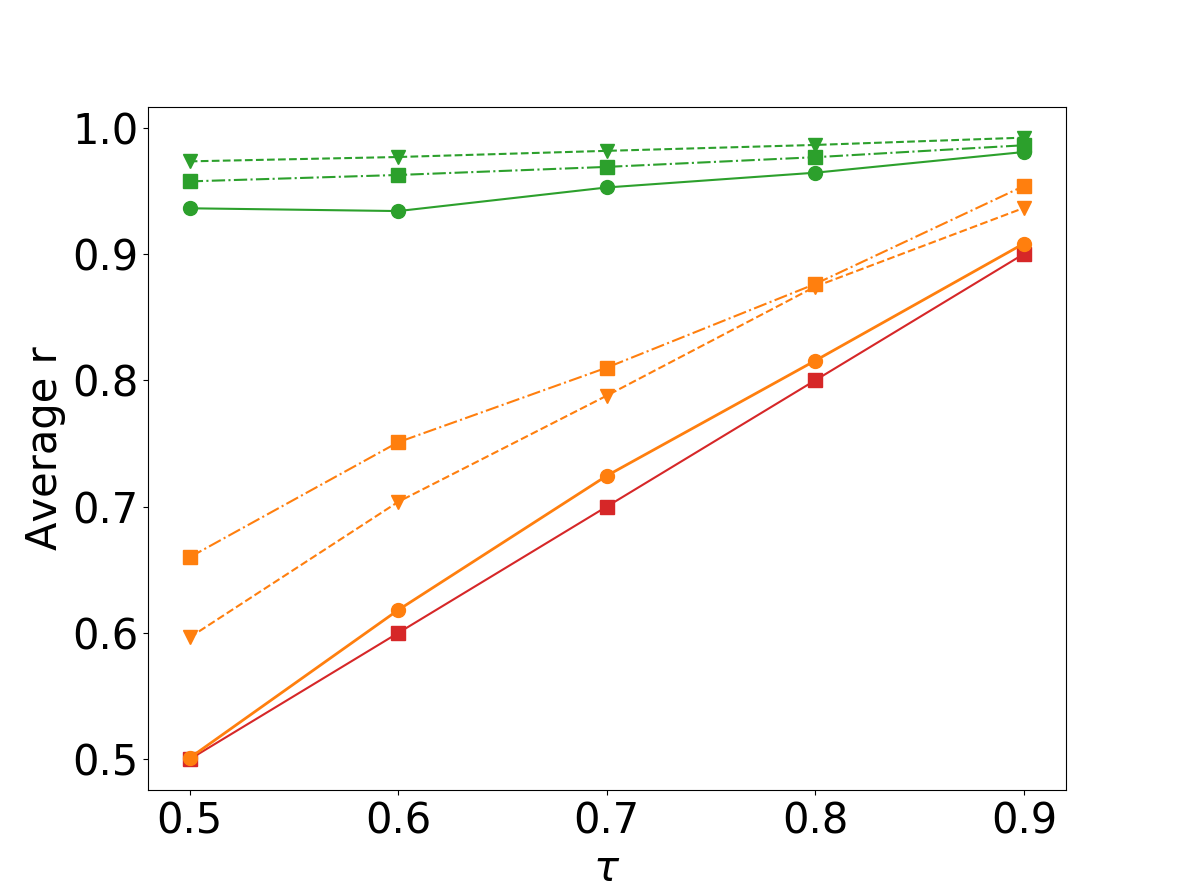}}
	\subfloat[amazon0302 \label{3c}]{\includegraphics[width=4.2cm, height=3.36cm]{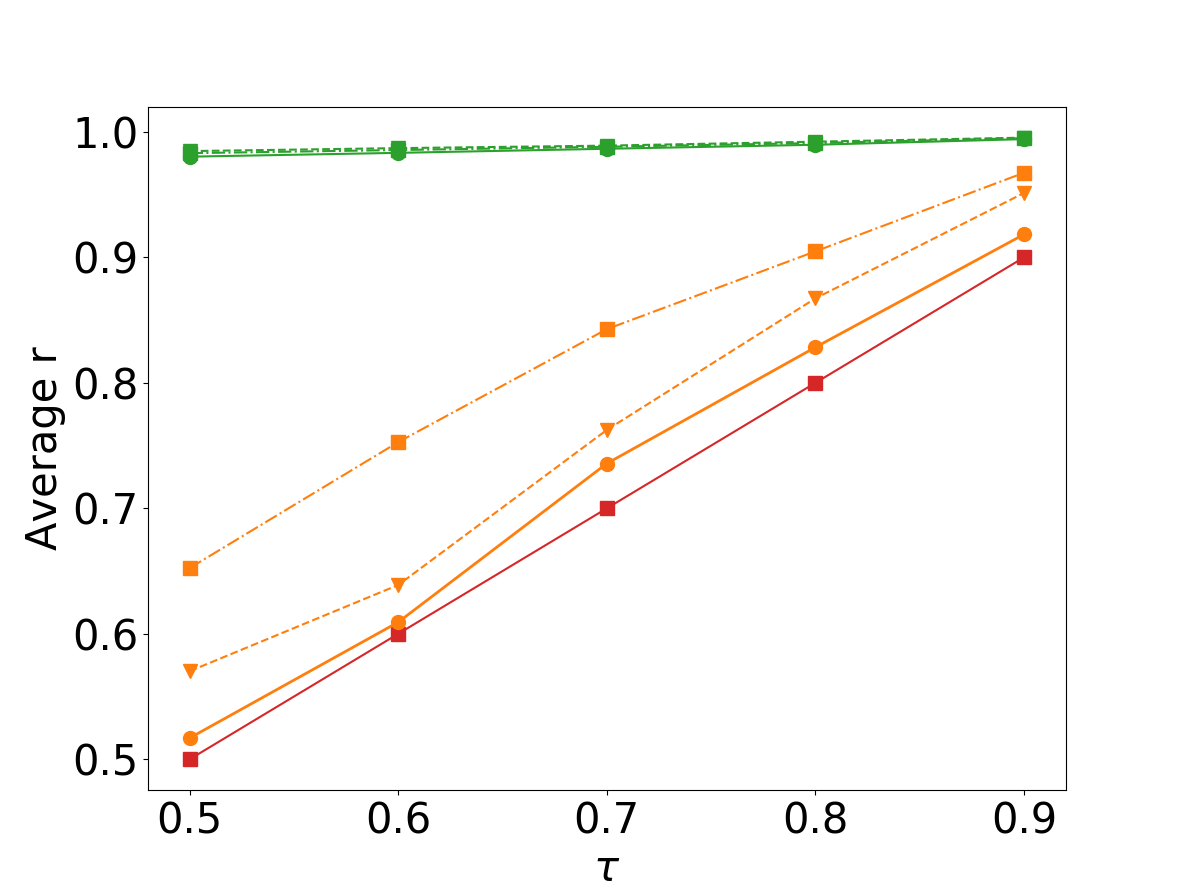}}
\subfloat[email-EuAll  \label{3d}]{\includegraphics[width=4.2cm, height=3.36cm]{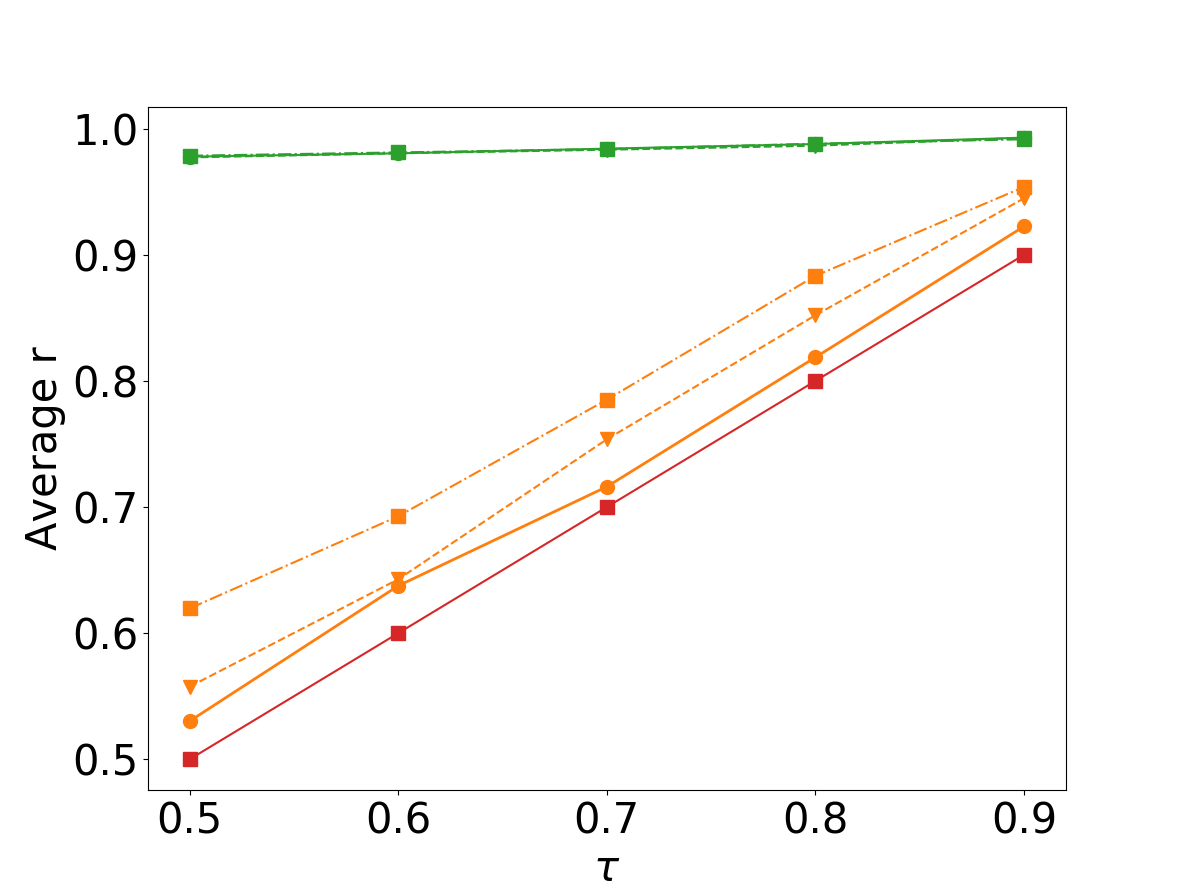}}	

\vspace{-11pt}

\subfloat[web-NotreDame \label{3f}]{\includegraphics[width=4.2cm, height=3.36cm]{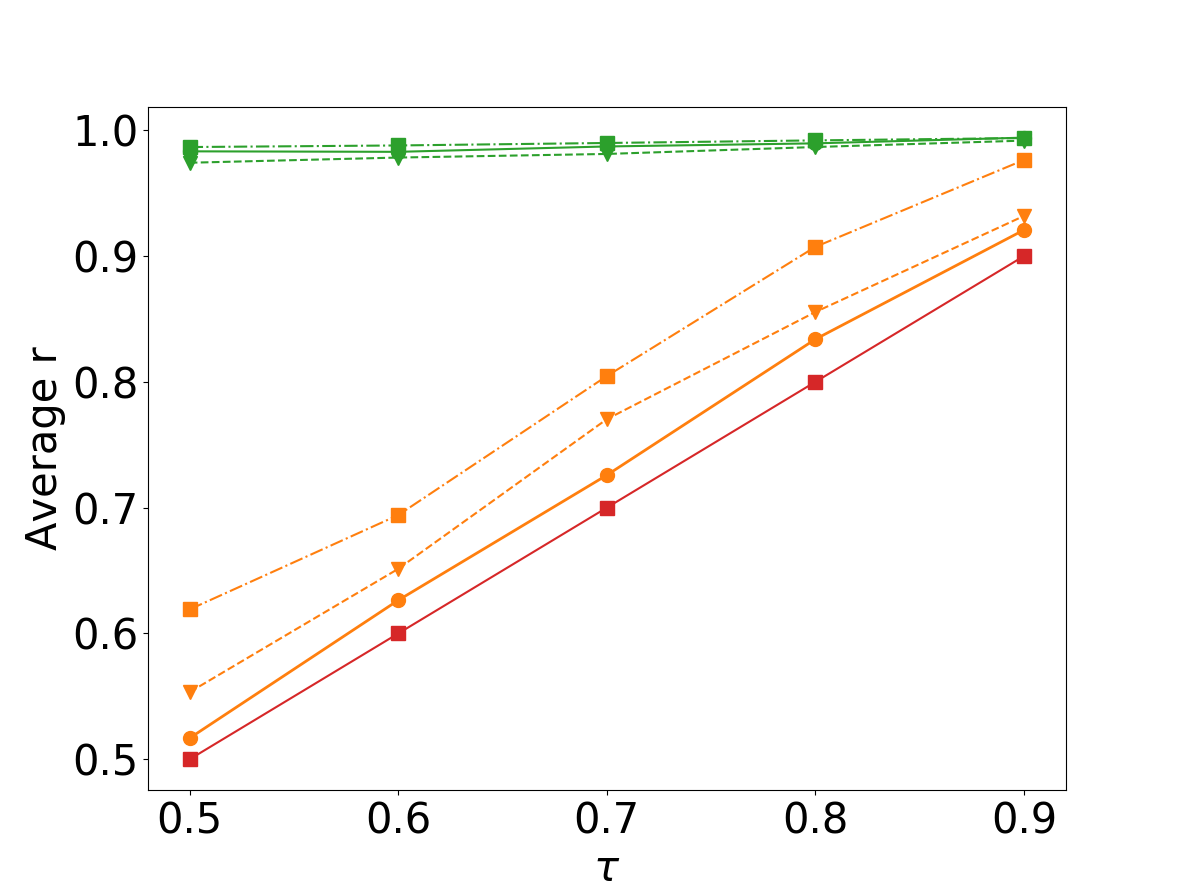}}
\subfloat[com-youtube \label{3g}]{\includegraphics[width=4.2cm, height=3.36cm]{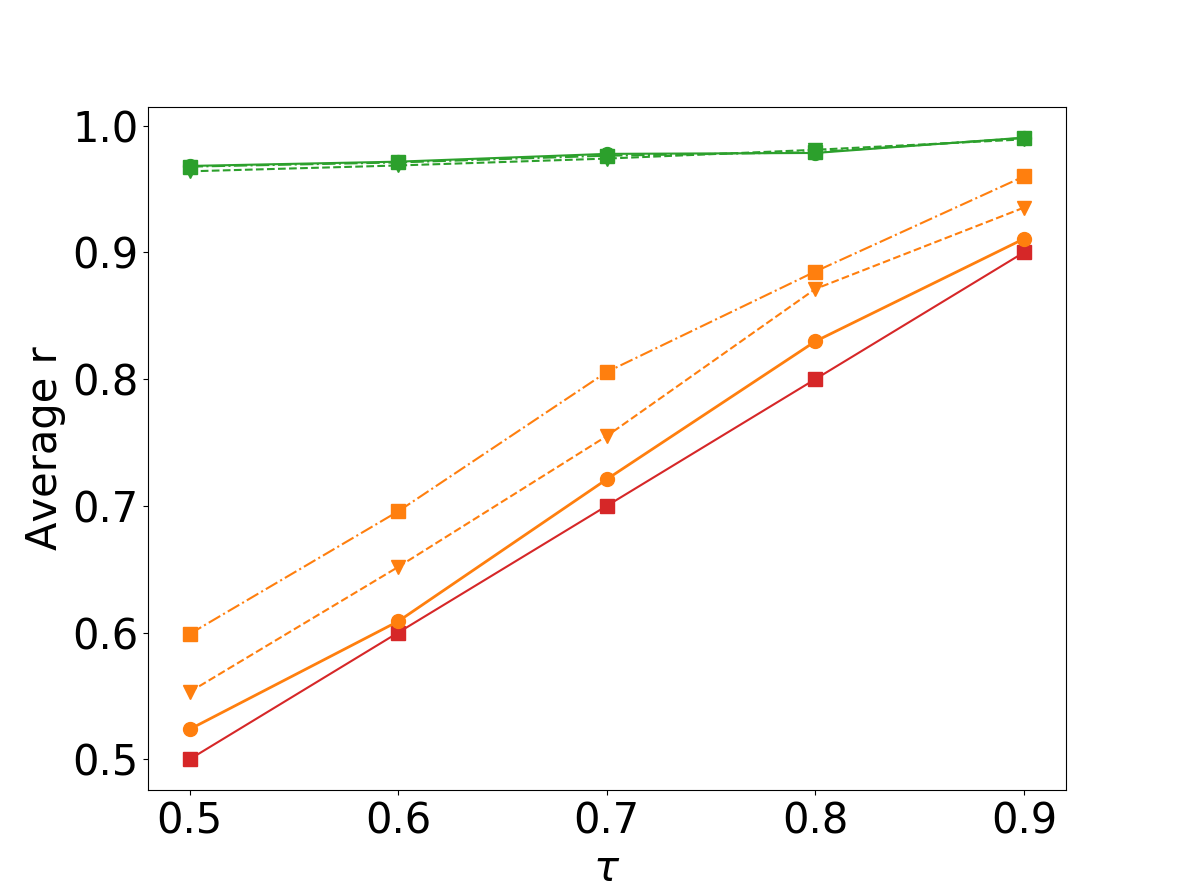}}
\subfloat[soc-pokec \label{3h}]{\includegraphics[width=4.2cm, height=3.36cm]{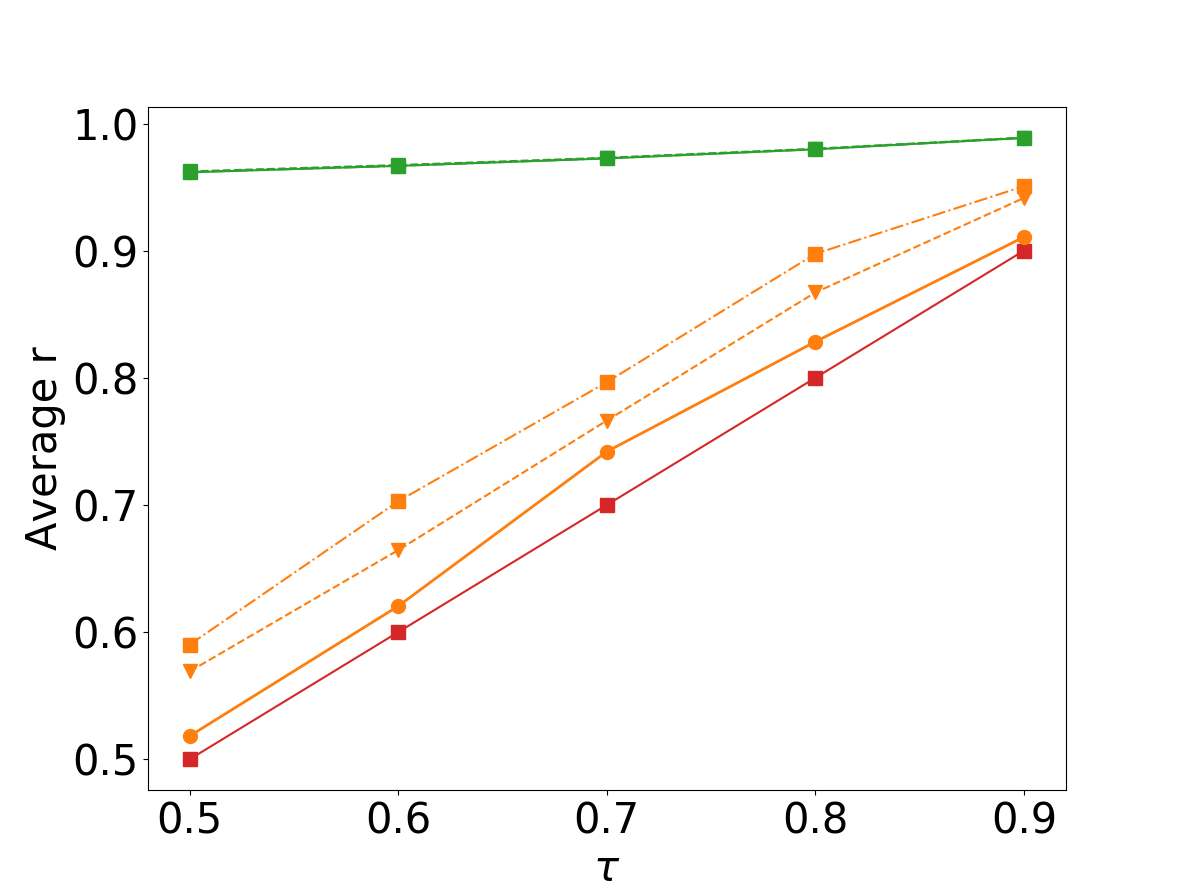}}
\subfloat[cit-Patents \label{3i}]{\includegraphics[width=4.2cm, height=3.36cm]{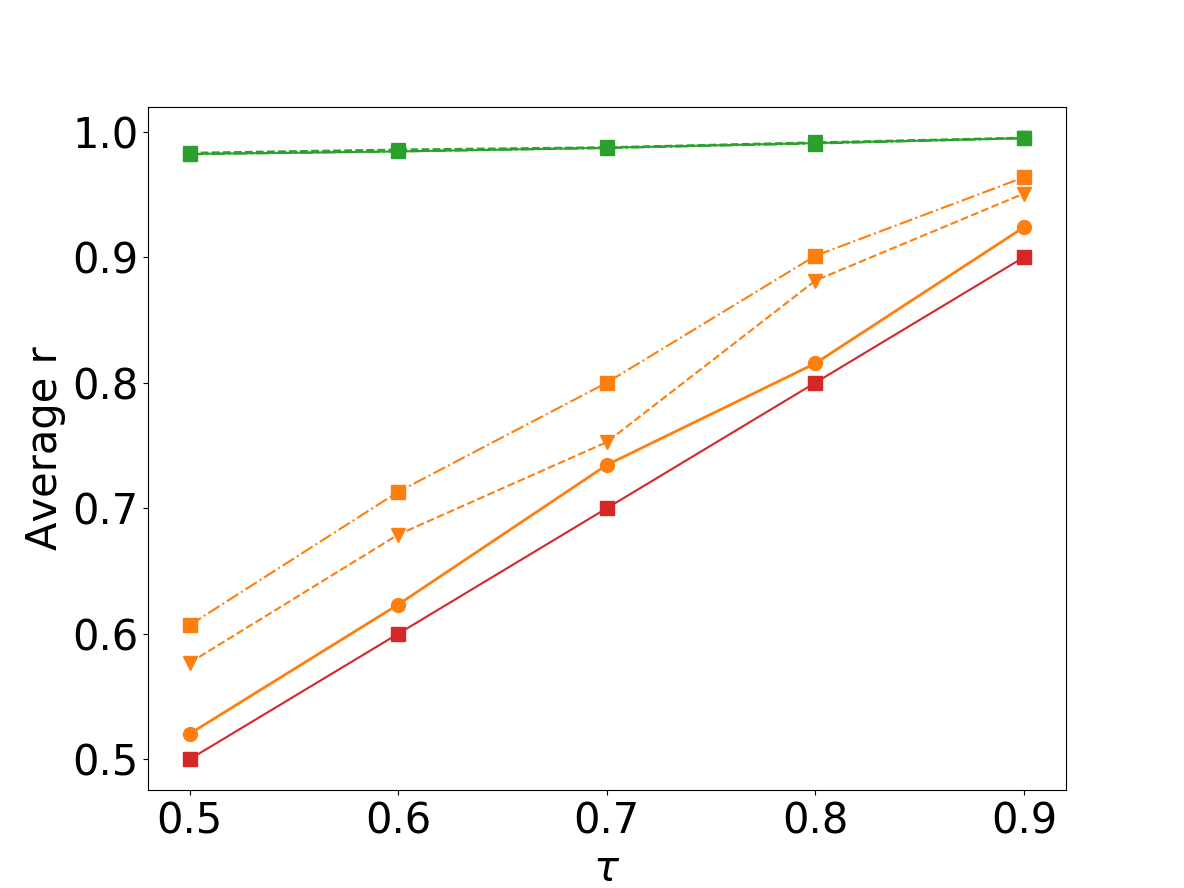}}


	\captionsetup{justification=centering}
	\caption{Average $r$ of $\tau$-R$^+$MCE and $\tau$-RMCE on eight datasets with different orders, $\tau$ varies from 0.5 to 0.9, T bound as default}\label{fig:OrderTau}
	\vspace{-10pt}
\end{figure*}

\subsubsection{Summary size}
\label{sec:ss}

We implemented $\tau$-RMCE and $\tau$-R$^+$MCE with the best configurations using truss bound and truss order, which are denoted by $\tau$-RMCE-TU and $\tau$-R$^+$MCE-TU respectively. 
The results are shown in Fig.~\ref{fig:stotal}. 
 We see that $\tau$-R$^+$MCE-TU consistently outperforms $\tau$-RMCE-TU on all datasets with all the $\tau$ values. 


When $\tau = 0.9$, $\tau$-R$^+$MCE-TU significantly reduces more than $50\%$ output cliques vs. $\tau$-RMCE-TU on all datasets,  
three of  which (Fig.~\ref{1g},~\ref{1h},~\ref{1i})  achieve $70\%$, 
and two of which (Fig.~\ref{1d},~\ref{1f}) even achieve more than $80\%$. 
When $\tau$ decreases, the difference is more dramatic,
i.e., the percentage of reduction monotonically increases. 
At $\tau = 0.5$, the reduction for all datasets is more than $70\%$, 
two of which (Fig.~\ref{1a},~\ref{1b}) reach $85\%$, 
and five of which (Fig.~\ref{1d},~\ref{1f},~\ref{1g},~\ref{1h},~\ref{1i}) reach  $90\%$  to reduce the summary size by more than one order of magnitude.
This monotonic increasing trend implies that the performance of $\tau$-R$^+$MCE-TU performs more significantly than $\tau$-RMCE-TU along with $\tau$ decreasing. 
 This is because for a small threshold, $\tau$-RMCE includes more unnecessary cliques whose visibilities are greater than $\tau$ into the summary with high probabilities,  
which confirms our intuition in Section \ref{sec:new} that  $s_{opt}(r)$ should be set to  $0$ for $r\in [\tau, 1]$. 
 Another reason is that for a clique $C$ whose visibility is close to $0$, $s(r)$ forces $\tau$-RMCE to output $C$ immediately, while $\tau$-R$^+$MCE considers the potential that $C$ may be covered by some future cliques, thus more carefully outputs such a clique with a proper probability. 
To show the robustness of our proposed method, we tested the algorithms on eight real-world datasets with different scales. 
The results show that $\tau$-R$^+$MCE achieves relatively better performance on large graphs. 
For the convenience of our discussion, now we focus on the results at $\tau = 0.5$. 
We see that all the five datasets (Fig.~\ref{1d},~\ref{1f},~\ref{1g},~\ref{1h},~\ref{1i}) that have more than $90\%$ reductions  are the top five largest graphs among all eight datasets.
This implies that our proposed method are more capable to handle contemporary large scale graphs than the state-of-the-art approach.

 \subsubsection{Effect of vertex orders}
 \label{sc:6.1.2}
To see  to what extent the performance of $\tau$-R$^+$MCE can be further improved by employing a vertex order with strong locality, we implemented $\tau$-R$^+$MCE and $\tau$-RMCE with three types of orders: random order (R), degeneracy order (I) and truss order (U).   
The default bound was set as truss bound (T). 
The results are shown in Fig.~\ref{fig:sorder}. 

Fig.~\ref{fig:sorder} shows that truss order consistently outperforms degeneracy order  and random order for both $\tau$-R$^+$MCE and $\tau$-RMCE, 
while generally  degeneracy order is superior to random order except one exception 
($\tau$-R$^+$MCE on {\it loc-Gowalla} at $\tau = 0.8$). 
Now we focus our discussion on $\tau$-R$^+$MCE. 
Generally $\tau$-R$^+$MCE-TI 
is superior to
 $\tau$-R$^+$MCE-TR for all $\tau$ values on 7 out of 8 datasets (except {\it web-NotreDame}), 
while the reduction percentage  ($10\%\sim  20\%$) is not significant. 
However, the reduction for $\tau$-R$^+$MCE-TU vs. $\tau$-R$^+$MCE-TI is much dramatic: 
at $\tau = 0.5$, the reduction percentage varies from $33\%$ ({\it soc-pokec}) to $83\%$ ({\it com-youtube}), and 5 out of 8 achieve  more than $50\%$ (except {\it soc-Epinions1, amazon0103, soc-pokec}).

{To further show how different orders help our approach get close to the optimal, we calculate the actual average $r$ for each dataset w.r.t different $\tau$ values. 
 We take every maximal clique generated by MCE into consideration. If a search subtree should be pruned at a certain stage, we keep searching the subtree to calculate $r$ of the discarded cliques while disallowing these cliques to be added into  summary $\mathcal{S}$.}  
{Results are shown in Fig.~\ref{fig:OrderTau}. 
This set of experiments show that with $\tau$ decreasing from $0.9$ to $0.5$, 
the average visibility $r$ of $\tau$-RMCE changes very mildly and keeps no less than $0.9$ for all the datasets, 
regardless which vertex order is implemented. 
For datasets except {\it loc-Gowalla}, the  lines of $\tau$-RMCE overlap with each other, which implies that different types of vertex orders may have  limited impact on the state-of-the-art approach.} 
{However, 
the three  lines of $\tau$-R$^+$MCE drop sharply from around $0.95$ to below $0.6$ for every dataset. They  are thus much closer to the optimal reference line (represented by the  user-specified threshold in red). 
$\tau$-R$^+$MCE-TR shows the worst performance among the tested three vertex orders,  of which the line drops from around $0.95$ to $0.6$ for all the datasets. 
 $\tau$-R$^+$MCE-TI shows better effectiveness (which drops from around $0.95$ to $0.55$) than $\tau$-R$^+$MCE-TR,  
and $\tau$-R$^+$MCE-TU shows the best performance (which drops from $0.93$ to $0.52$). 
We see that the setting of truss order U successfully helps $\tau$-R$^+$MCE step forward to the optimal reference line, which leaves a narrow gap between them.} 

Fig.~\ref{fig:sorder} and Fig.~\ref{fig:OrderTau} confirm our assumption that the effectiveness of $\tau$-R$^+$MCE can be further improved by properly reordering vertices. The newly designed truss order significantly outperforms the degeneracy order by a large margin due to strong locality provided by the cohesiveness of $k$-truss.

\subsubsection{Effect of bounds}
\label{sc:6.1.3}
To see  to what extent the effectiveness of $\tau$-R$^+$MCE can be further improved by employing a tight bound, we implemented $\tau$-R$^+$MCE and $\tau$-RMCE with three different bounds: H bound (H), core bound (C) and truss bound (T).  
Truss order (U) was set to be the default.  
The results are shown in Fig.~\ref{fig:sbound}. 

We see that for both $\tau$-R$^+$MCE and $\tau$-RMCE, the performance of effectiveness consistently follows this order: T outperforms C, and C outperforms H. 
When we focus on $\tau$-R$^+$MCE, results show that $\tau$-R$^+$MCE-CU reduces the summary size vs. $\tau$-R$^+$MCE-HU by less than $10\%$ for all $\tau$ values on all datasets. 
However, the reduction between $\tau$-R$^+$MCE-TU and $\tau$-R$^+$MCE-CU ranges from $21\%$ to $43\%$. At $\tau = 0.5$, the percentage achieves more than $30\%$ for four out of eight datasets (except {\it email-EuAll, web-NotreDame, com-youtube, cit-Patents}). 

{
Fig.~\ref{fig:BoundTau} reports the average of  $r$ for $\tau$-RMCE and $\tau$-R$^+$MCE with different bounds. The results are similar to Fig.~\ref{fig:OrderTau}. 
We see that the three lines of $\tau$-RMCE stay much closer with each other for all datasets, 
and the line of $\tau$-RMCE-TU keeps slightly lower than the other two lines. 
This implies that various bounds can provide only limited improvements to the existing sampling approach. 
Whereas for $\tau$-R$^+$MCE, we see that a better bound helps the three lines move towards the required threshold by a significant margin. 
The gap between $\tau$-R$^+$MCE-TU and $\tau$-R$^+$MCE-CU is much more dramatic than that between $\tau$-R$^+$MCE-CU and $\tau$-R$^+$MCE-HU, which confirms the superiority of the newly applied truss bound T.}  

Fig.~\ref{fig:sbound} and Fig.~\ref{fig:BoundTau} confirm the fact that the effectiveness of $\tau$-R$^+$MCE can be further improved by employing tight bounds. 
Although the extent of benefit brought by good bounds is inferior to that brought by vertex orders with strong locality, 
our proposed truss bound still surpasses the state-of-the-art core bound by a significant margin.

%

\begin{figure*}[ht]
\vspace{-12pt}
	\centering

	\subfloat[soc-Epinions1 \label{3a}]{\includegraphics[width=4.2cm, height=3.36cm]{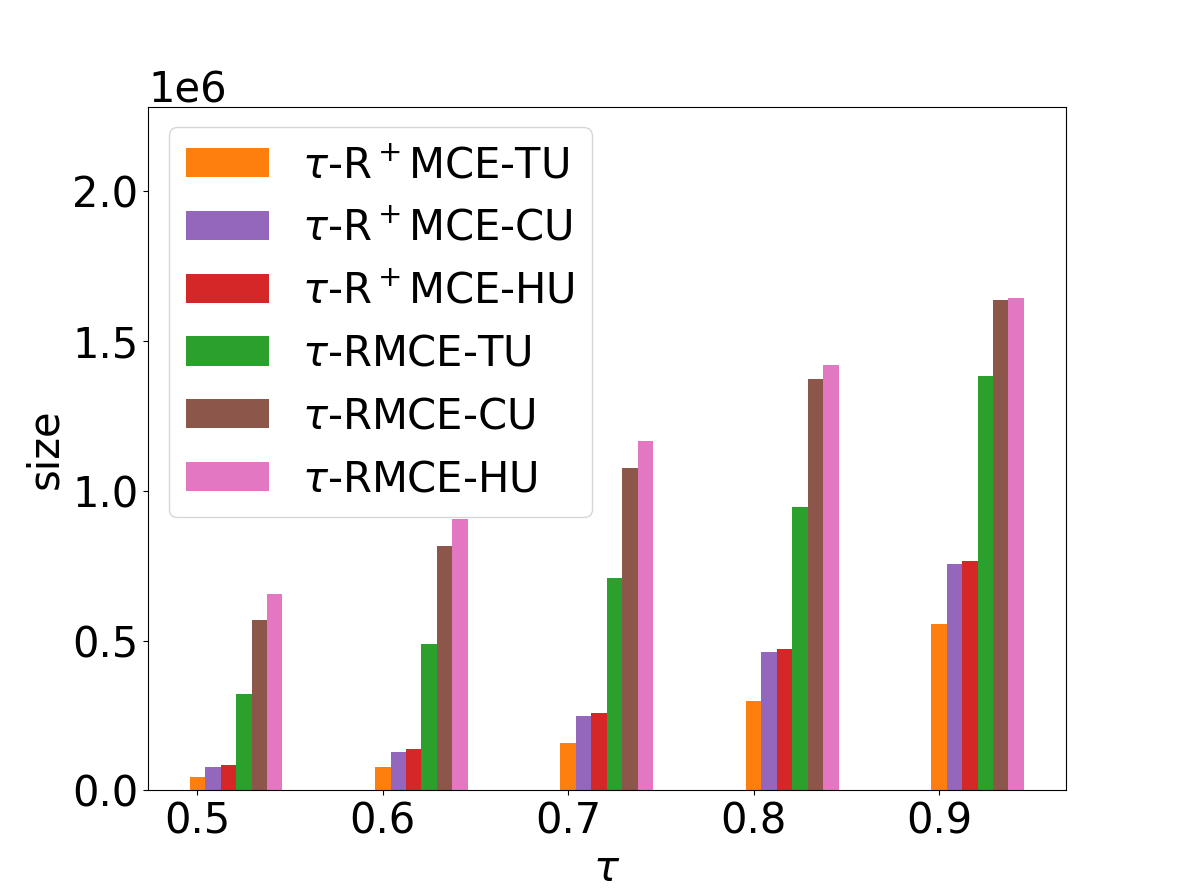}}
	\subfloat[loc-Gowalla \label{3b}]{\includegraphics[width=4.2cm, height=3.36cm]{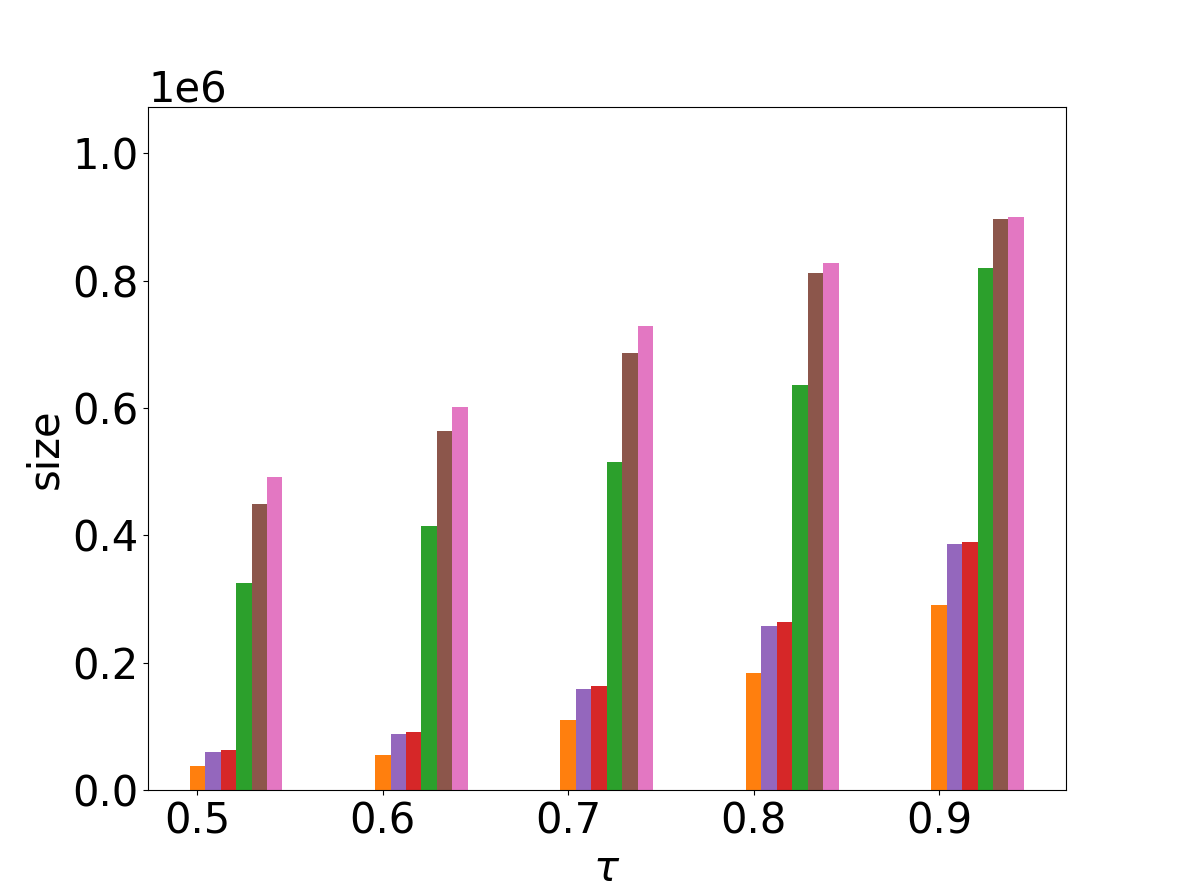}}
	\subfloat[amazon0302 \label{3c}]{\includegraphics[width=4.2cm, height=3.36cm]{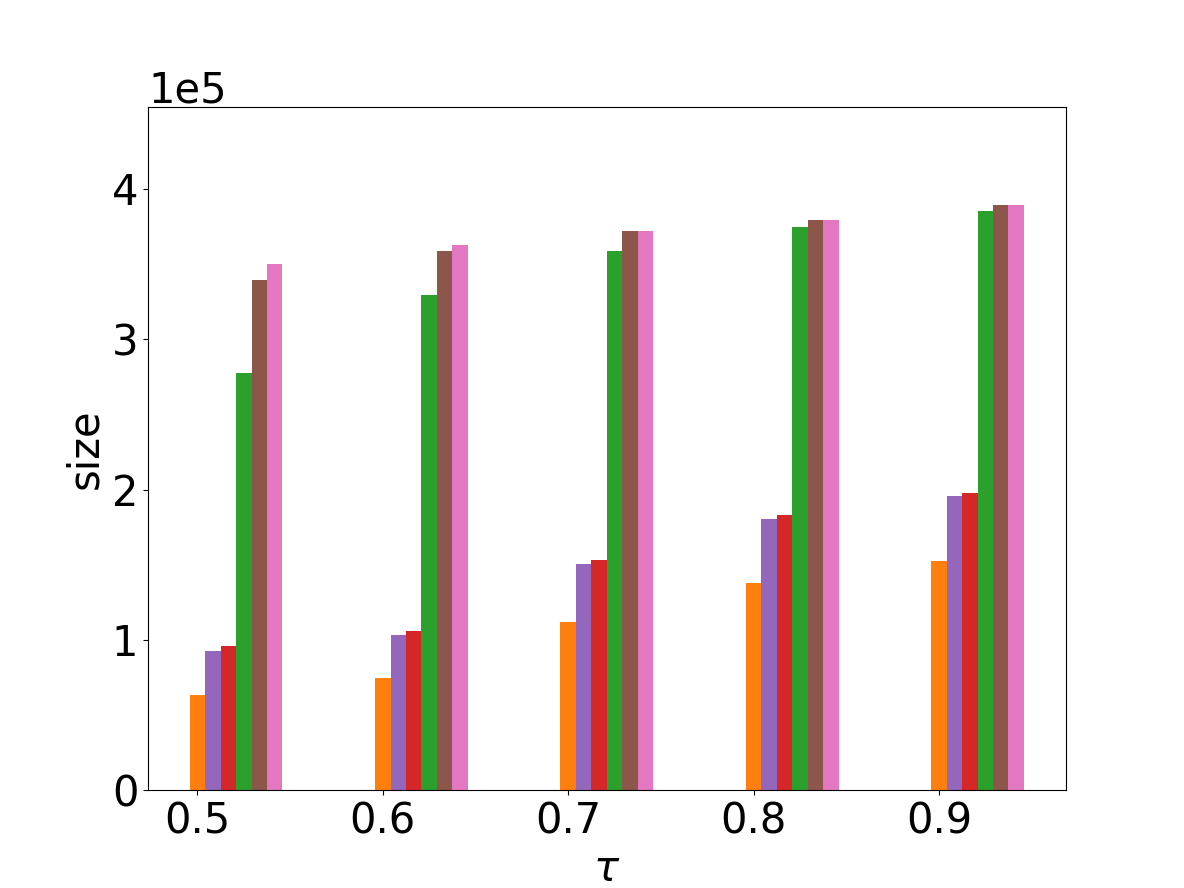}}
\subfloat[email-EuAll  \label{3d}]{\includegraphics[width=4.2cm, height=3.36cm]{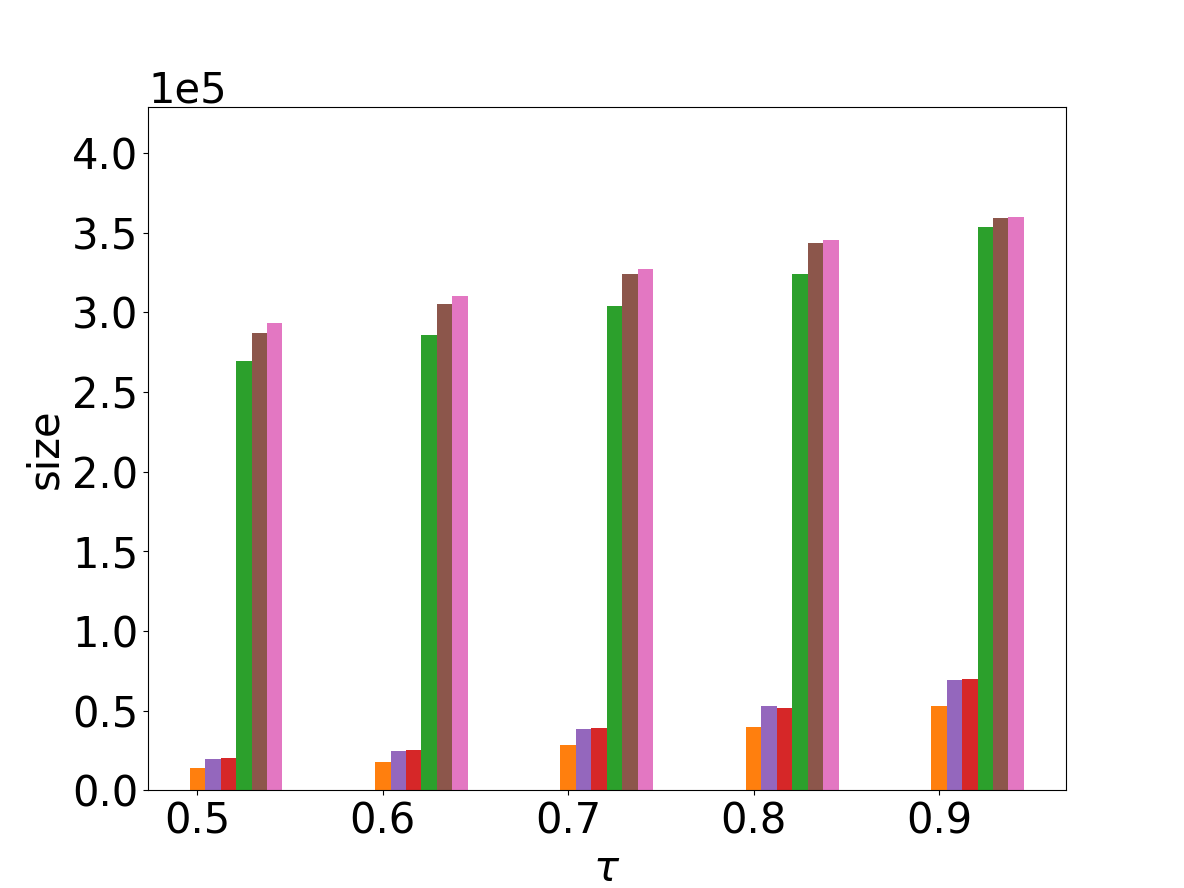}}	
\vspace{-11pt}

\subfloat[web-NotreDame \label{3f}]{\includegraphics[width=4.2cm, height=3.36cm]{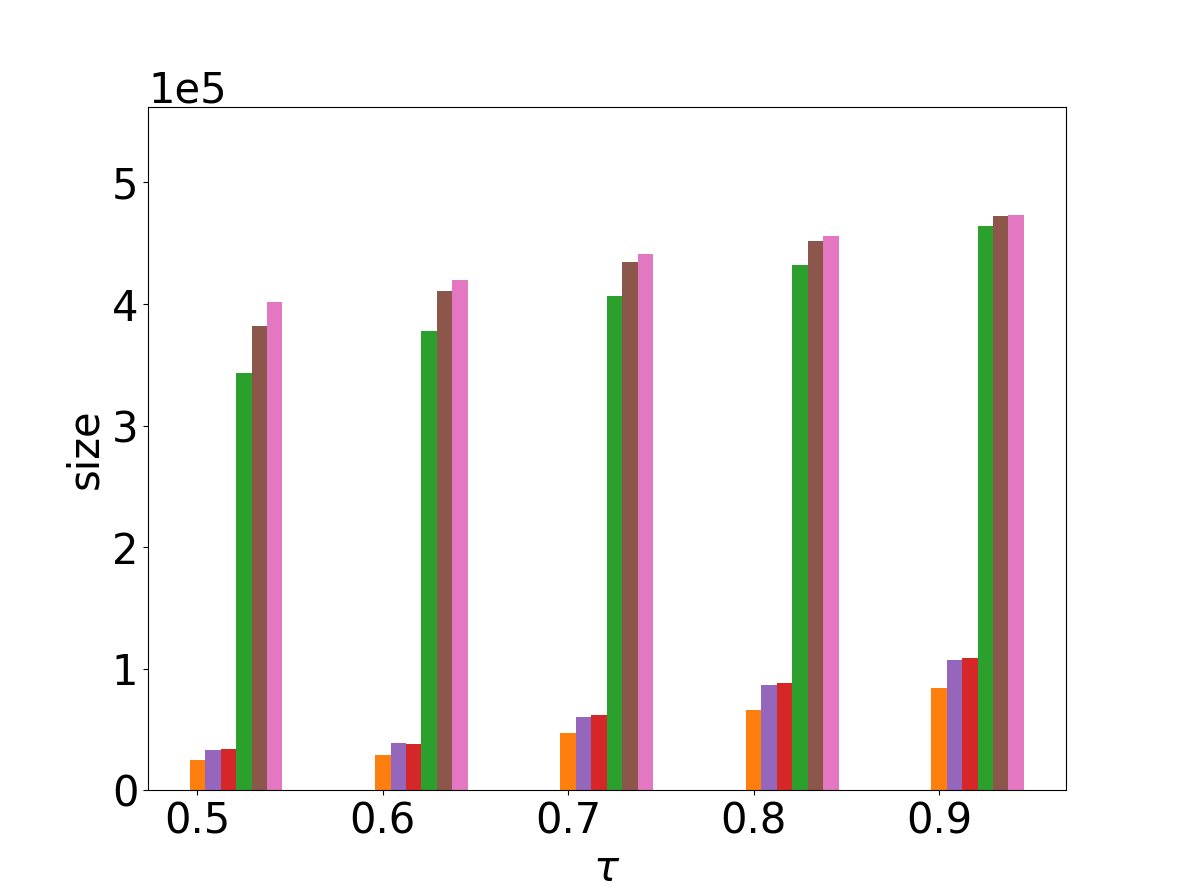}}
\subfloat[com-youtube \label{3g}]{\includegraphics[width=4.2cm, height=3.36cm]{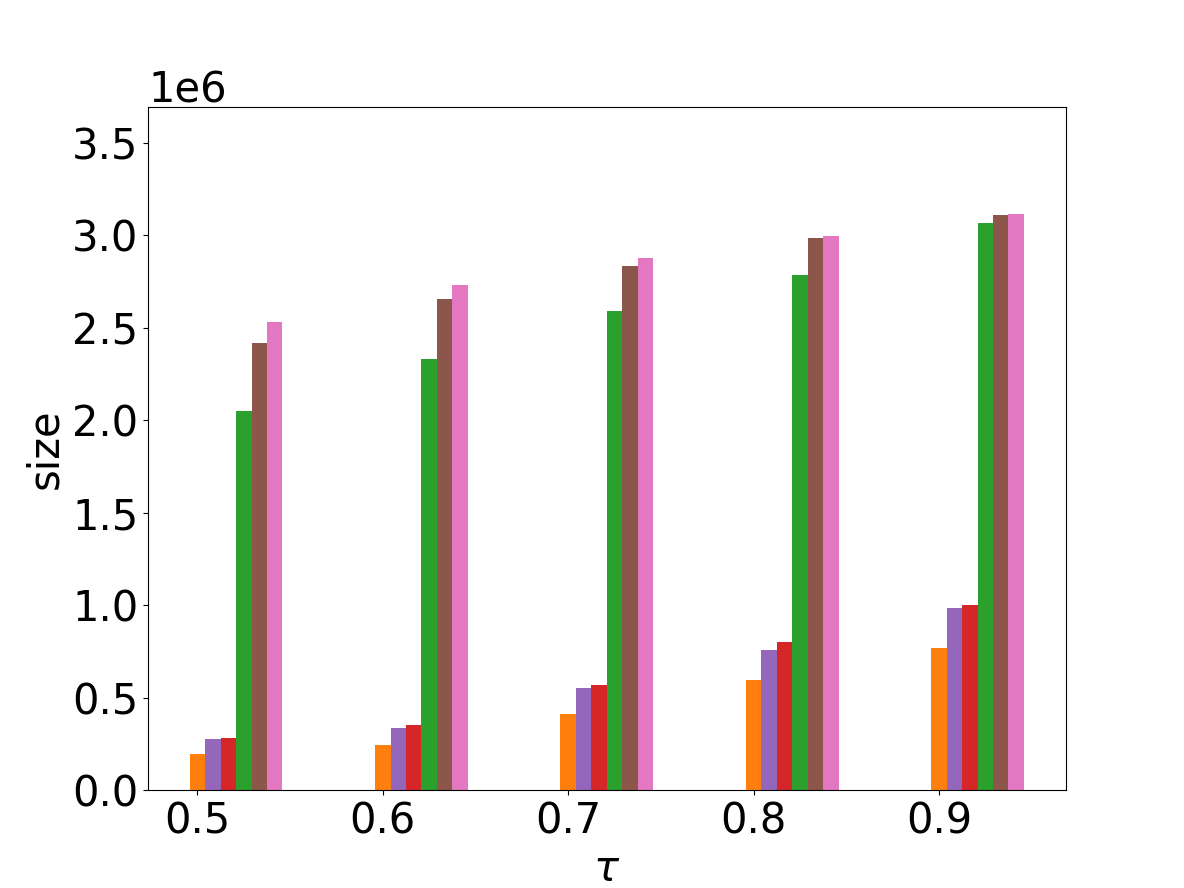}}
\subfloat[soc-pokec \label{3h}]{\includegraphics[width=4.2cm, height=3.36cm]{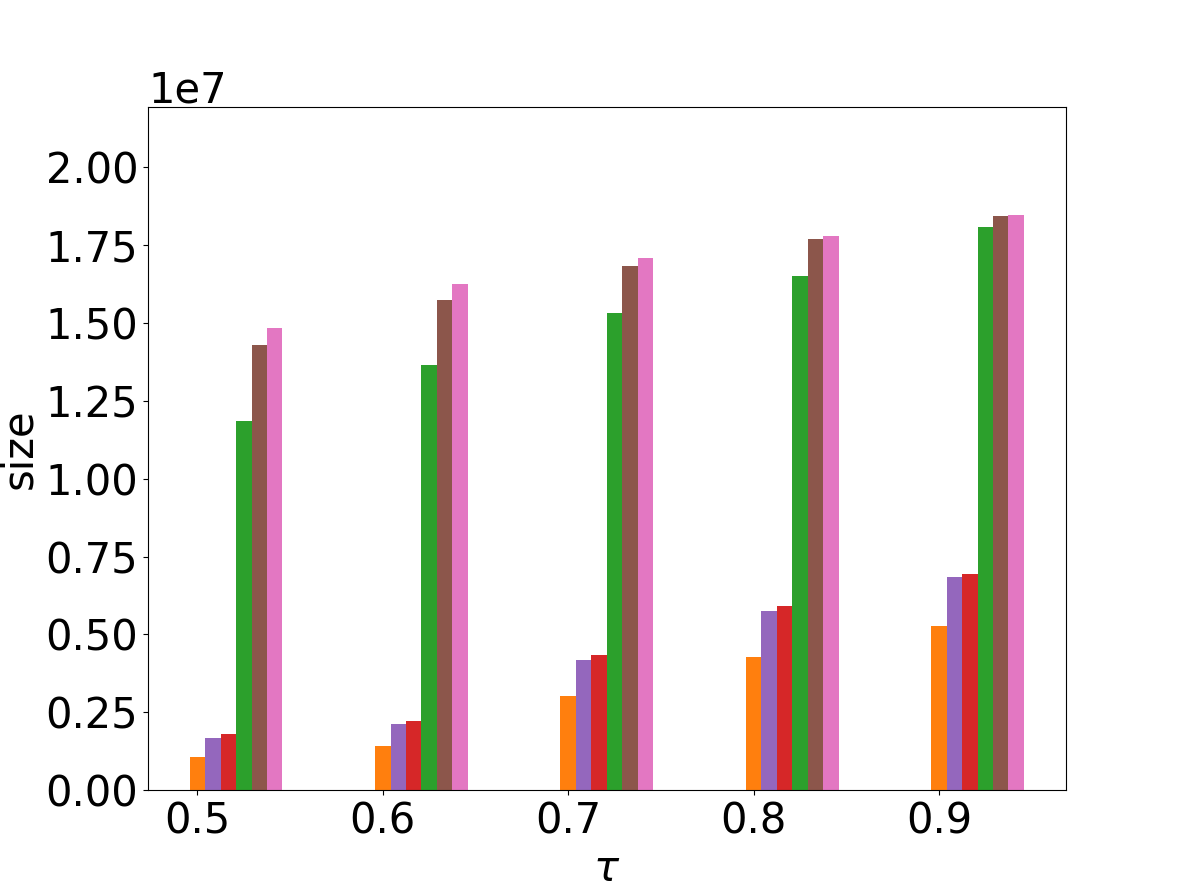}}
\subfloat[cit-Patents \label{3i}]{\includegraphics[width=4.2cm, height=3.36cm]{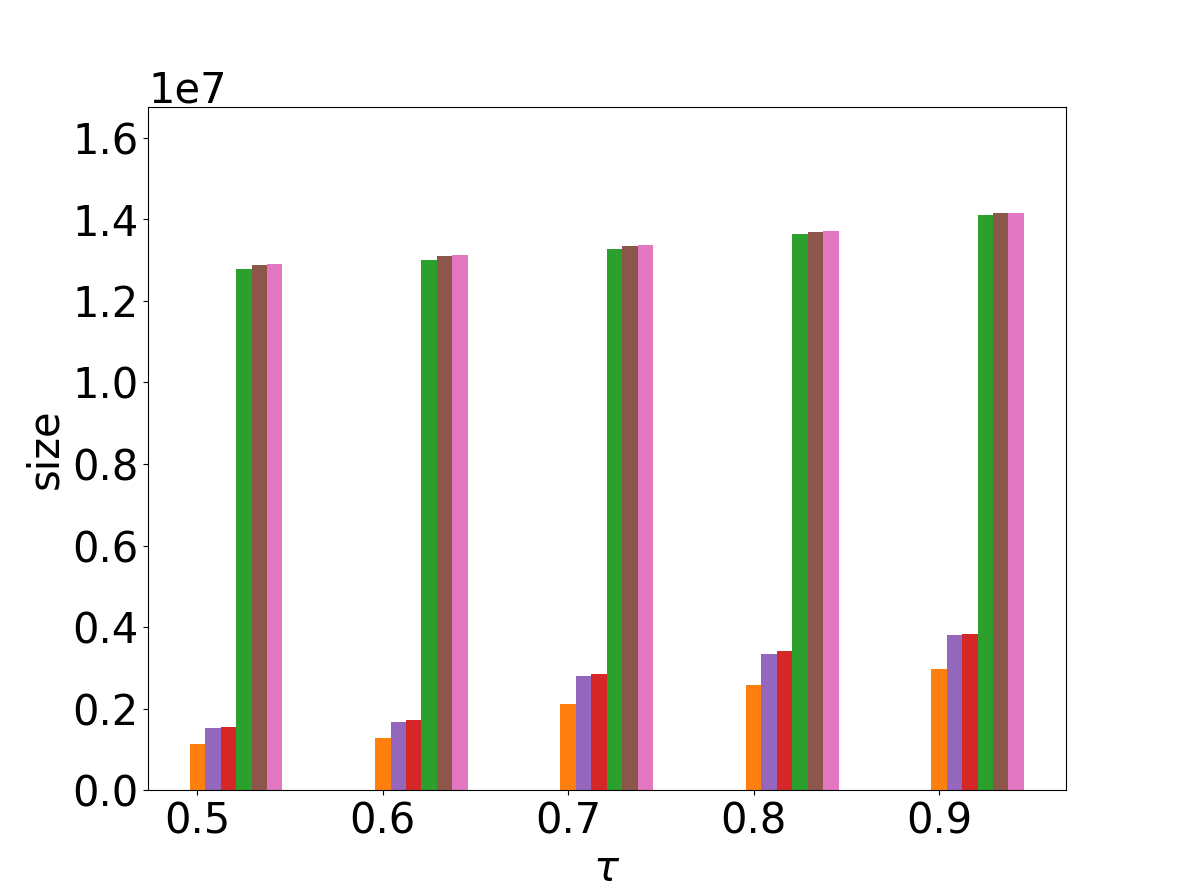}}


	\captionsetup{justification=centering}
	\caption{Summary size of $\tau$-R$^+$MCE and $\tau$-RMCE on eight datasets with different bounds, $\tau$ varies from 0.5 to 0.9, U order as default}\label{fig:sbound}
	\vspace{-10pt}
\end{figure*}

\begin{figure*}[ht]
\vspace{-10pt}
	\centering

	\subfloat[soc-Epinions1 \label{3a}]{\includegraphics[width=4.2cm, height=3.36cm]{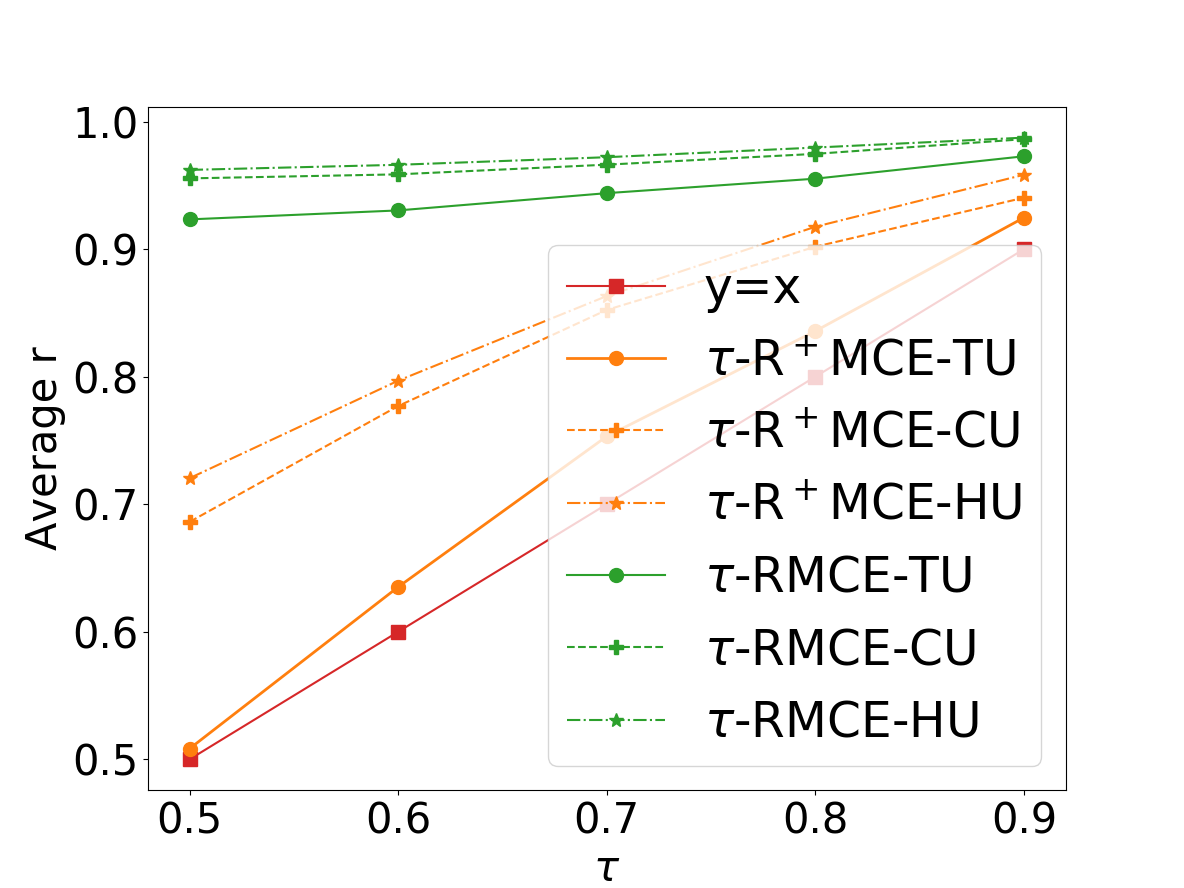}}
	\subfloat[loc-Gowalla \label{3b}]{\includegraphics[width=4.2cm, height=3.36cm]{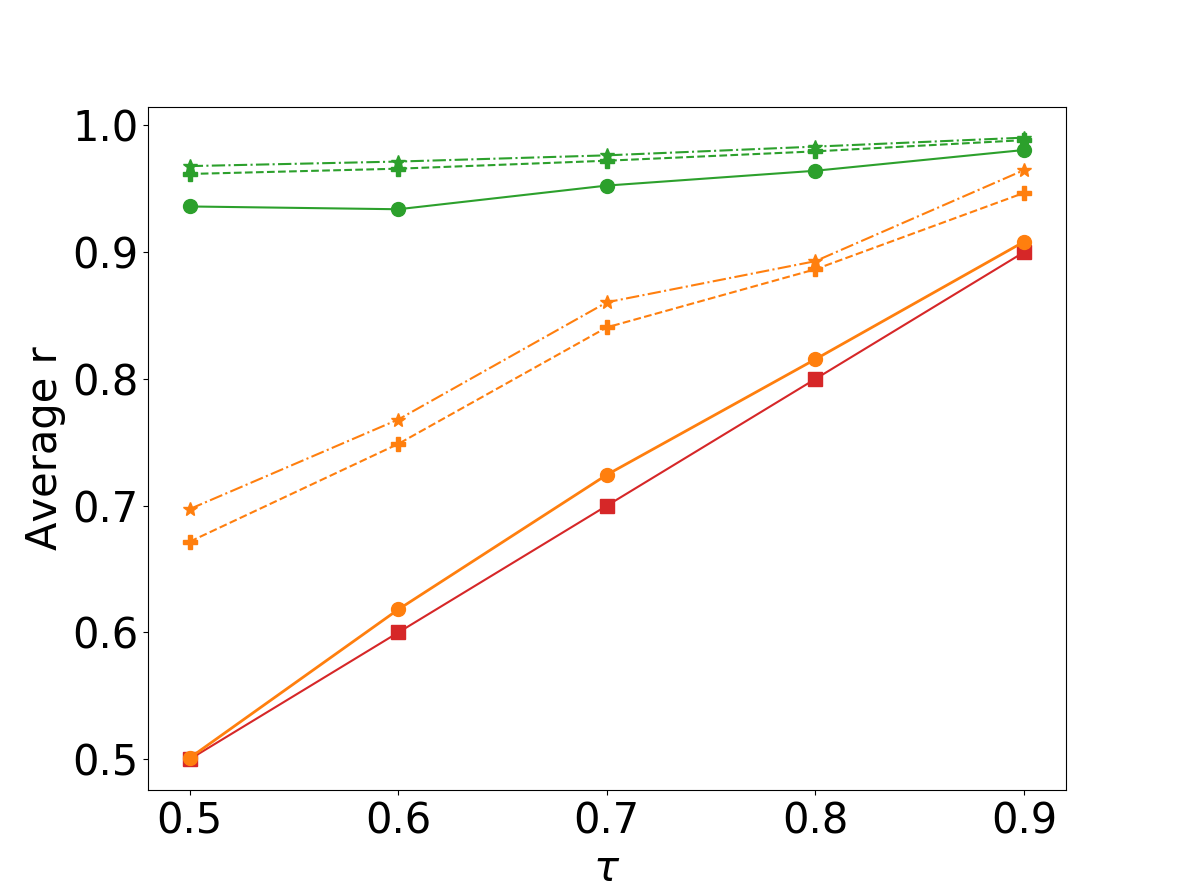}}
	\subfloat[amazon0302 \label{3c}]{\includegraphics[width=4.2cm, height=3.36cm]{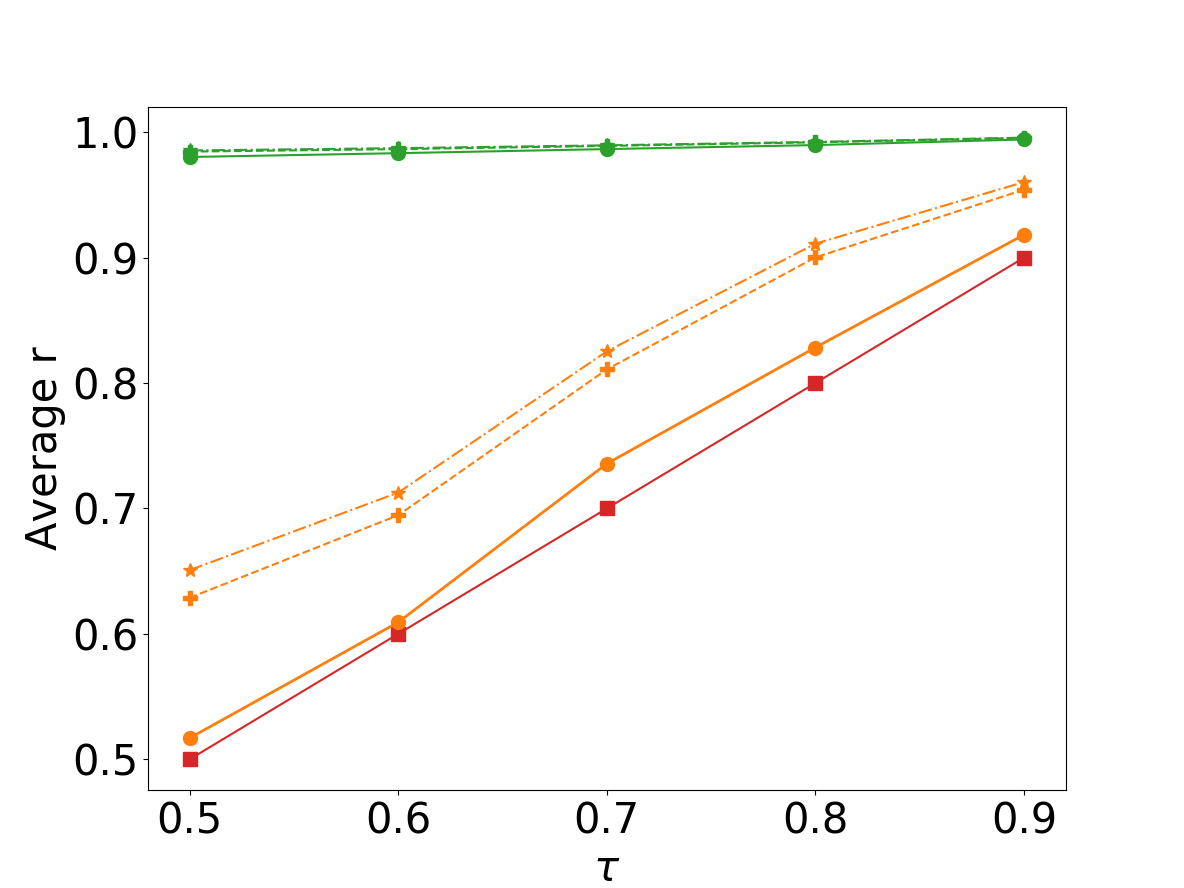}}
\subfloat[email-EuAll  \label{3d}]{\includegraphics[width=4.2cm, height=3.36cm]{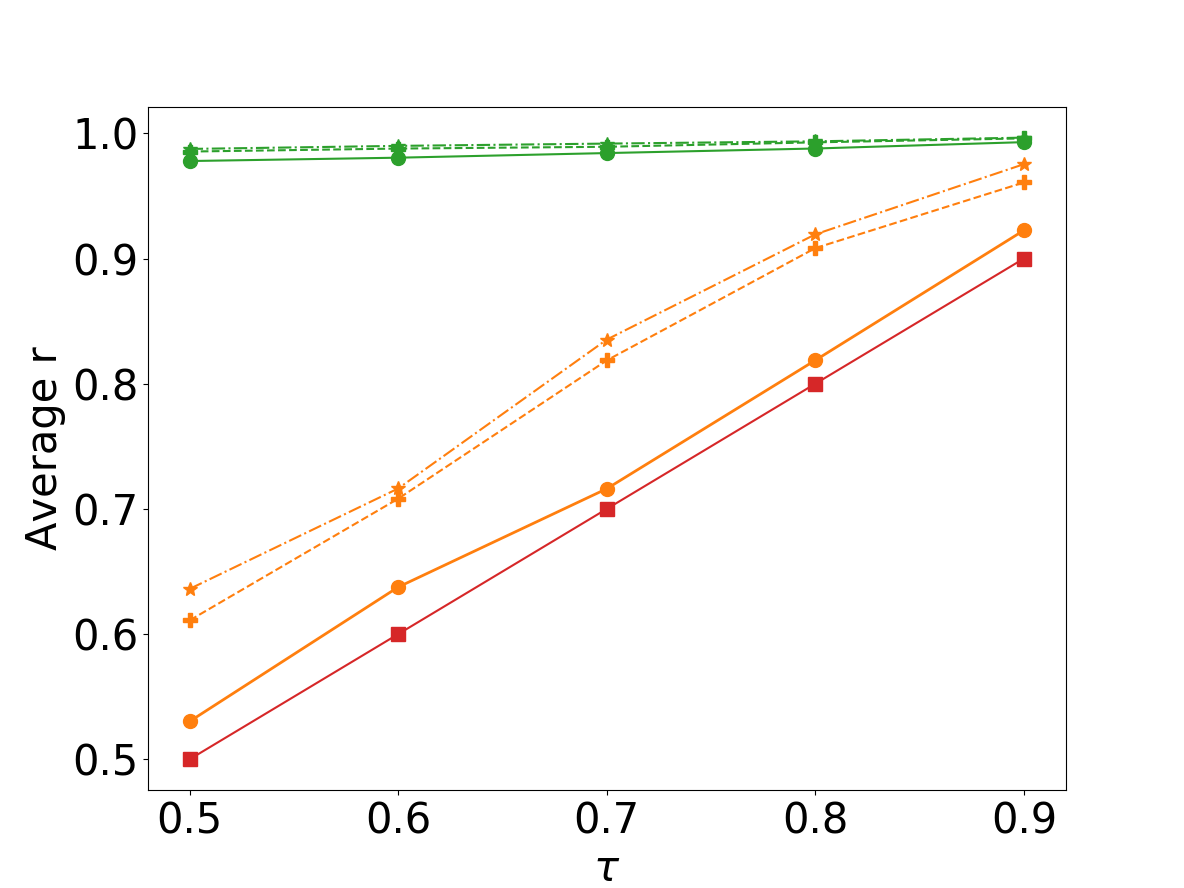}}	
\vspace{-11pt}

\subfloat[web-NotreDame \label{3f}]{\includegraphics[width=4.2cm, height=3.36cm]{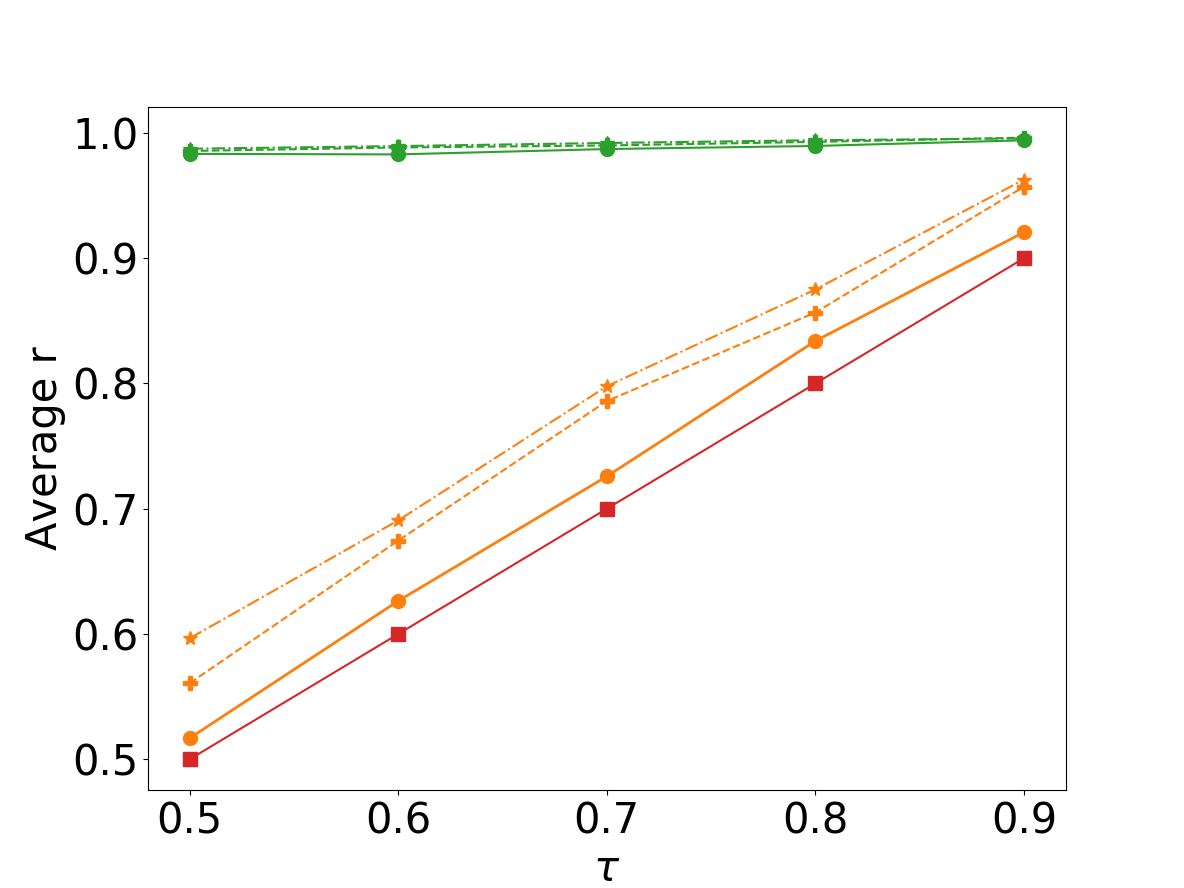}}
\subfloat[com-youtube \label{3g}]{\includegraphics[width=4.2cm, height=3.36cm]{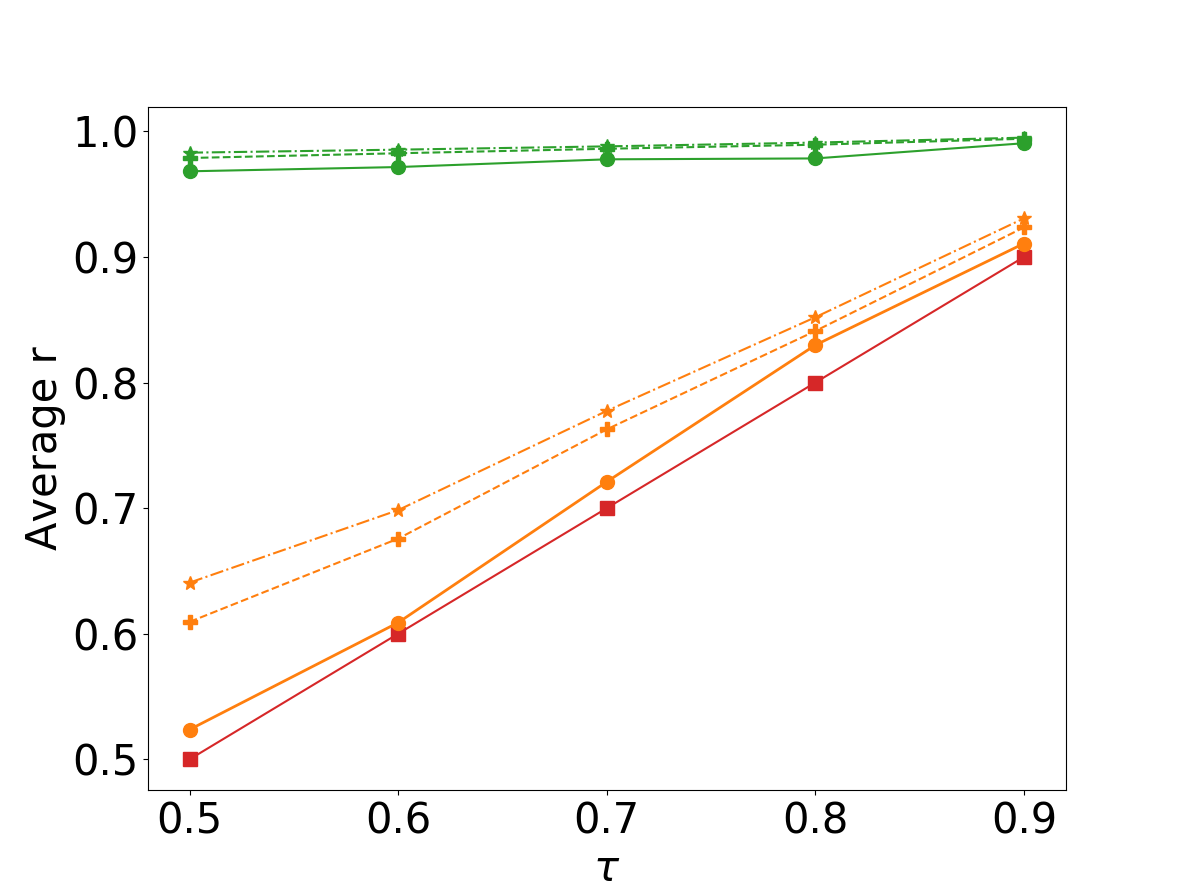}}
\subfloat[soc-pokec \label{3h}]{\includegraphics[width=4.2cm, height=3.36cm]{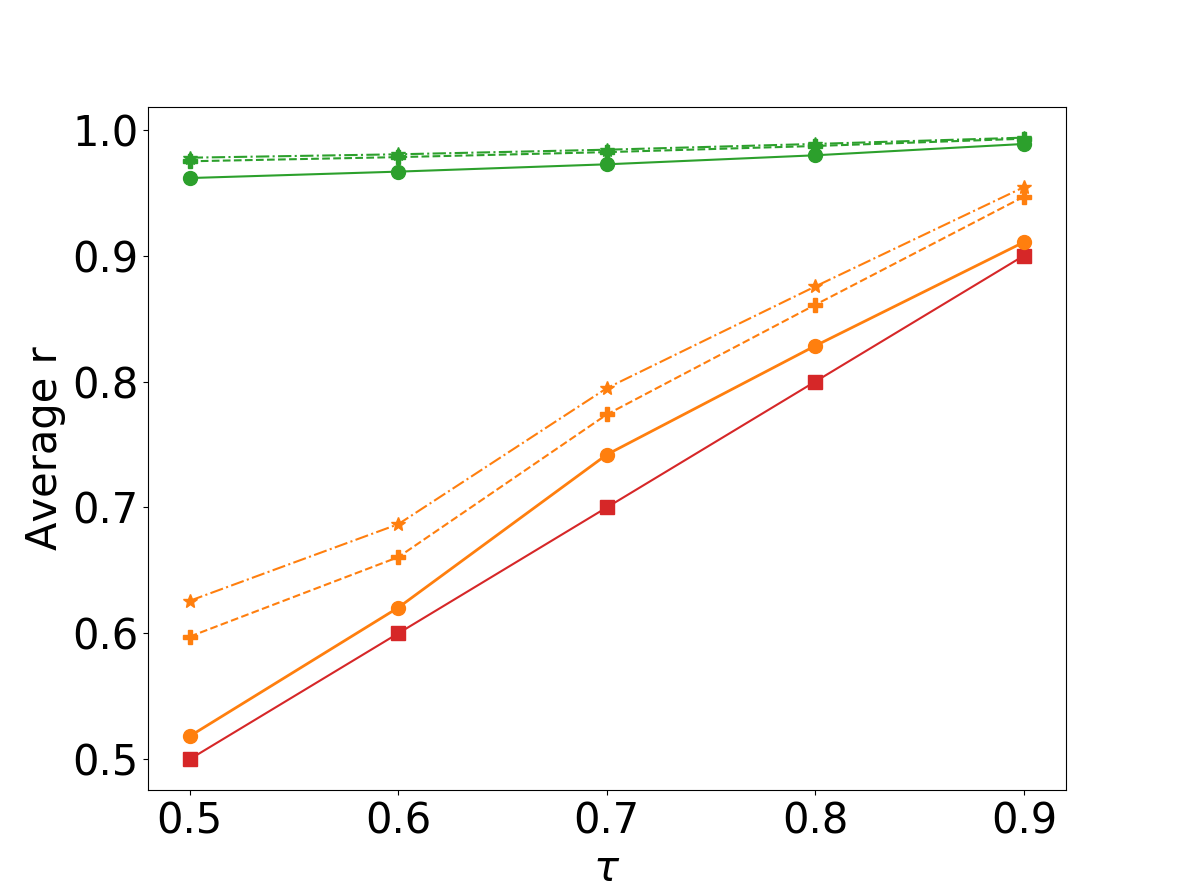}}
\subfloat[cit-Patents \label{3i}]{\includegraphics[width=4.2cm, height=3.36cm]{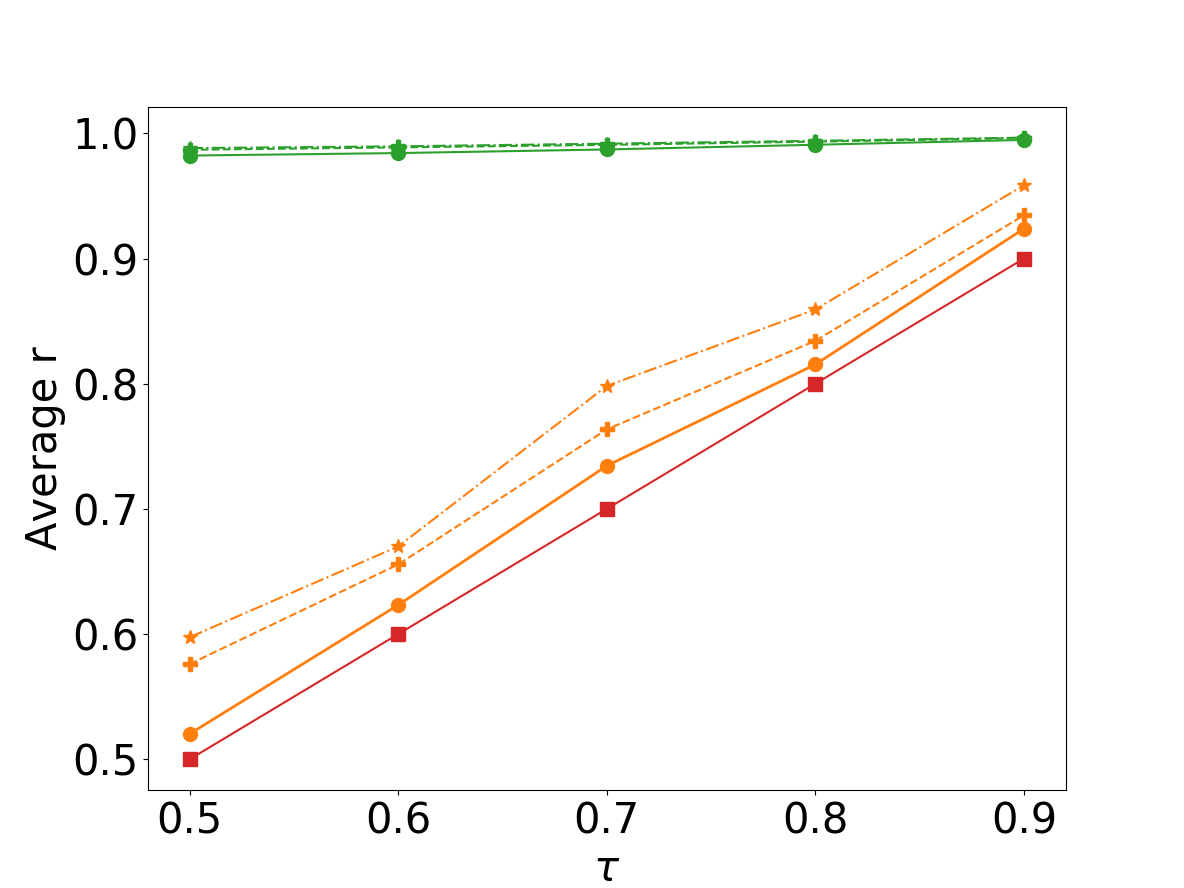}}


	\captionsetup{justification=centering}
	\caption{Average $r$ of $\tau$-R$^+$MCE and $\tau$-RMCE on eight datasets with different bounds, $\tau$ varies from 0.5 to 0.9, U order as default}\label{fig:BoundTau}
	\vspace{-10pt}
\end{figure*}

\begin{figure*}[ht]
\vspace{-12pt}
	\centering

	\subfloat[soc-Epinions1 \label{3a}]{\includegraphics[width=4.2cm, height=3.36cm]{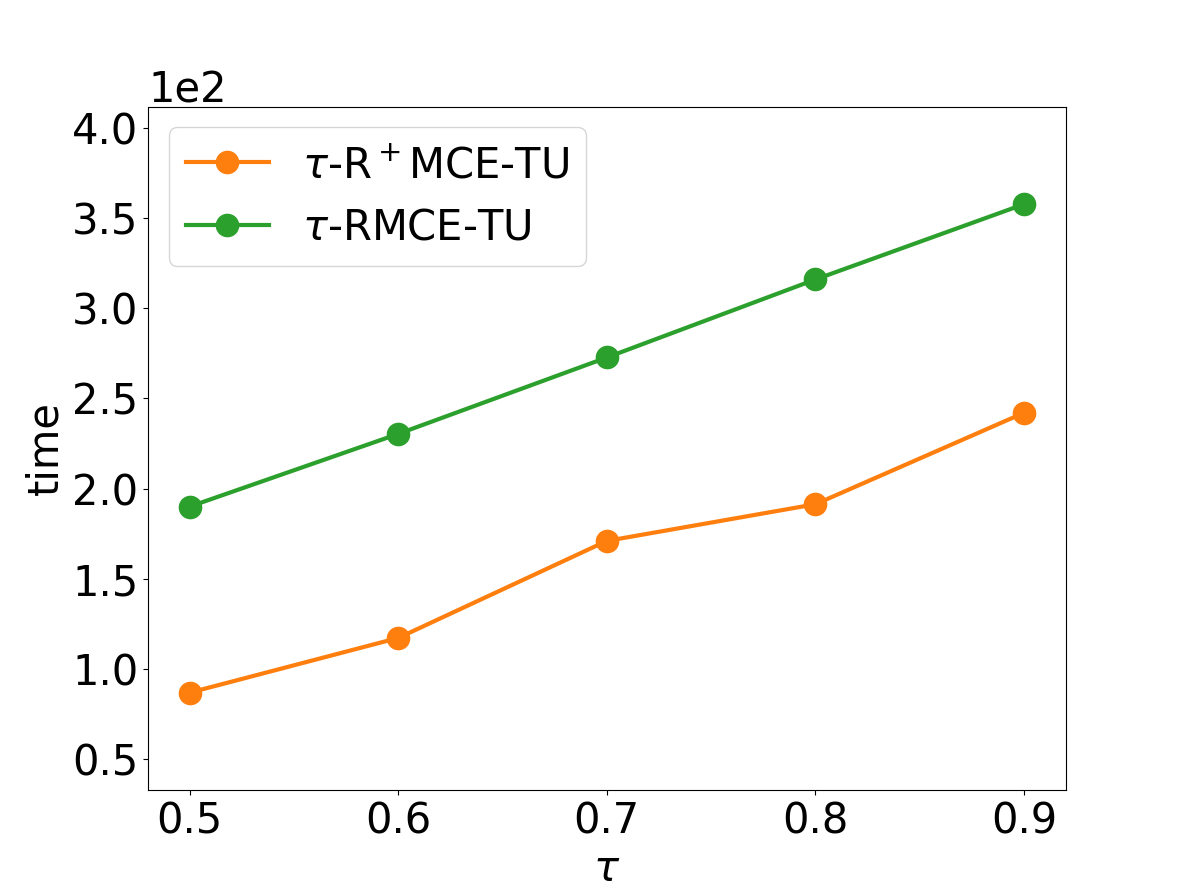}}
	\subfloat[loc-Gowalla \label{3b}]{\includegraphics[width=4.2cm, height=3.36cm]{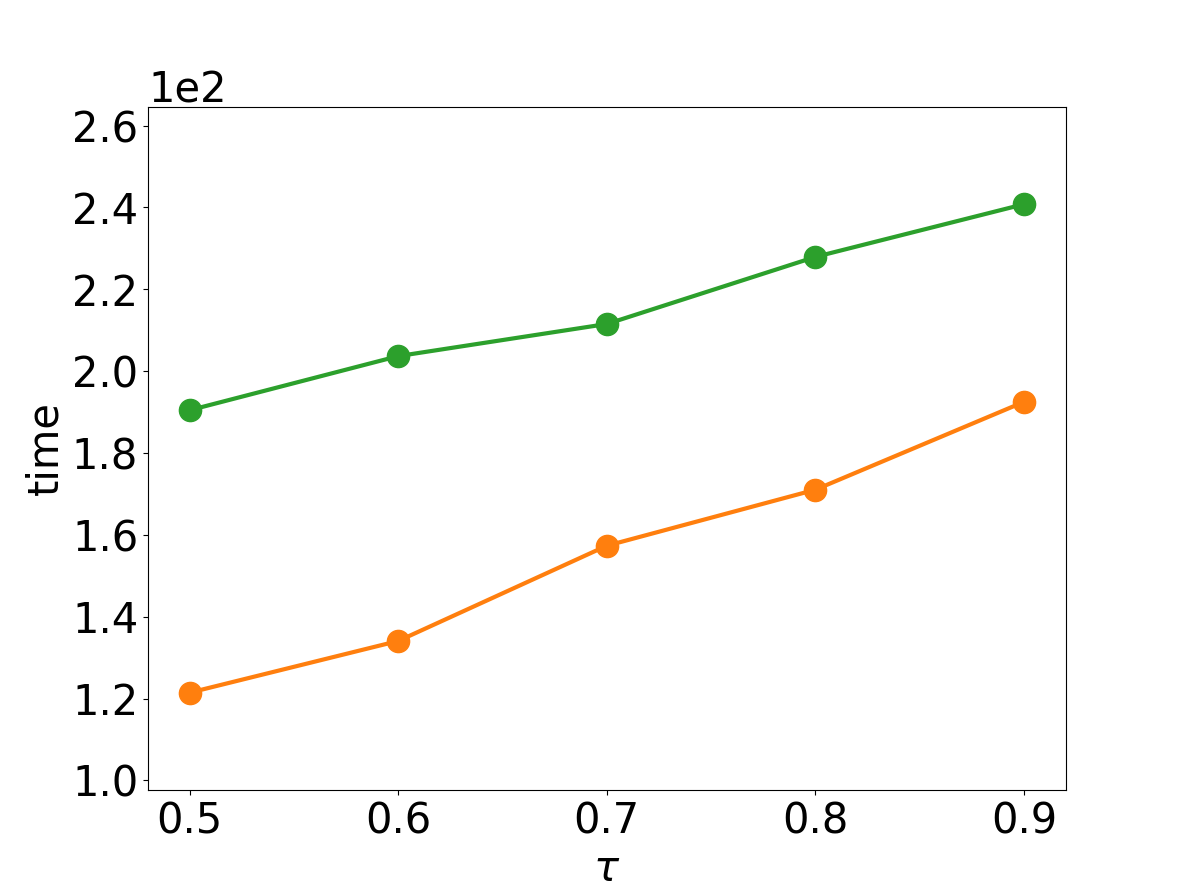}}
	\subfloat[amazon0302 \label{3c}]{\includegraphics[width=4.2cm, height=3.36cm]{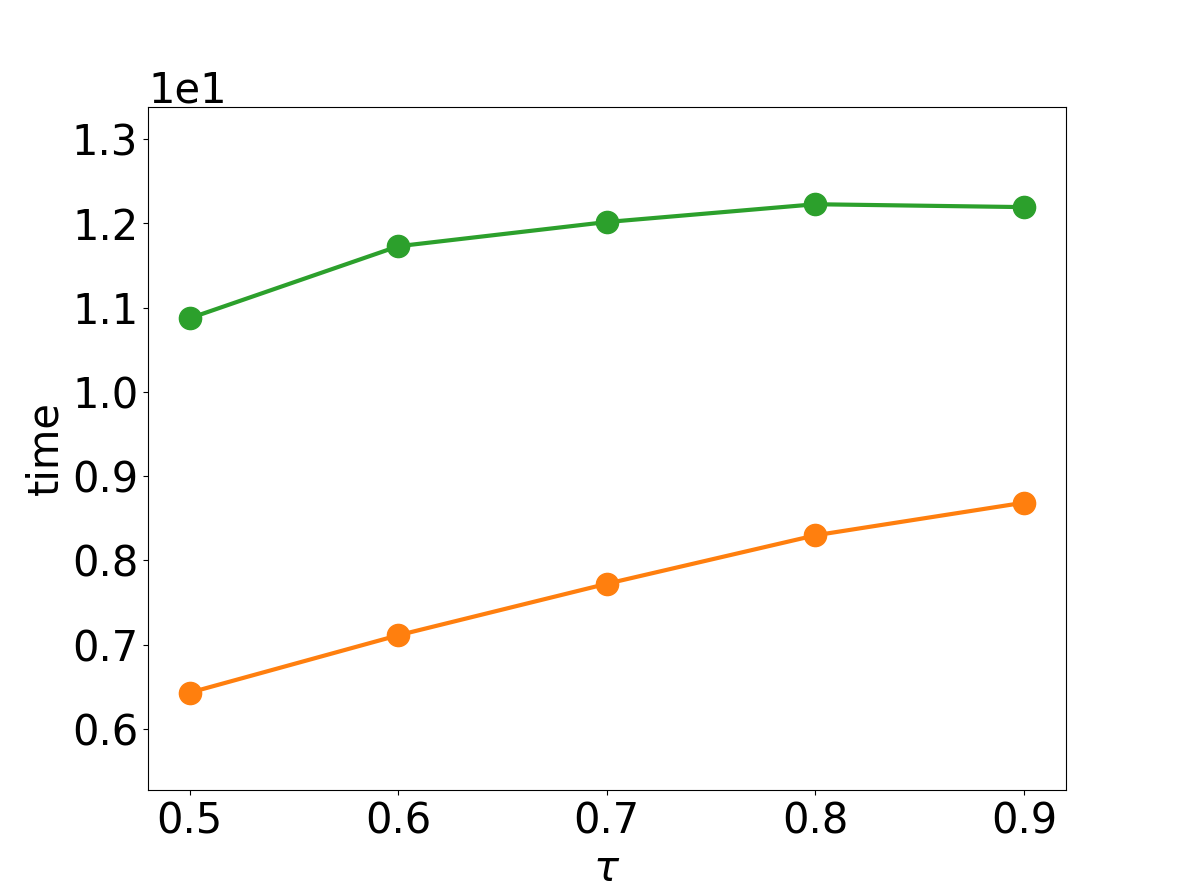}}
\subfloat[email-EuAll  \label{3d}]{\includegraphics[width=4.2cm, height=3.36cm]{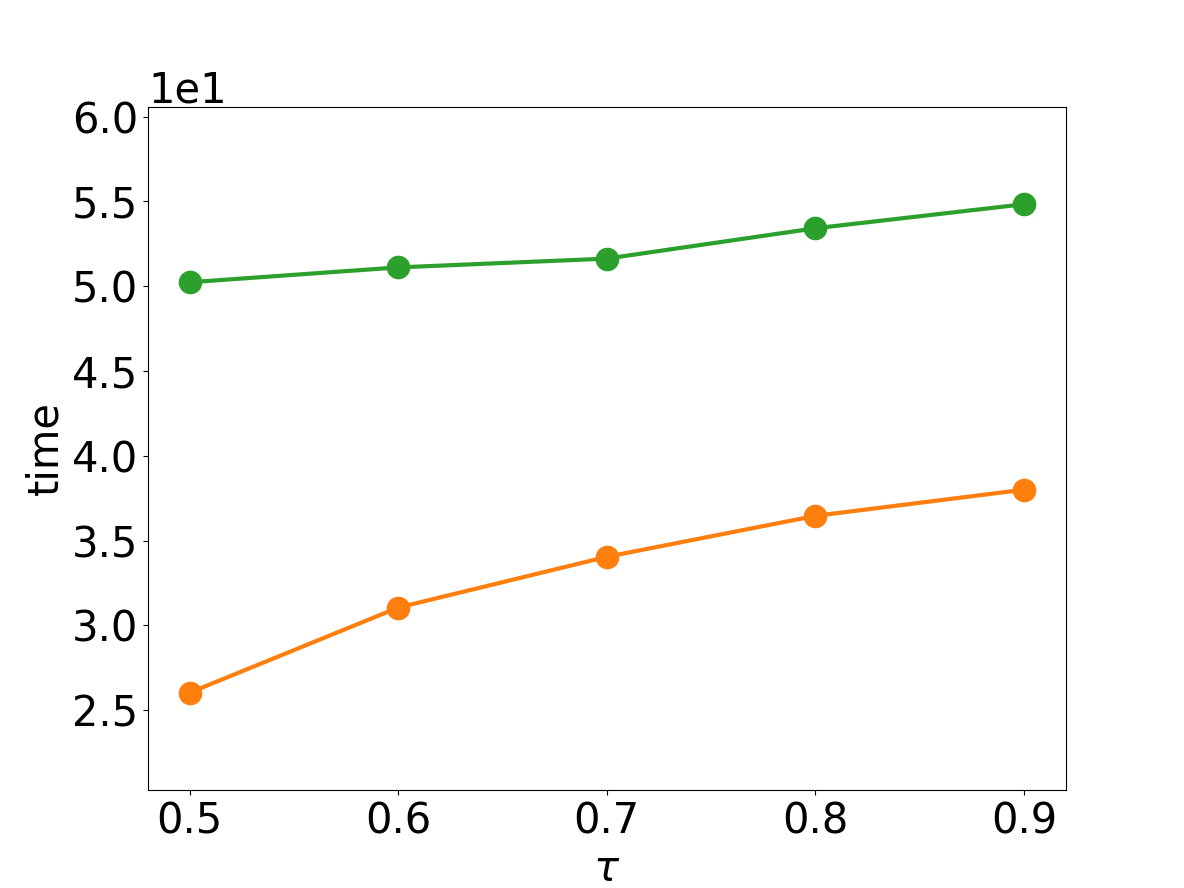}}	

\vspace{-11pt}

\subfloat[web-NotreDame \label{3f}]{\includegraphics[width=4.2cm, height=3.36cm]{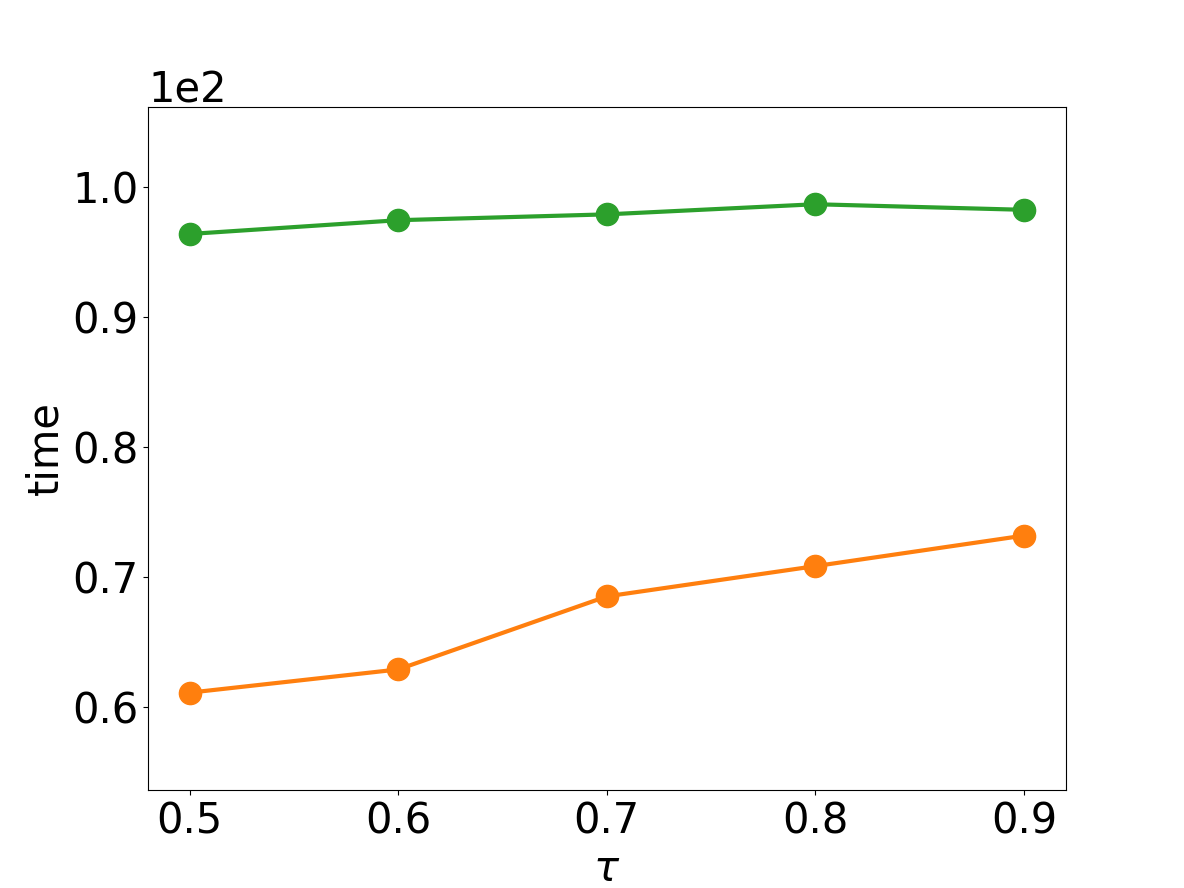}}
\subfloat[com-youtube \label{3g}]{\includegraphics[width=4.2cm, height=3.36cm]{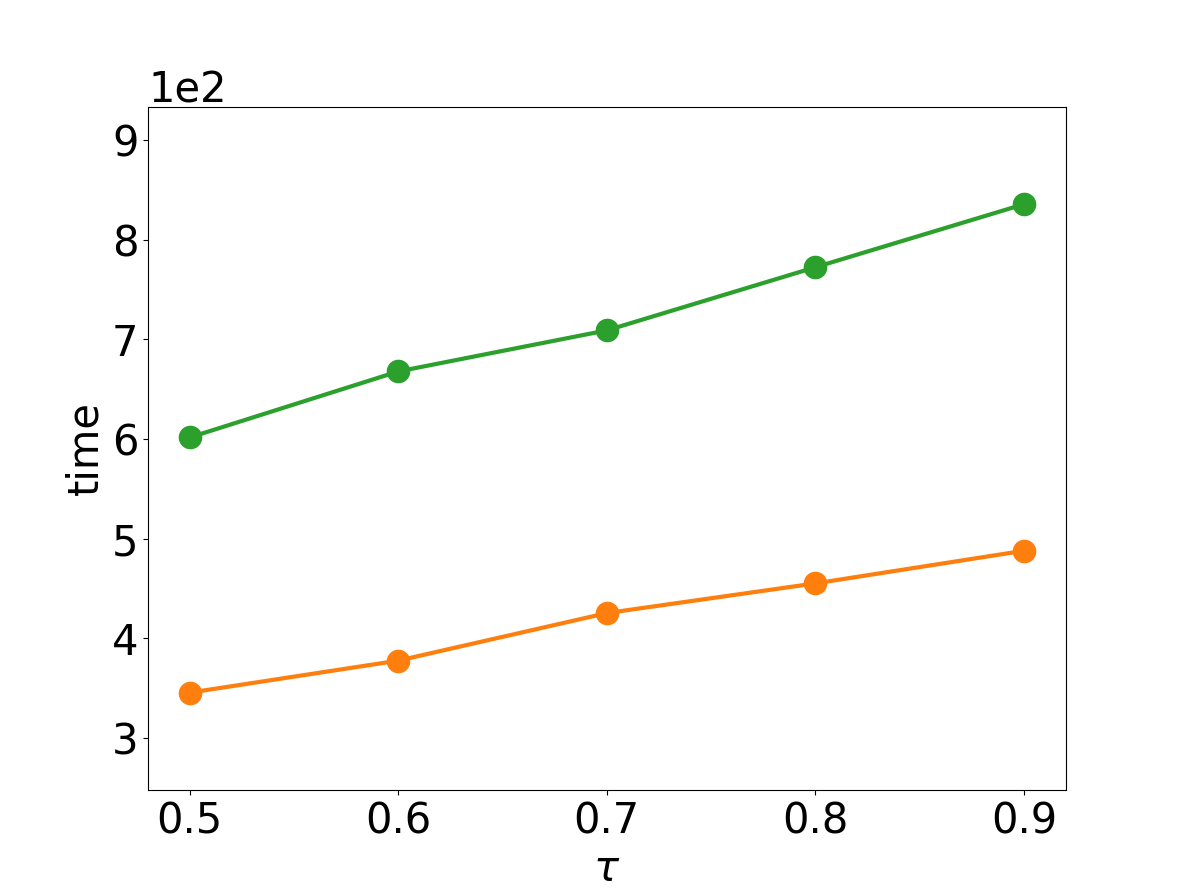}}
\subfloat[soc-pokec \label{3h}]{\includegraphics[width=4.2cm, height=3.36cm]{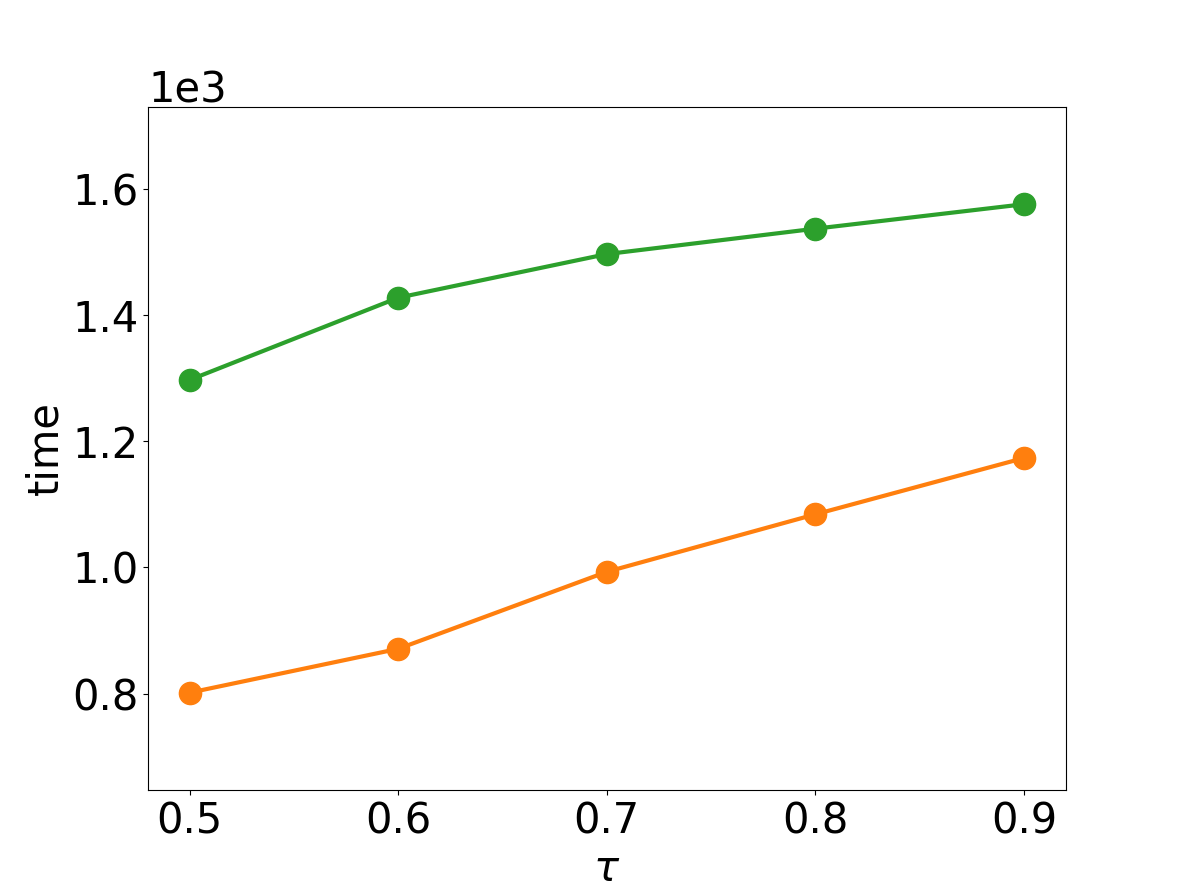}}
\subfloat[cit-Patents \label{3i}]{\includegraphics[width=4.2cm, height=3.36cm]{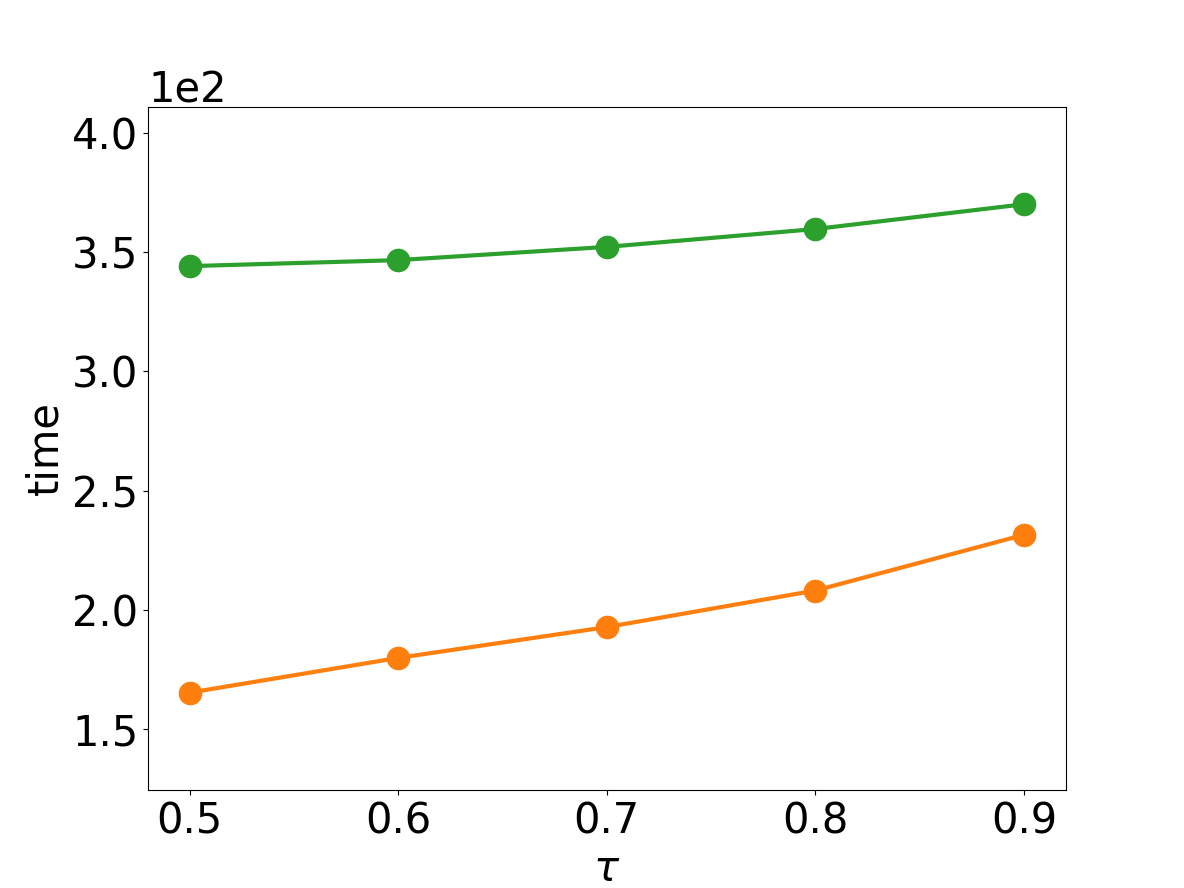}}


	\captionsetup{justification=centering}
	\caption{\small{Running time of $\tau$-R$^+$MCE and $\tau$-RMCE on eight datasets, $\tau$ varies from 0.5 to 0.9, T bound and U order as default}}\label{fig:ttotal}
	\vspace{-10pt}
\end{figure*}

{ 
\subsection{Efficiency}
\label{sec:time}

{
While our main concern in this paper is the output size,
the efficiency of $\tau$-RMCE and $\tau$-R$^+$MCE (with three types of bounds and orders) is also  reported. 
To provide a fuller discussion of the efficiency, we plotted both the total running time and the memory requirement. 
\subsubsection{Running time}
We compare the total running time of $\tau$-R$^+$MCE and  $\tau$-RMCE with default setting of U and T (see Fig.~\ref{fig:ttotal}). 
Results show that $\tau$-R$^+$MCE consistently surpasses $\tau$-RMCE on eight datasets. 
When $\tau = 0.9$, the time reduction is more than $20\%$ for all datasets, among which four datasets ({\it soc-Epinions1, email-EuAll, com-youtube, cit-Patents}) achieve $30\%$. 
When $\tau = 0.5$, this percentage exceeds $35\%$ for all datasets, and five of them ({\it soc-Epinions1, amazon0302, email-EuAll, com-youtube, cit-Patents}) 
achieve more than $40\%$. 

To get a full understanding of why our proposed method benefits efficiency 
(although our initial purpose is to target the effectiveness), 
we recorded the first-result time, 
that is, the duration from the beginning to the first maximal clique being included into summary. 
We found that the result varies very little for different  bounds and vertex orders. 
Thus we use Table~\ref{frt} to briefly summarize the results (T and U are set as default, $\tau=0.9$). 
The first-result time takes up only a very small proportion (less than $7\%$) of the total running time, and this holds for both $\tau$-R$^+$MCE and  $\tau$-RMCE on eight datasets with all the $\tau$ values. 
The fact is that most of the running time (more than $93\%$) is consumed by the enumeration procedure,  
which implies that the benefited efficiency of  $\tau$-R$^+$MCE comes from the early pruning power that speeds up the enumeration recursion.  
The search tree of $\tau$-R$^+$MCE does not have to explore as deep as $\tau$-RMCE does to finally determine whether to discard a candidate clique, thus
less  time is wasted on growing cliques that would result in redundancy. 

\begin{figure*}[ht]
\vspace{-15pt}
	\centering

	\subfloat[soc-Epinions1 \label{3a}]{\includegraphics[width=4.2cm, height=3.36cm]{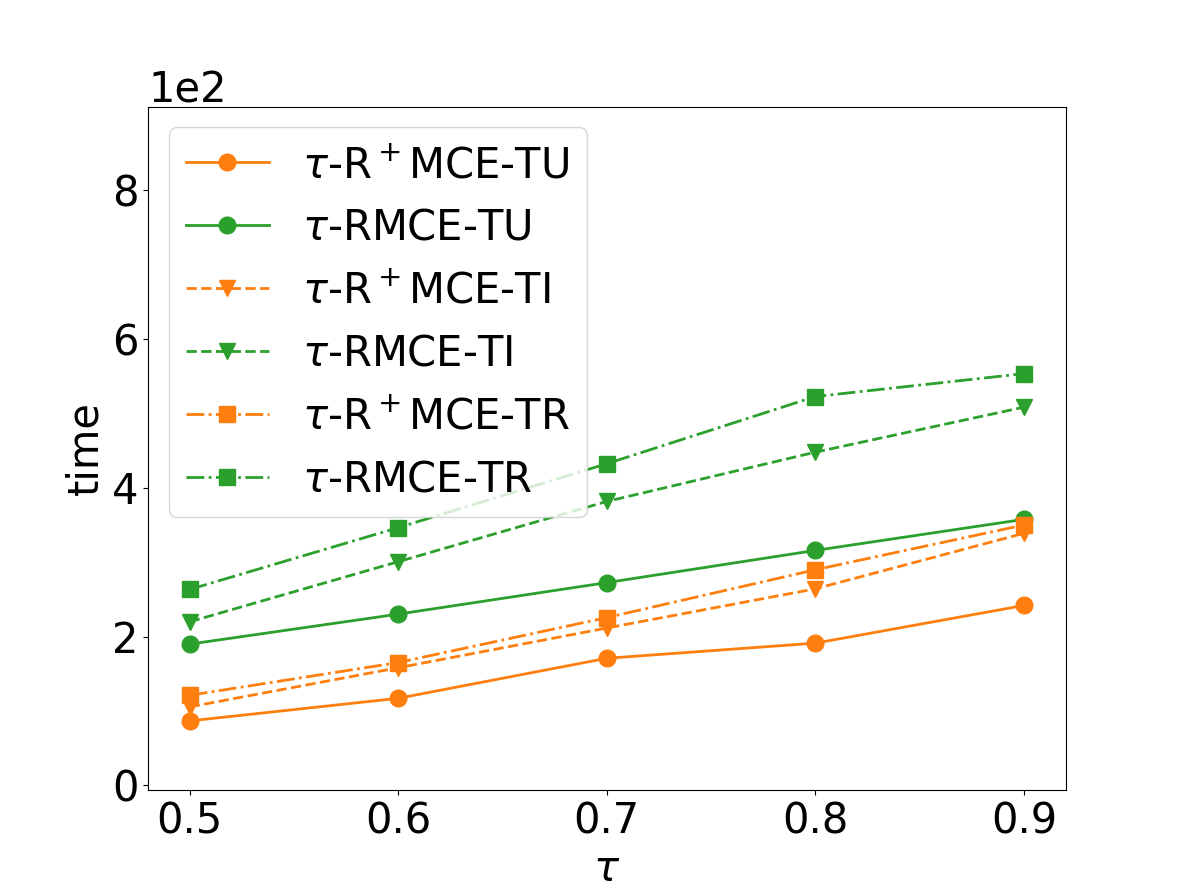}}
	\subfloat[loc-Gowalla \label{3b}]{\includegraphics[width=4.2cm, height=3.36cm]{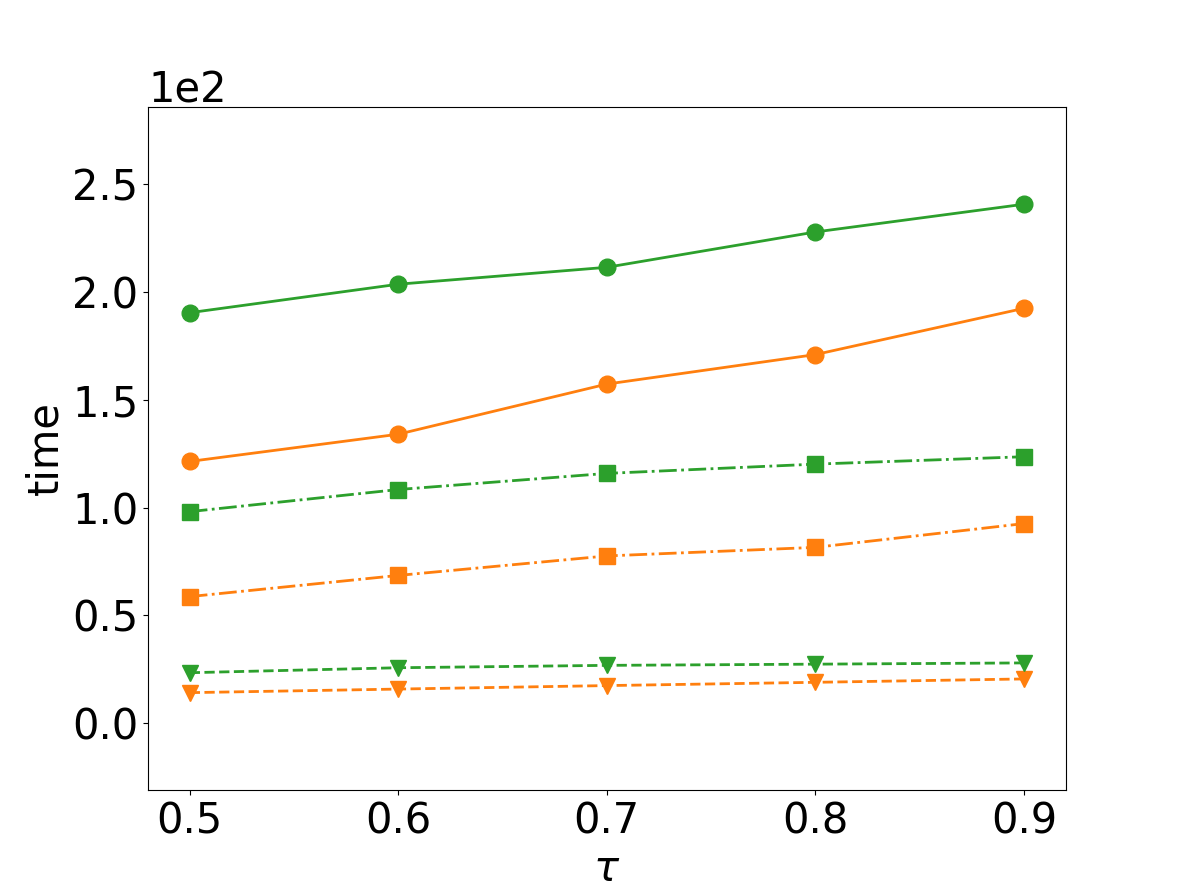}}
	\subfloat[amazon0302 \label{3c}]{\includegraphics[width=4.2cm, height=3.36cm]{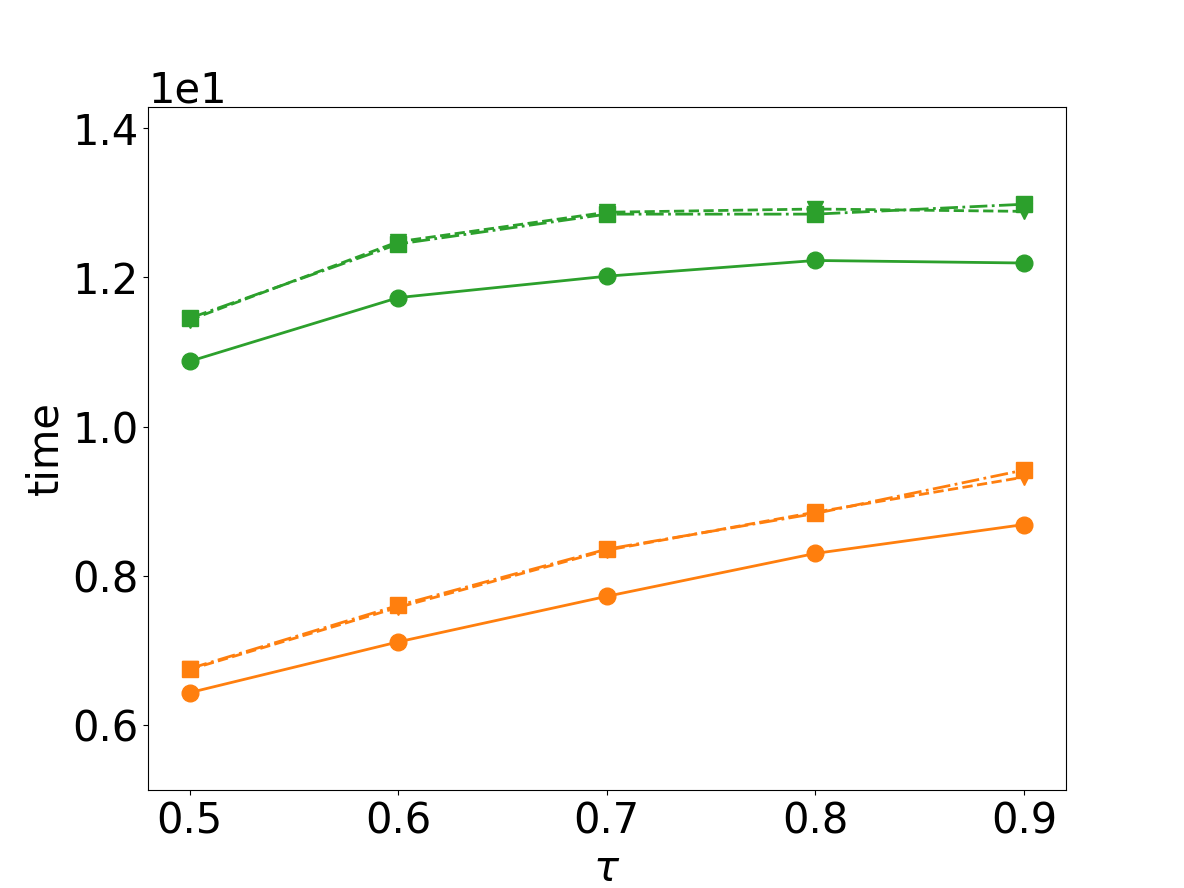}}
\subfloat[email-EuAll  \label{3d}]{\includegraphics[width=4.2cm, height=3.36cm]{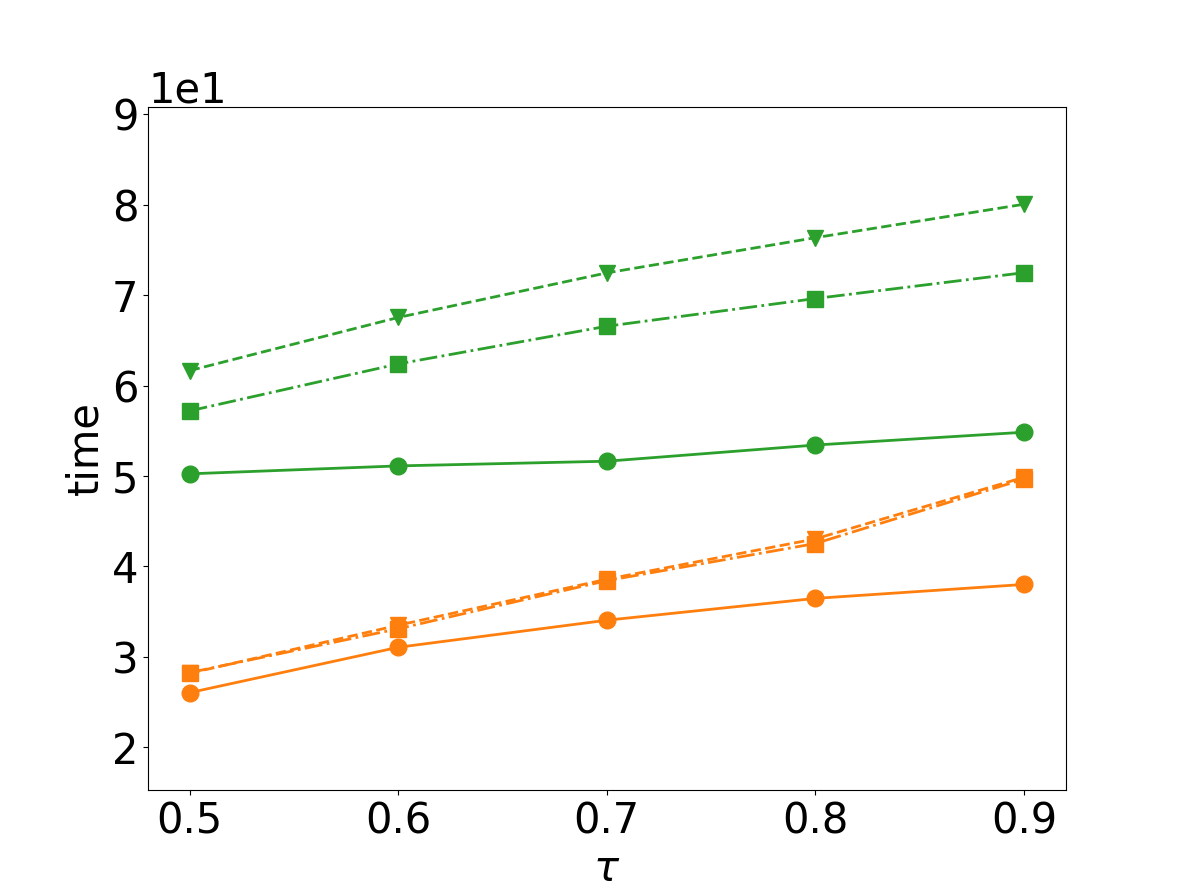}}	

\vspace{-11pt}

\subfloat[web-NotreDame \label{3f}]{\includegraphics[width=4.2cm, height=3.36cm]{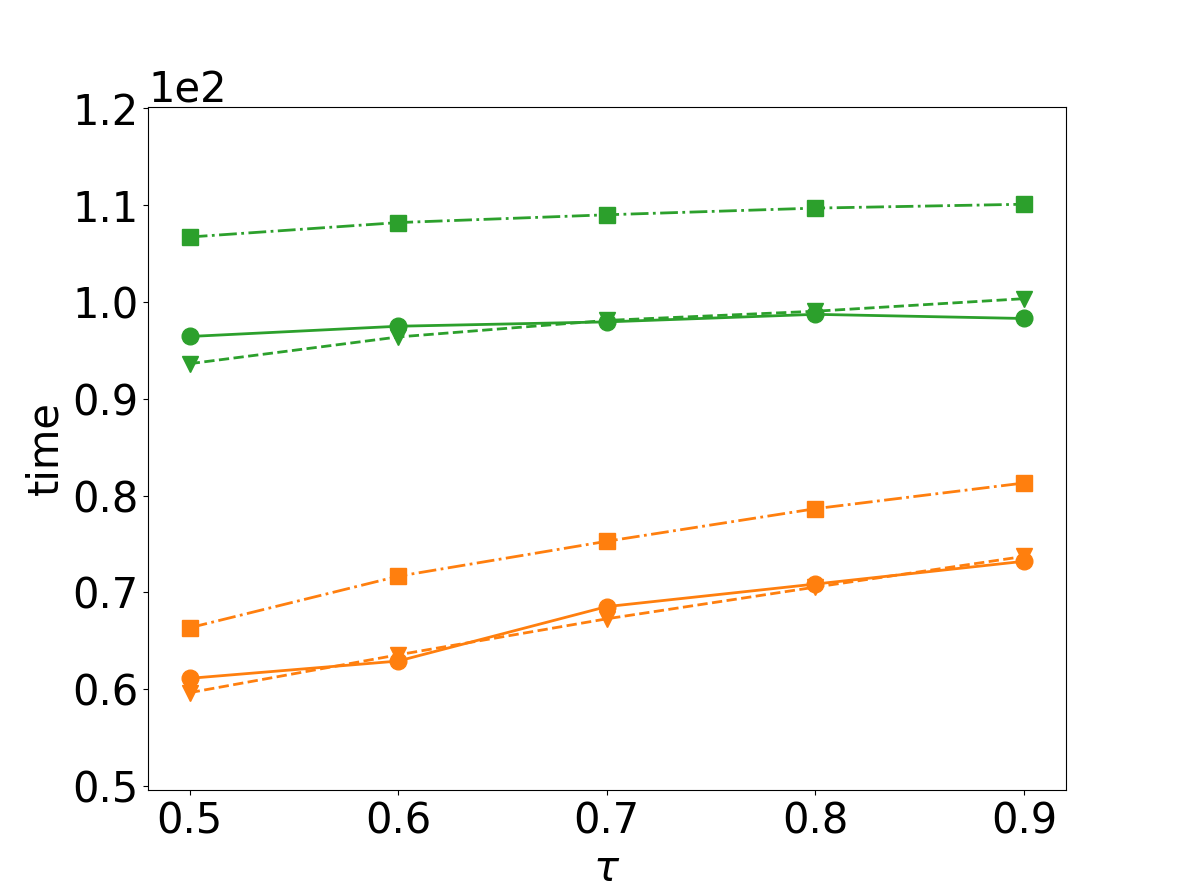}}
\subfloat[com-youtube \label{3g}]{\includegraphics[width=4.2cm, height=3.36cm]{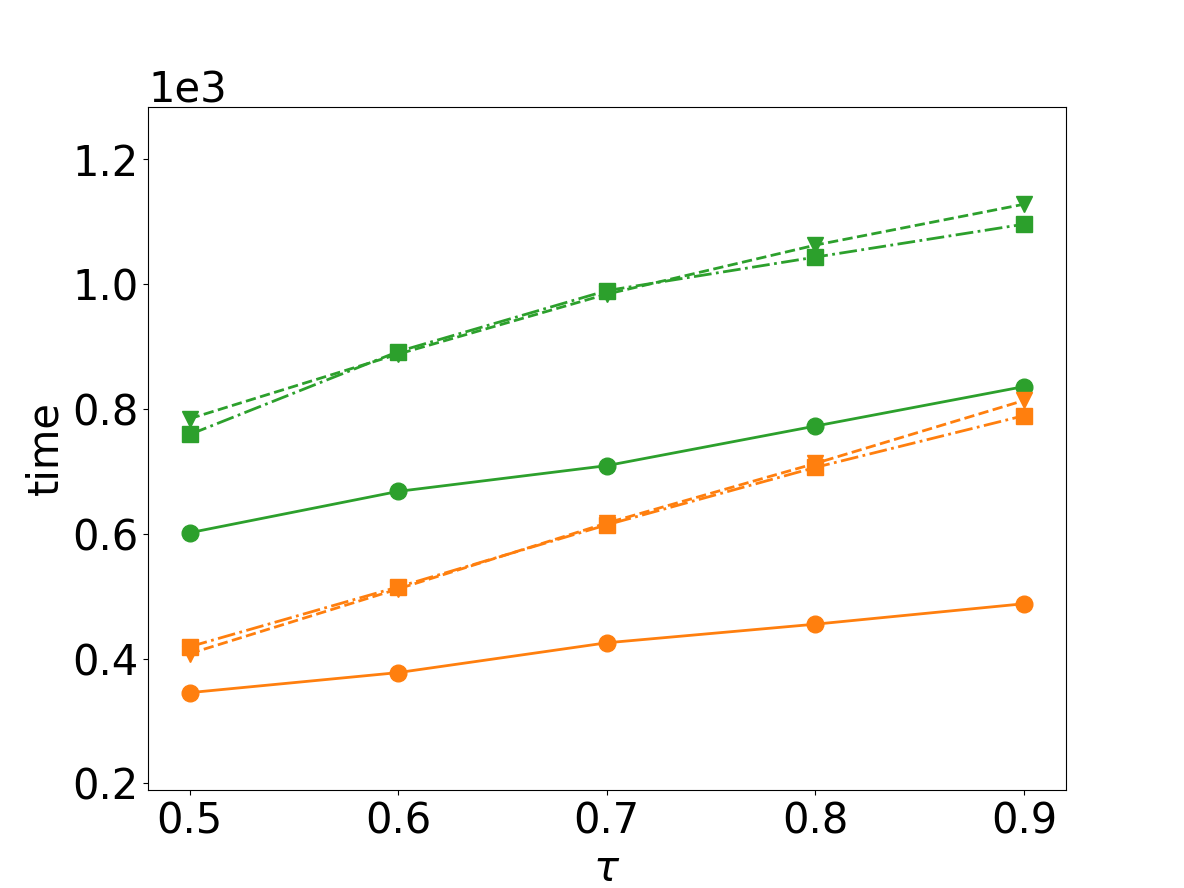}}
\subfloat[soc-pokec \label{3h}]{\includegraphics[width=4.2cm, height=3.36cm]{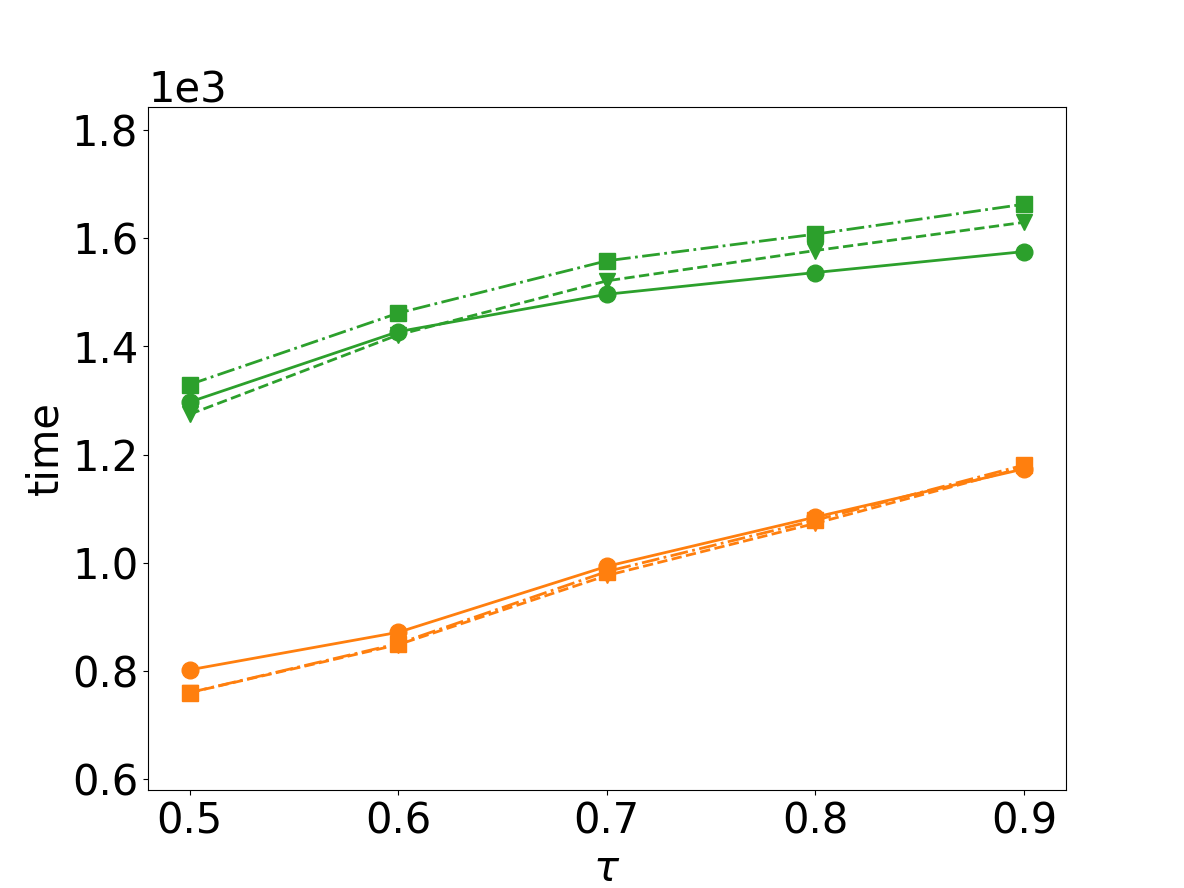}}
\subfloat[cit-Patents \label{3i}]{\includegraphics[width=4.2cm, height=3.36cm]{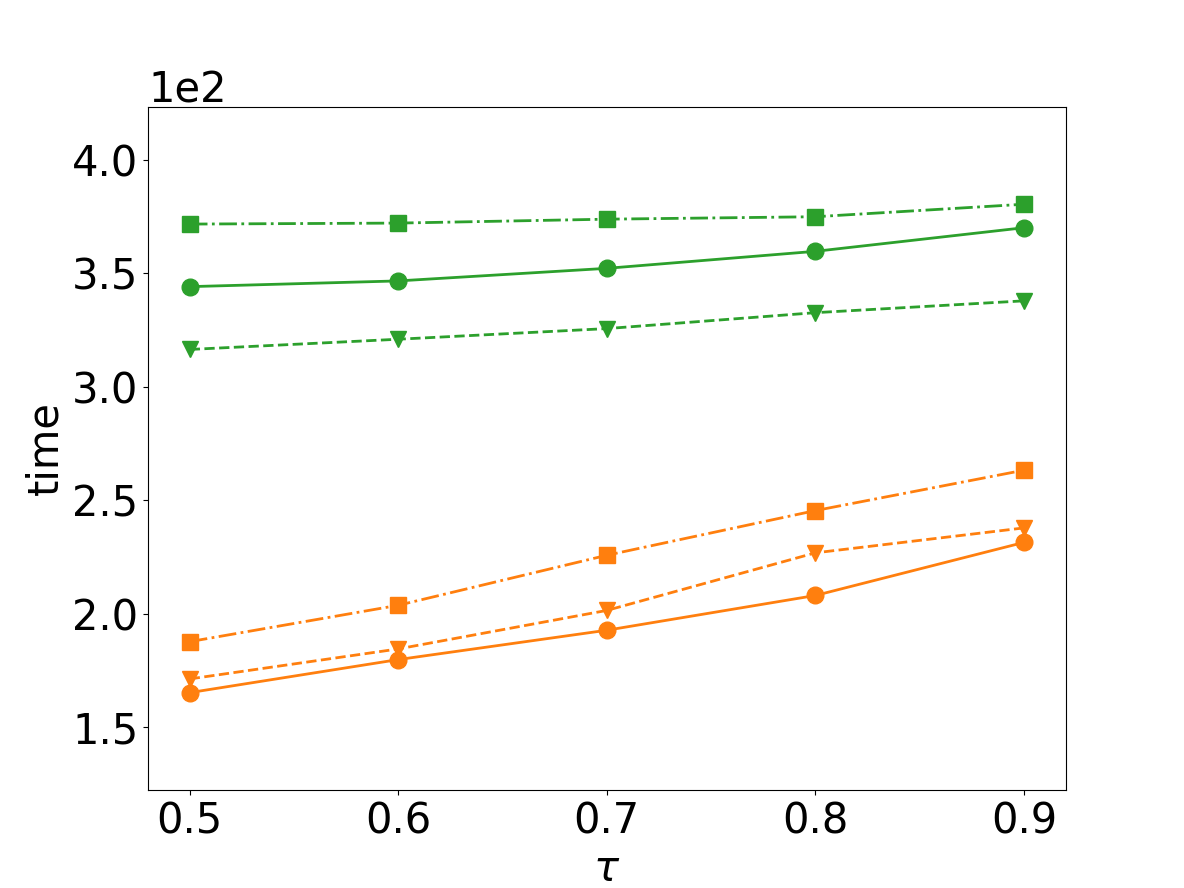}}

%
%
	\captionsetup{justification=centering}
	\caption{\small{Running time of $\tau$-R$^+$MCE and $\tau$-RMCE on eight datasets with different orders, $\tau$ varies from 0.5 to 0.9, T  as default}}\label{fig:torder}
	\vspace{-20pt}
\end{figure*}

\begin{figure*}[ht]
	\centering

	\subfloat[soc-Epinions1 \label{3a}]{\includegraphics[width=4.2cm, height=3.36cm]{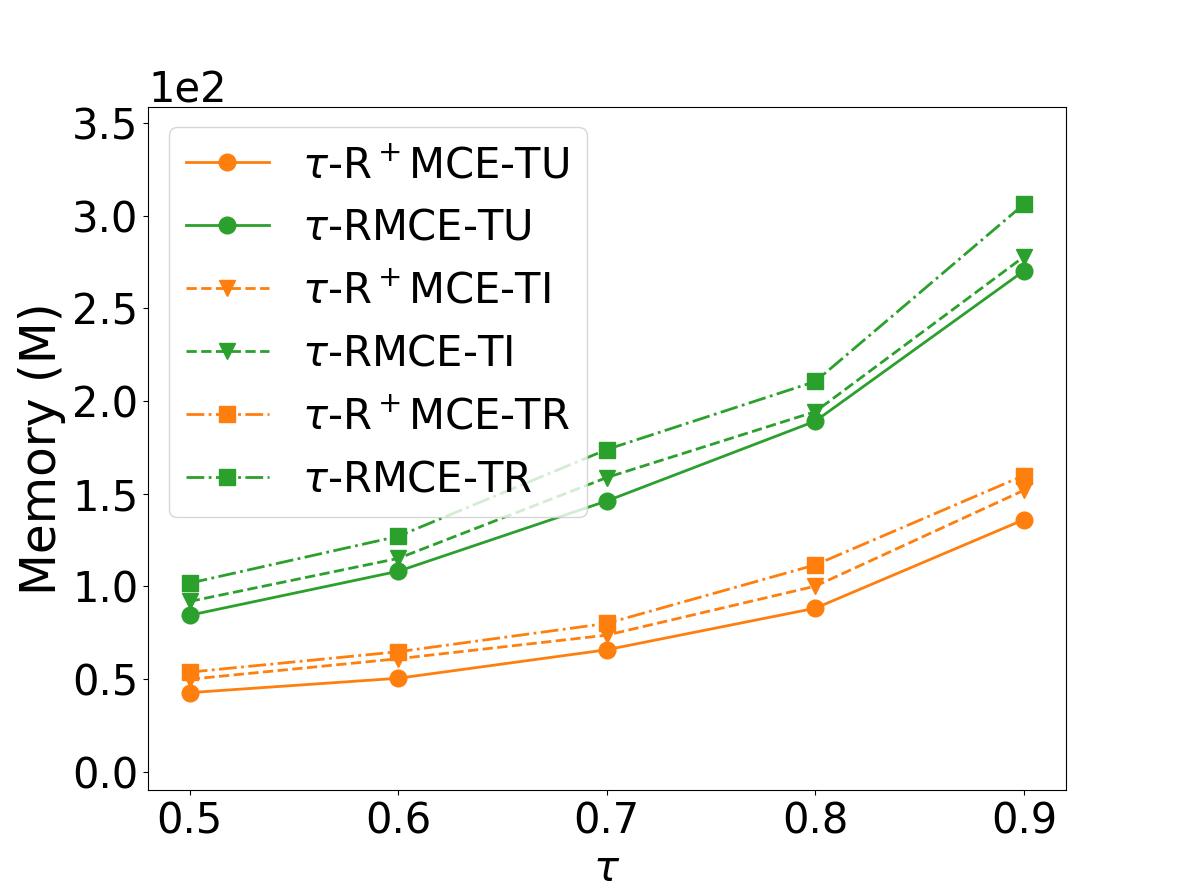}}
	\subfloat[loc-Gowalla \label{3b}]{\includegraphics[width=4.2cm, height=3.36cm]{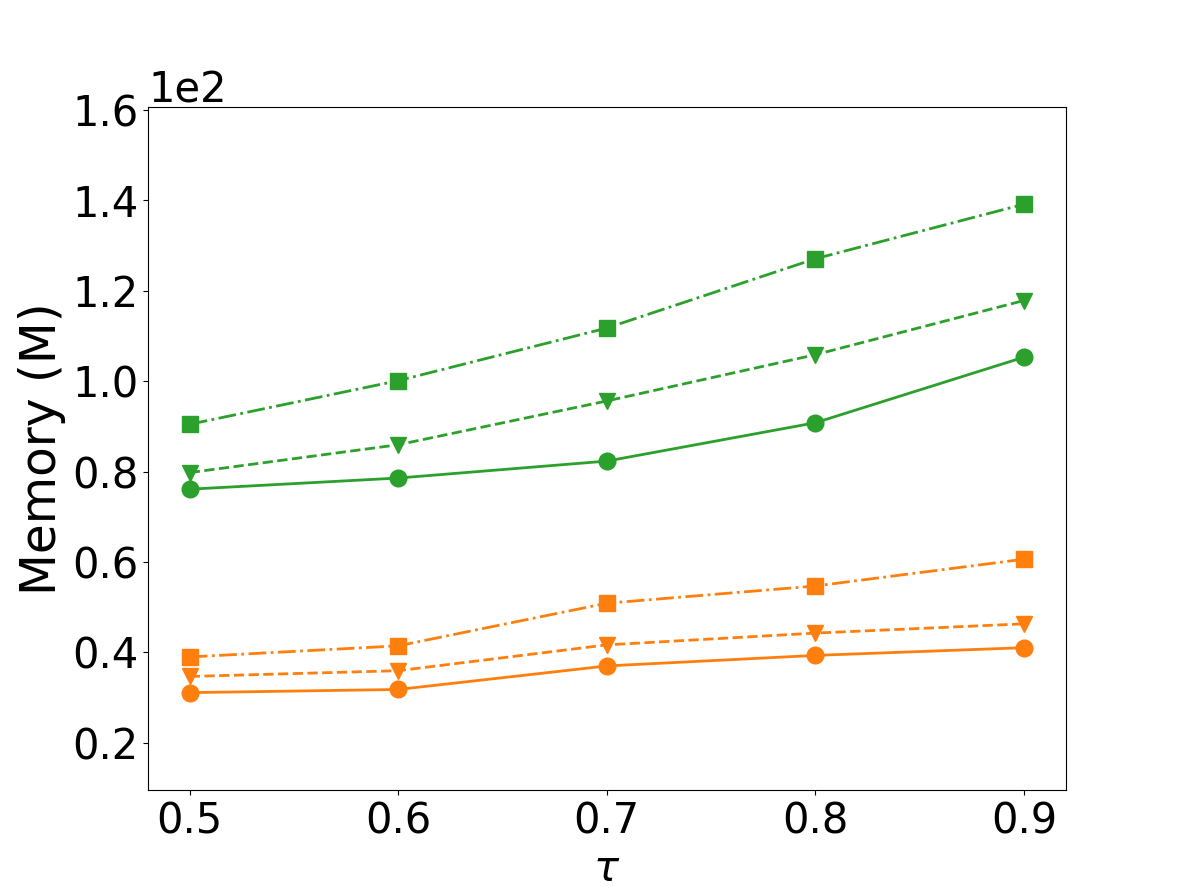}}
	\subfloat[amazon0302 \label{3c}]{\includegraphics[width=4.2cm, height=3.36cm]{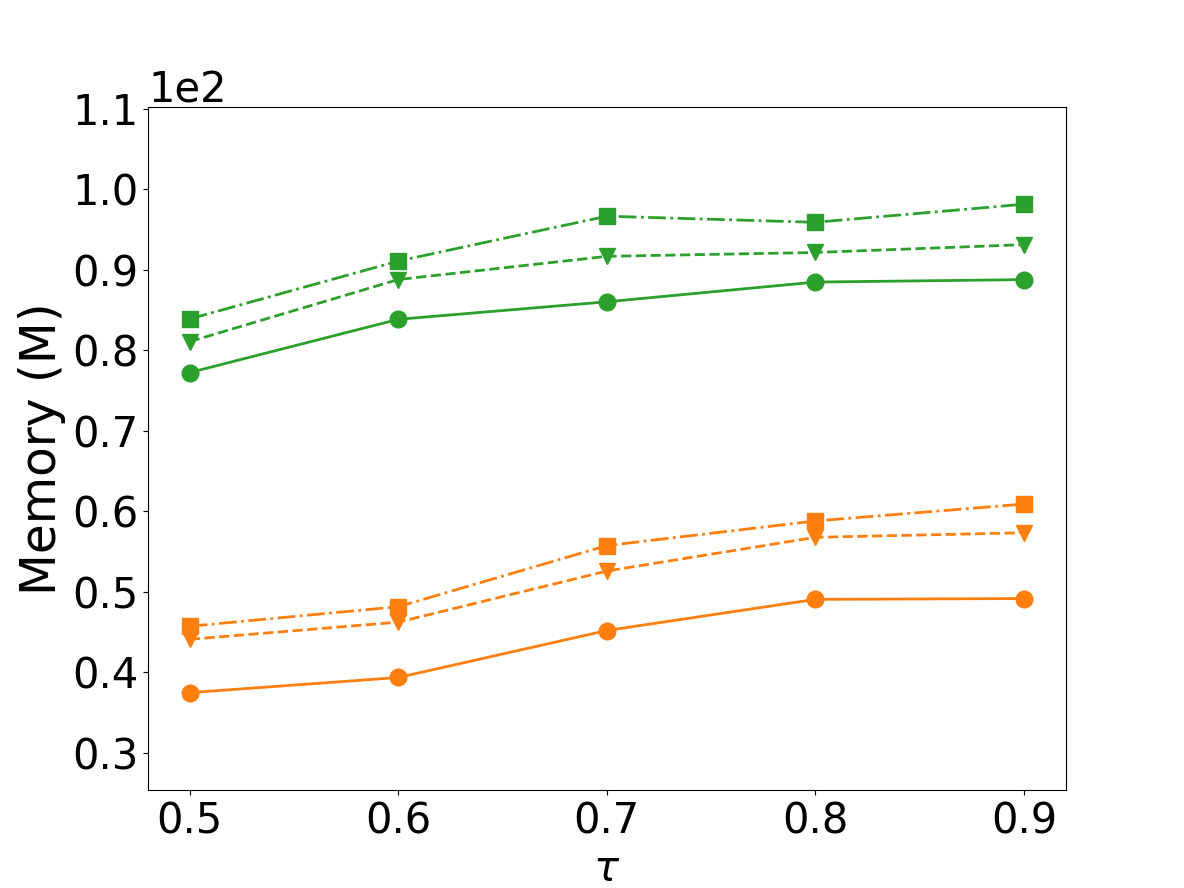}}
\subfloat[email-EuAll  \label{3d}]{\includegraphics[width=4.2cm, height=3.36cm]{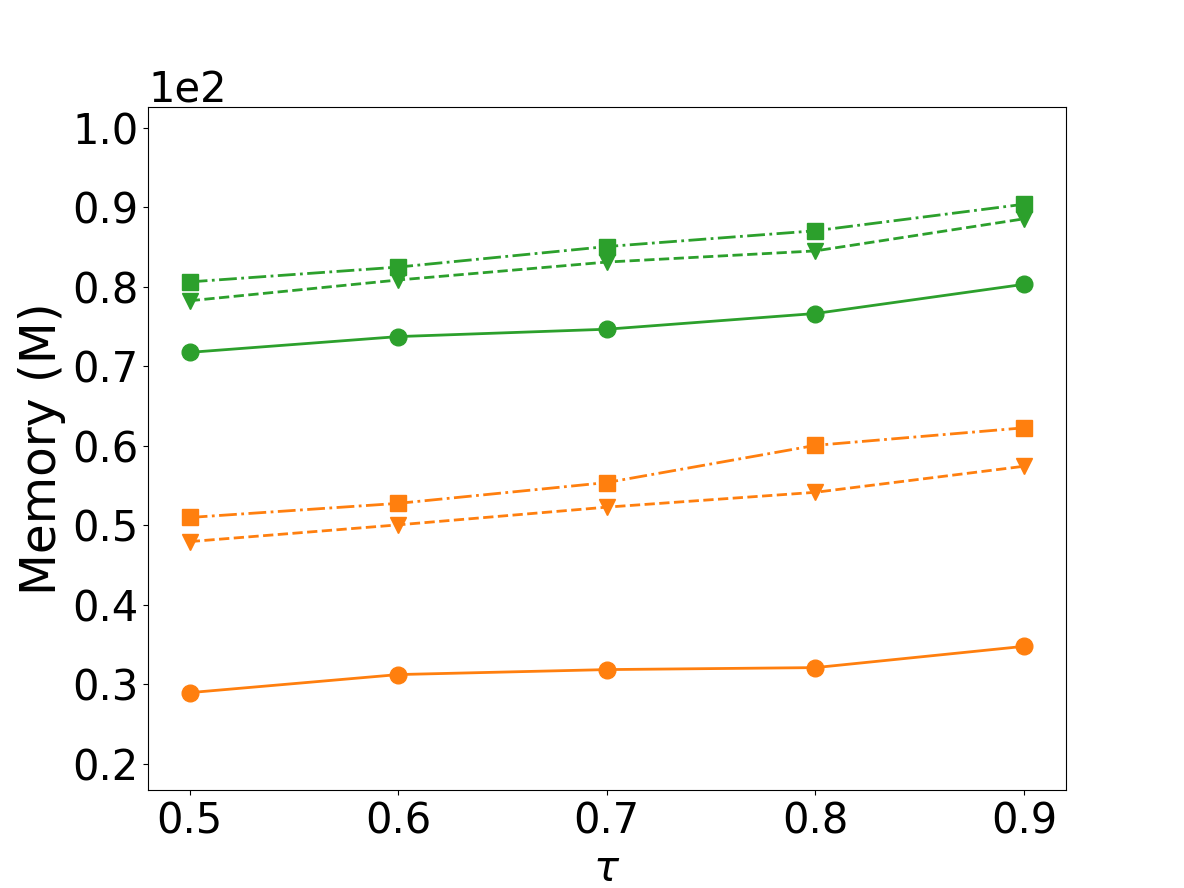}}	

\vspace{-11pt}

\subfloat[web-NotreDame \label{3f}]{\includegraphics[width=4.2cm, height=3.36cm]{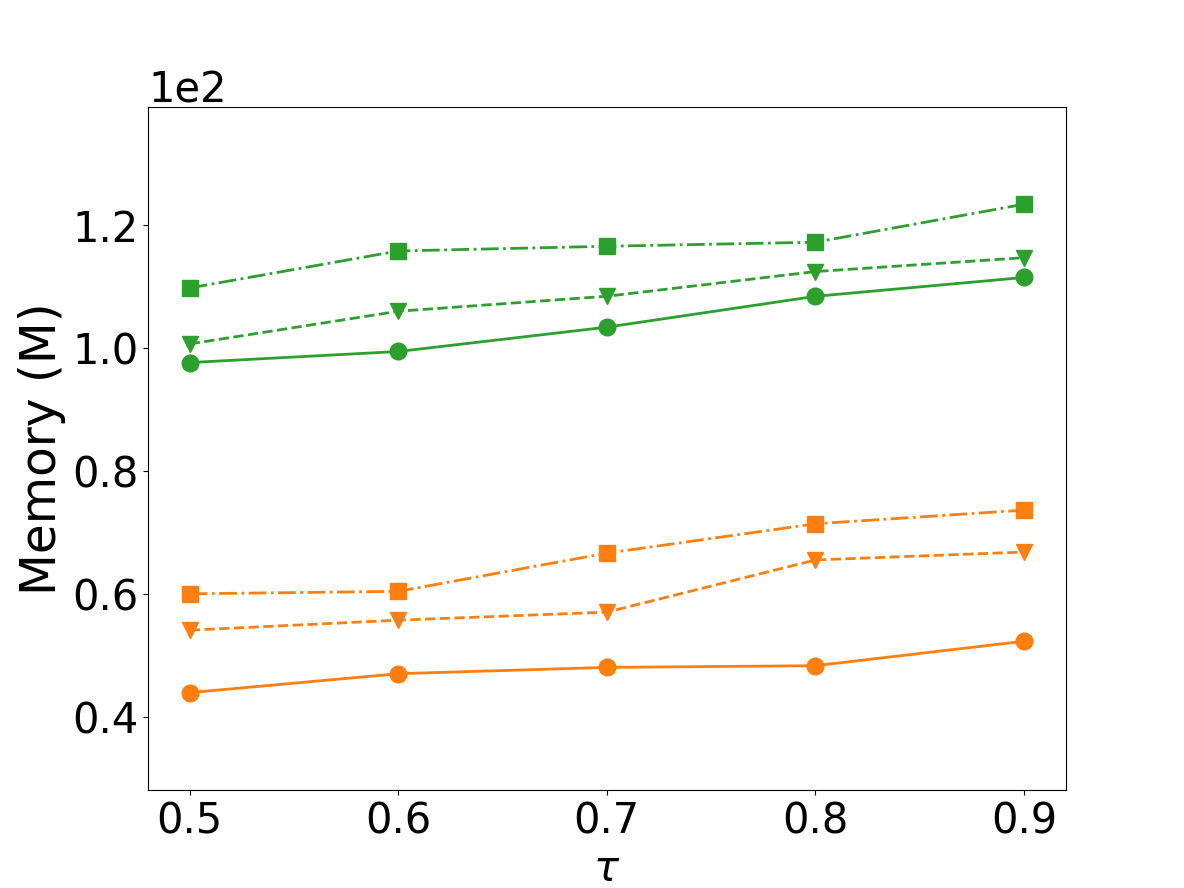}}
\subfloat[com-youtube \label{3g}]{\includegraphics[width=4.2cm, height=3.36cm]{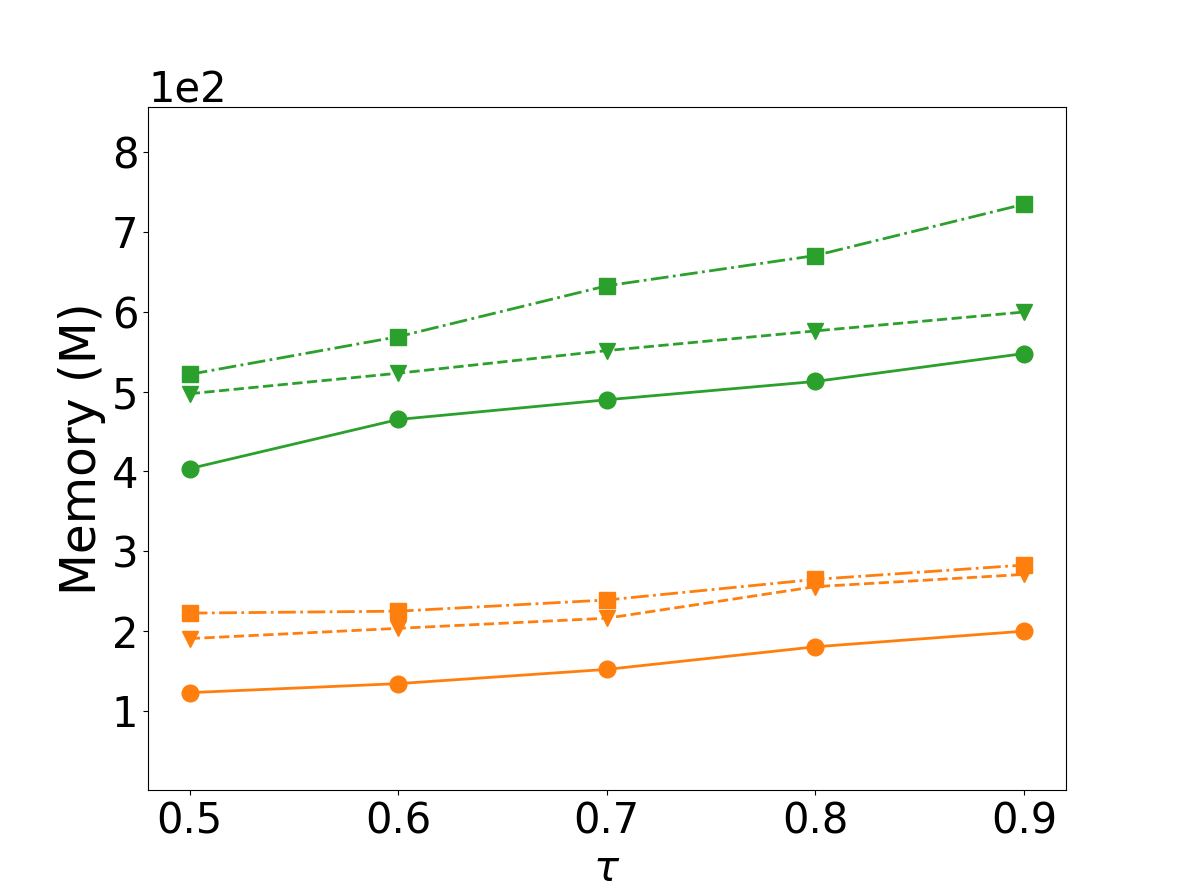}}
\subfloat[soc-pokec \label{3h}]{\includegraphics[width=4.2cm, height=3.36cm]{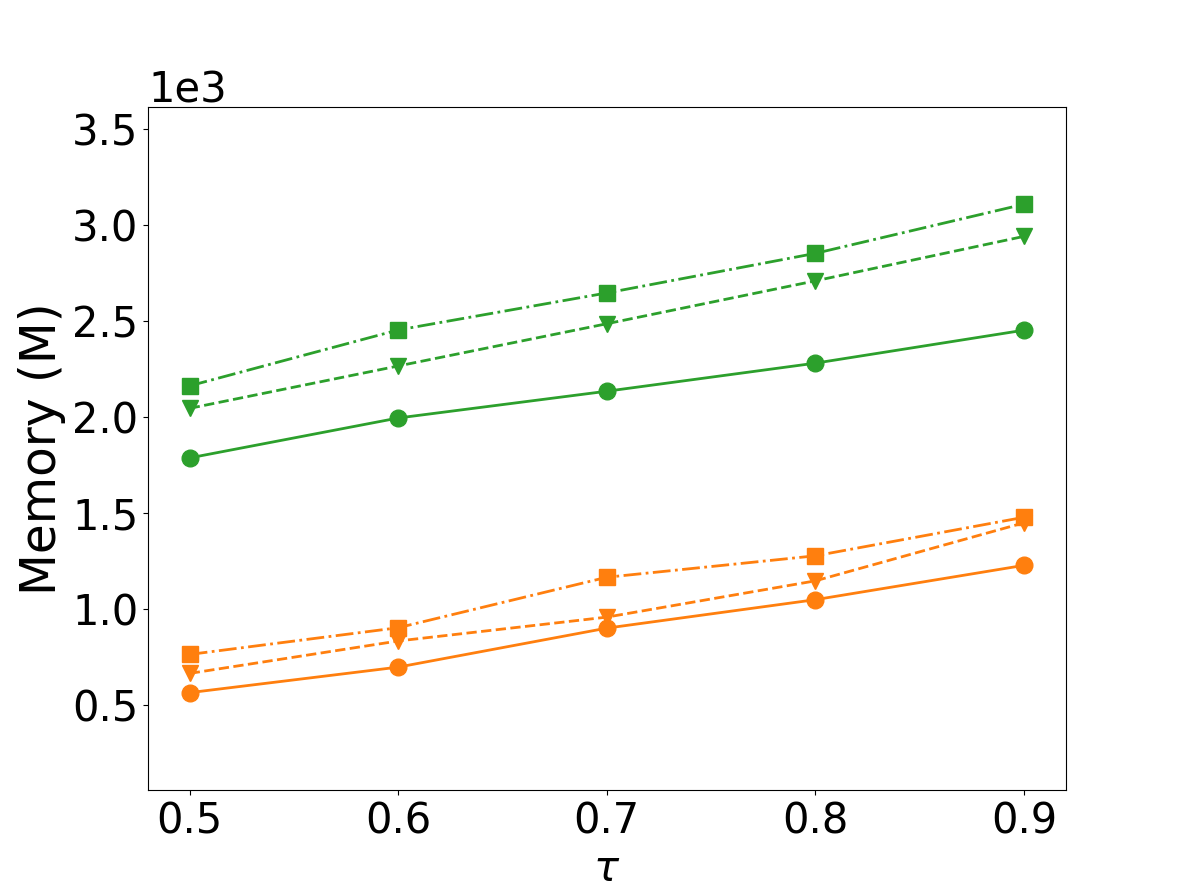}}
\subfloat[cit-Patents \label{3i}]{\includegraphics[width=4.2cm, height=3.36cm]{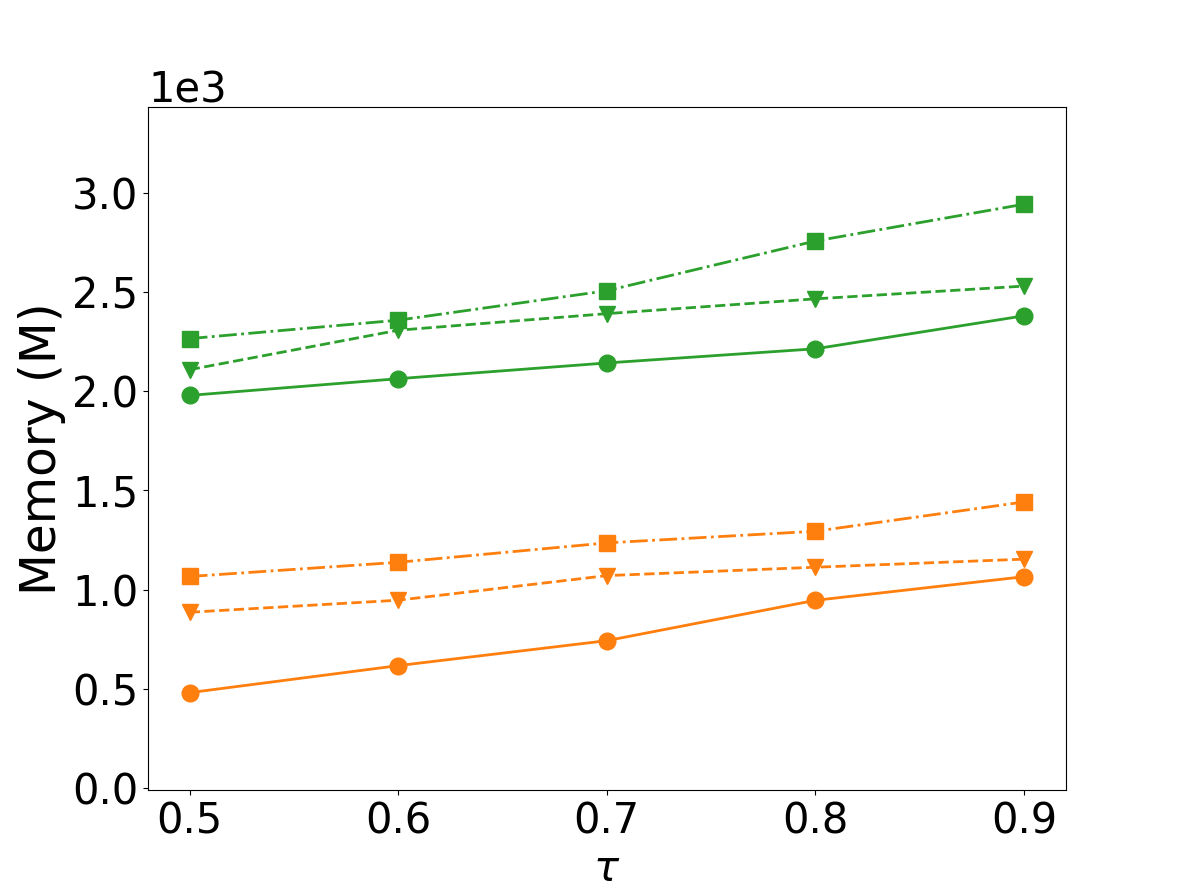}}

%
%
	\captionsetup{justification=centering}
	\caption{Memory requirement of $\tau$-R$^+$MCE and $\tau$-RMCE on eight datasets with different orders, $\tau$ varies from 0.5 to 0.9, T bound as default}\label{fig:morder}
	\vspace{-10pt}
\end{figure*}

\subsubsection{Efficiency of orders}
\label{sc:622}
\begin{figure*}[ht]
  	\centering
	\vspace{-15pt}

	\subfloat[soc-Epinions1 \label{3a}]{\includegraphics[width=4.2cm, height=3.36cm]{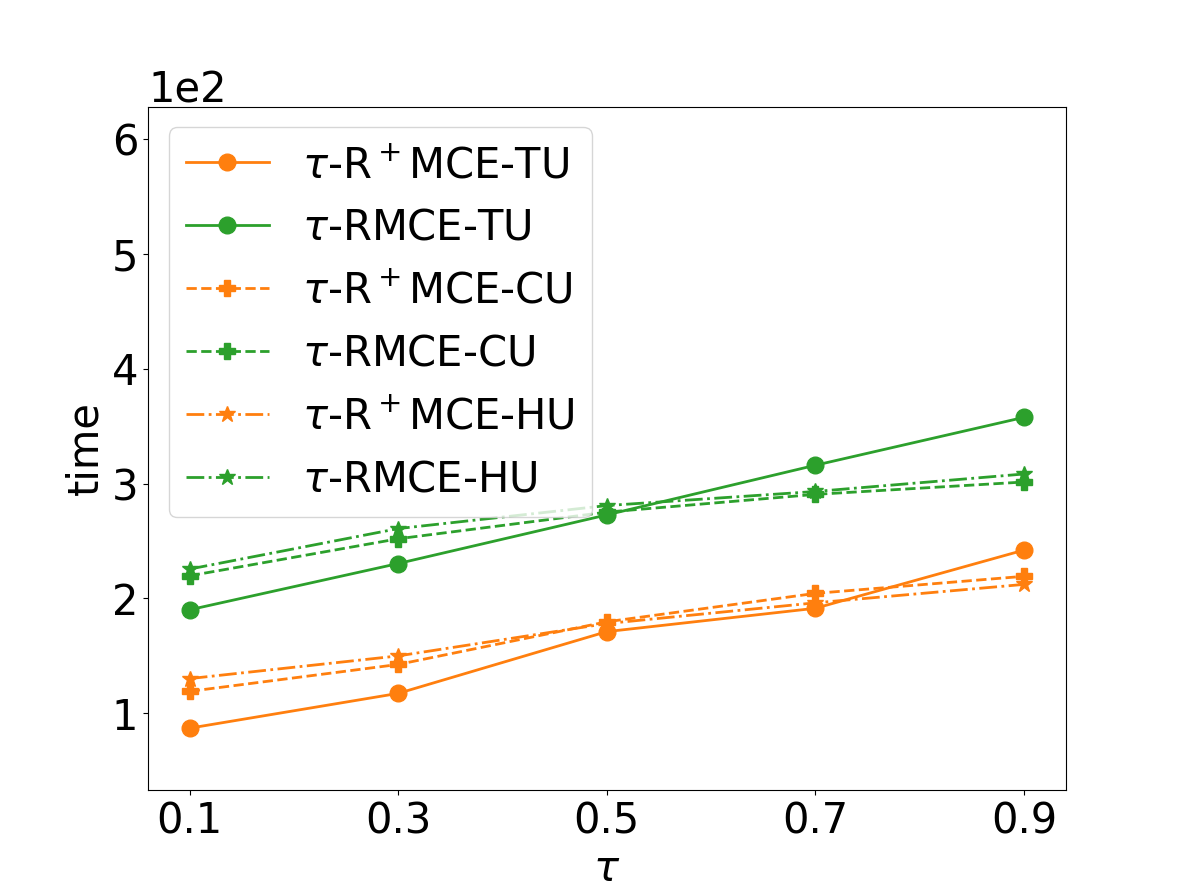}}
	\subfloat[loc-Gowalla \label{3b}]{\includegraphics[width=4.2cm, height=3.36cm]{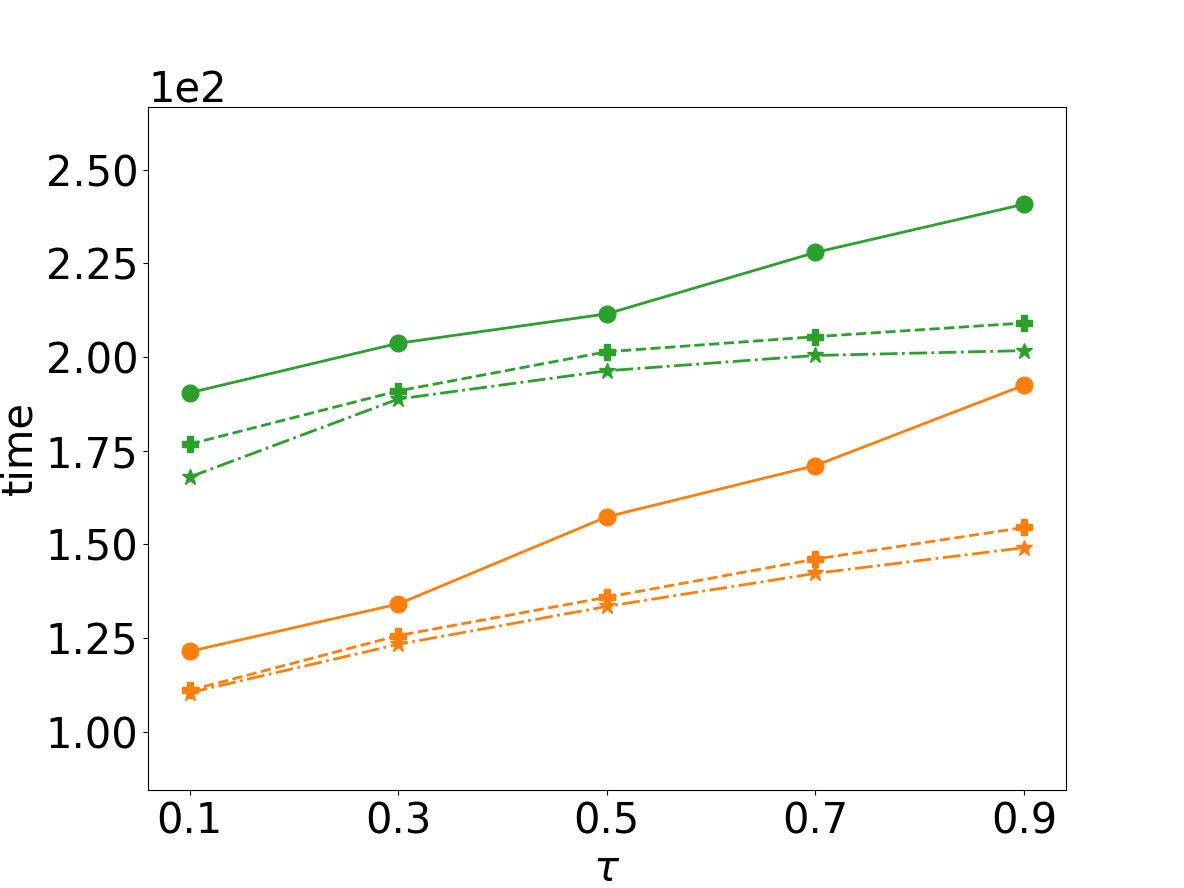}}
	\subfloat[amazon0302 \label{3c}]{\includegraphics[width=4.2cm, height=3.36cm]{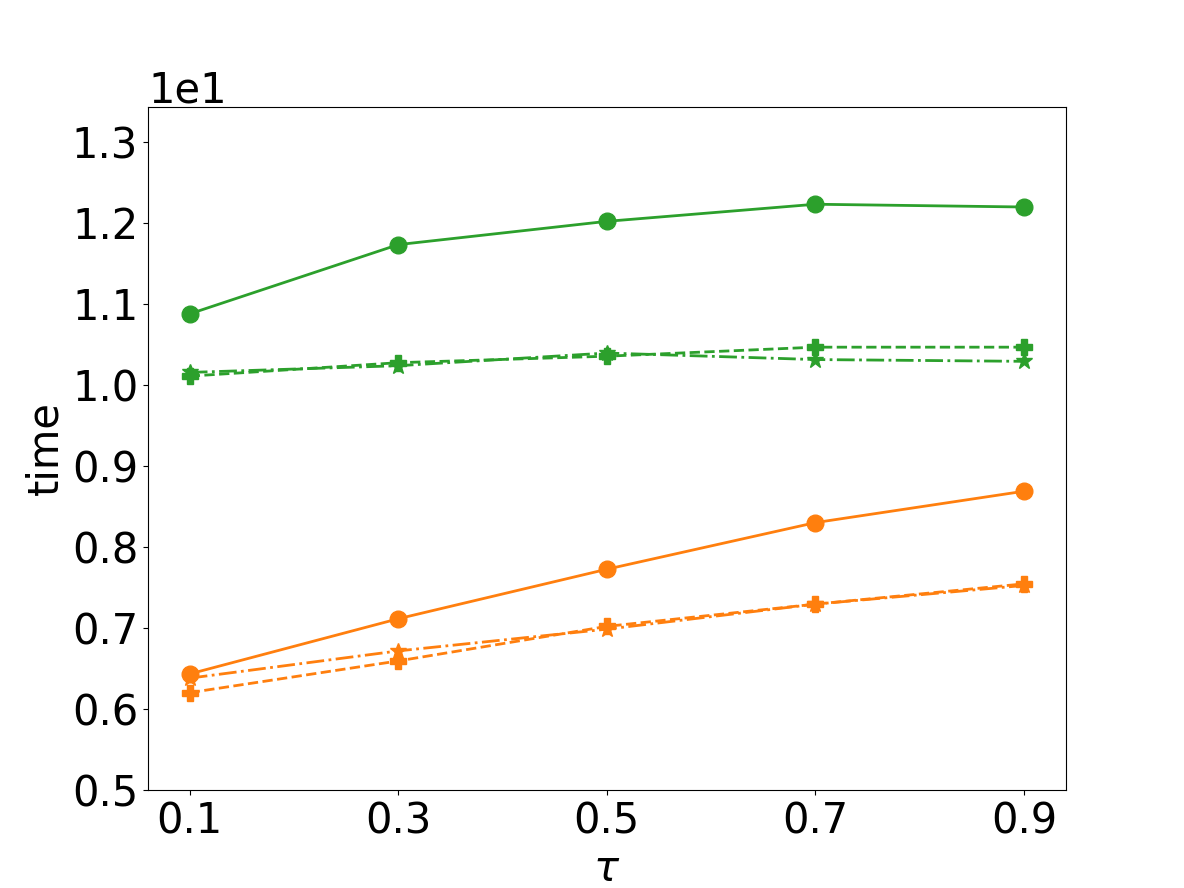}}
\subfloat[email-EuAll  \label{3d}]{\includegraphics[width=4.2cm, height=3.36cm]{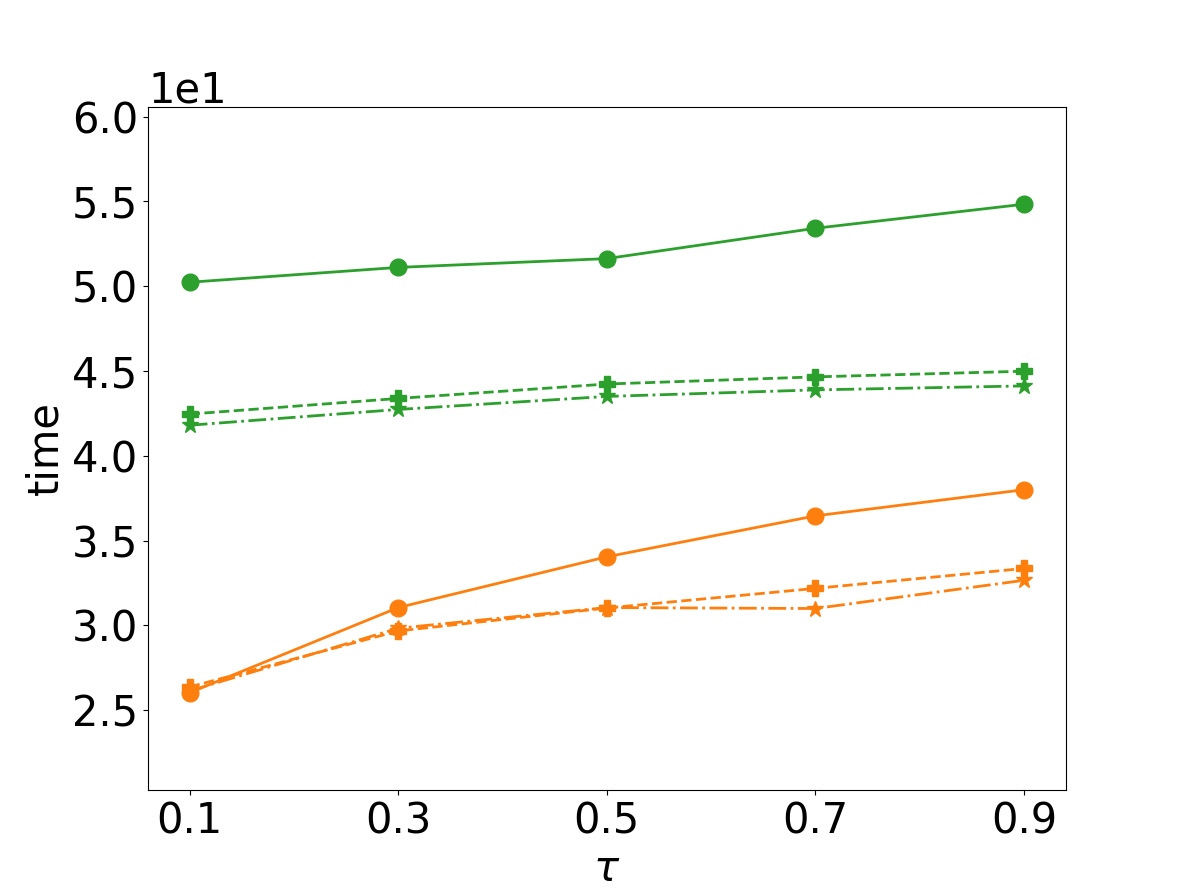}}	

\vspace{-13pt}

\subfloat[web-NotreDame \label{3f}]{\includegraphics[width=4.2cm, height=3.36cm]{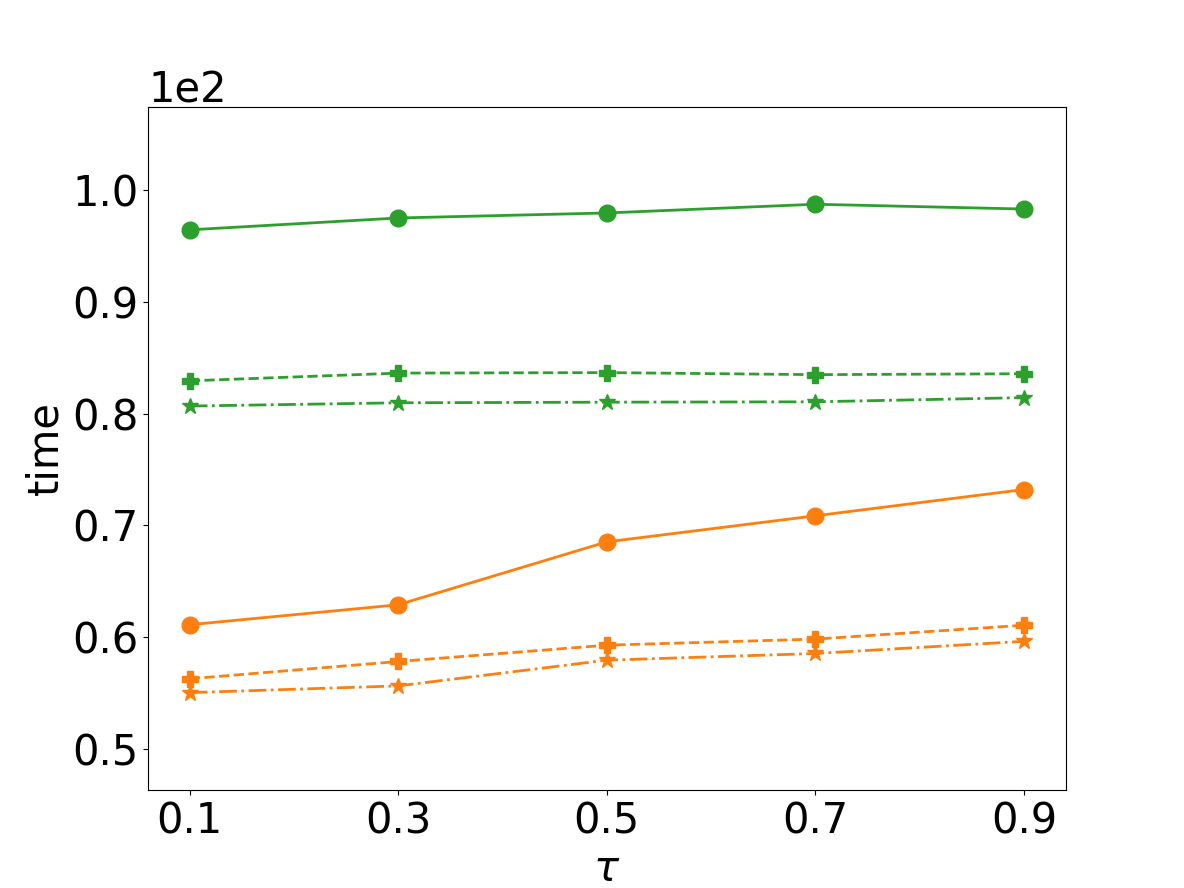}}
\subfloat[com-youtube \label{3g}]{\includegraphics[width=4.2cm, height=3.36cm]{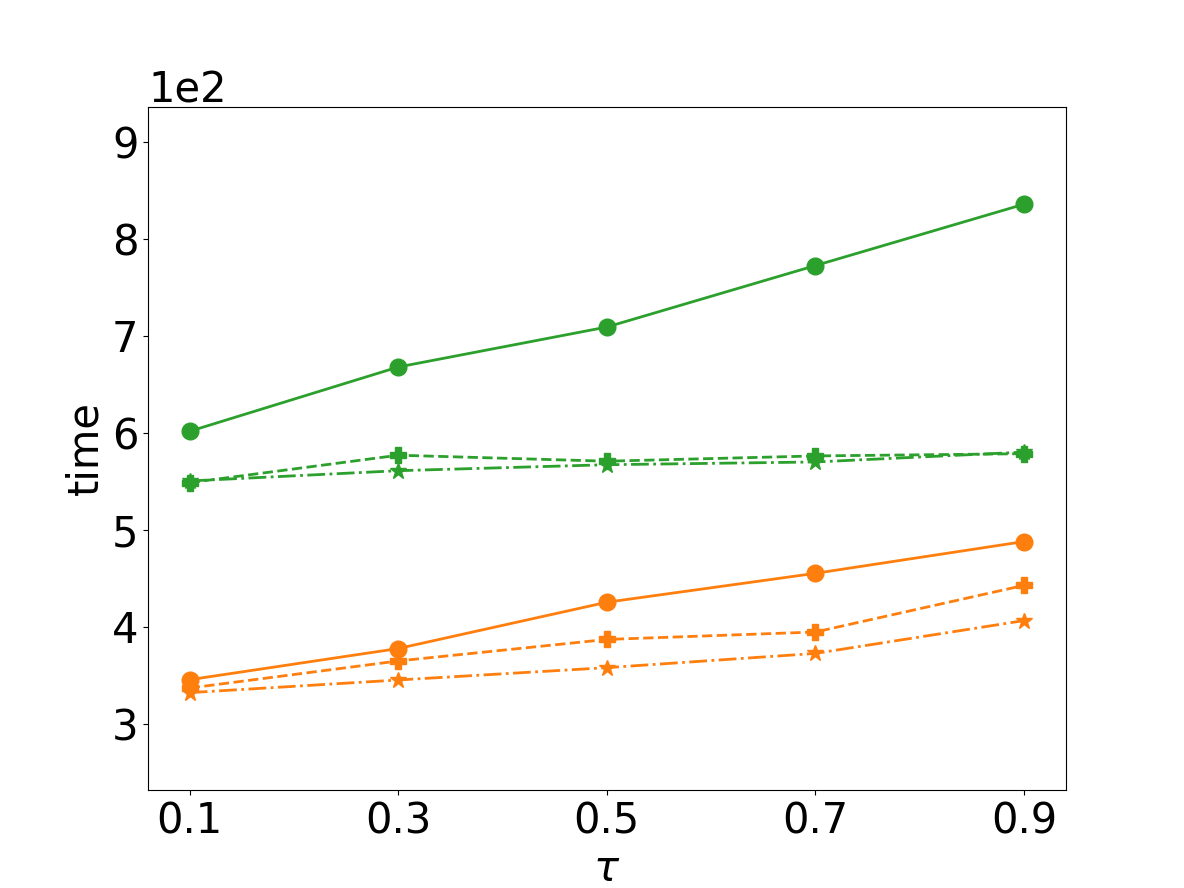}}
\subfloat[soc-pokec \label{3h}]{\includegraphics[width=4.2cm, height=3.36cm]{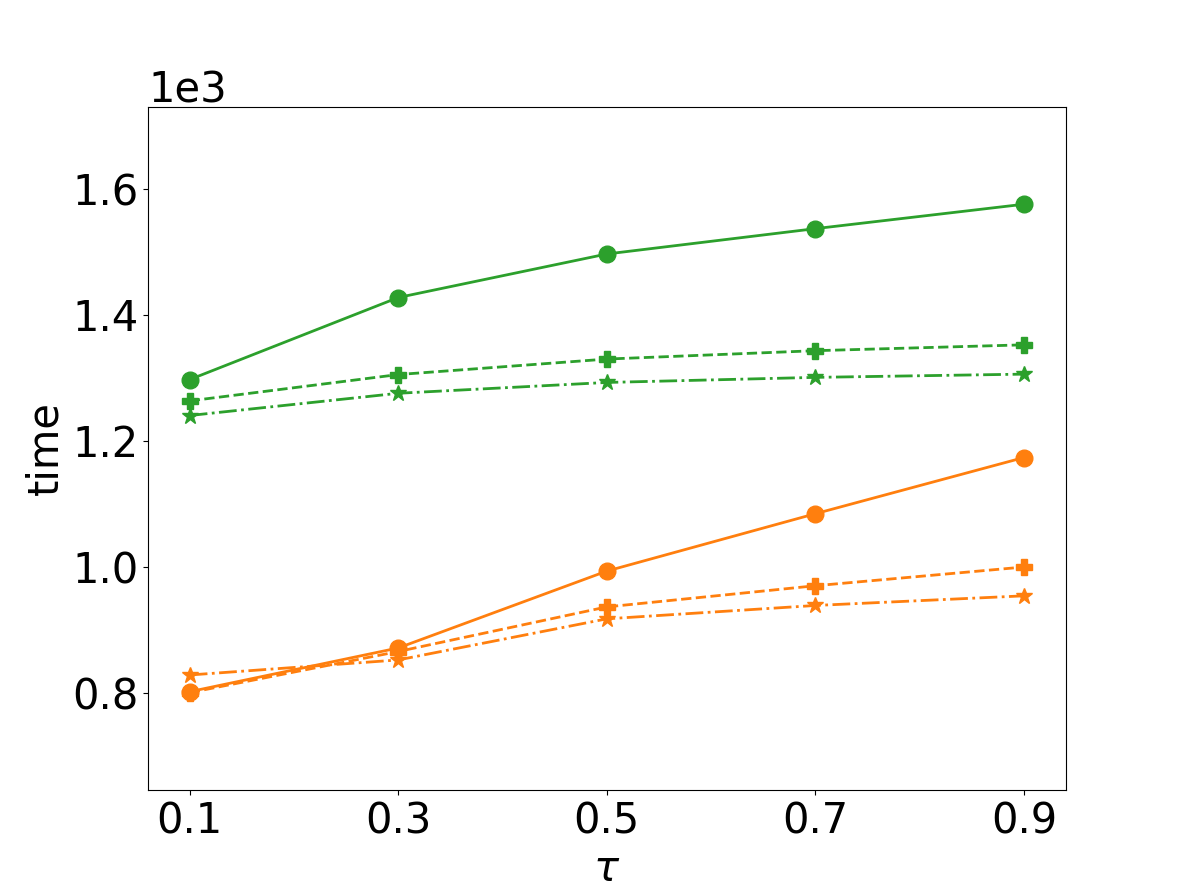}}
\subfloat[cit-Patents \label{3i}]{\includegraphics[width=4.2cm, height=3.36cm]{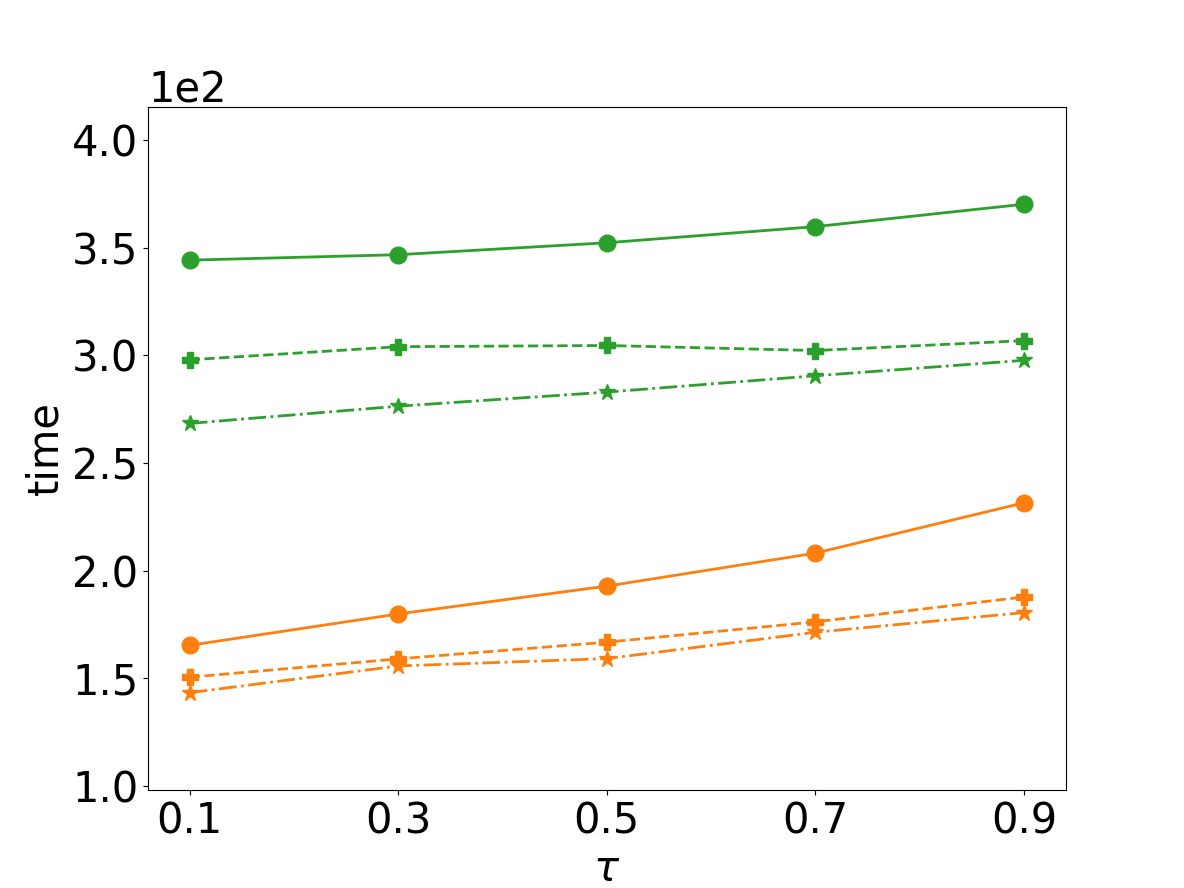}}

%
%
	\captionsetup{justification=centering}
	\caption{Running time of $\tau$-R$^+$MCE and $\tau$-RMCE on eight datasets with different bounds, $\tau$ varies from 0.5 to 0.9, U order as default}\label{fig:tbound}
	\vspace{-20pt}
\end{figure*}

%
%
%
%
%
%
%

\begin{figure*}[ht]
	\centering

	\subfloat[soc-Epinions1 \label{3a}]{\includegraphics[width=4.2cm, height=3.36cm]{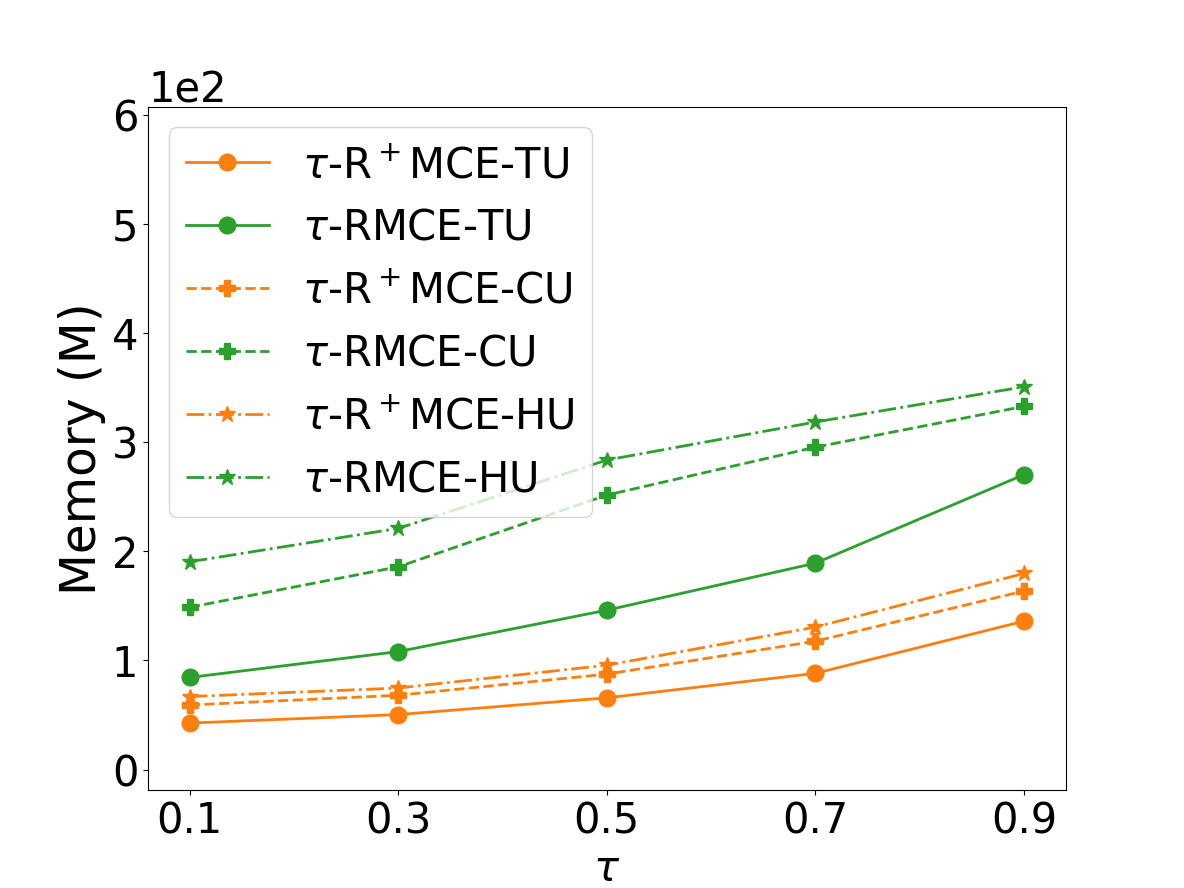}}
	\subfloat[loc-Gowalla \label{3b}]{\includegraphics[width=4.2cm, height=3.36cm]{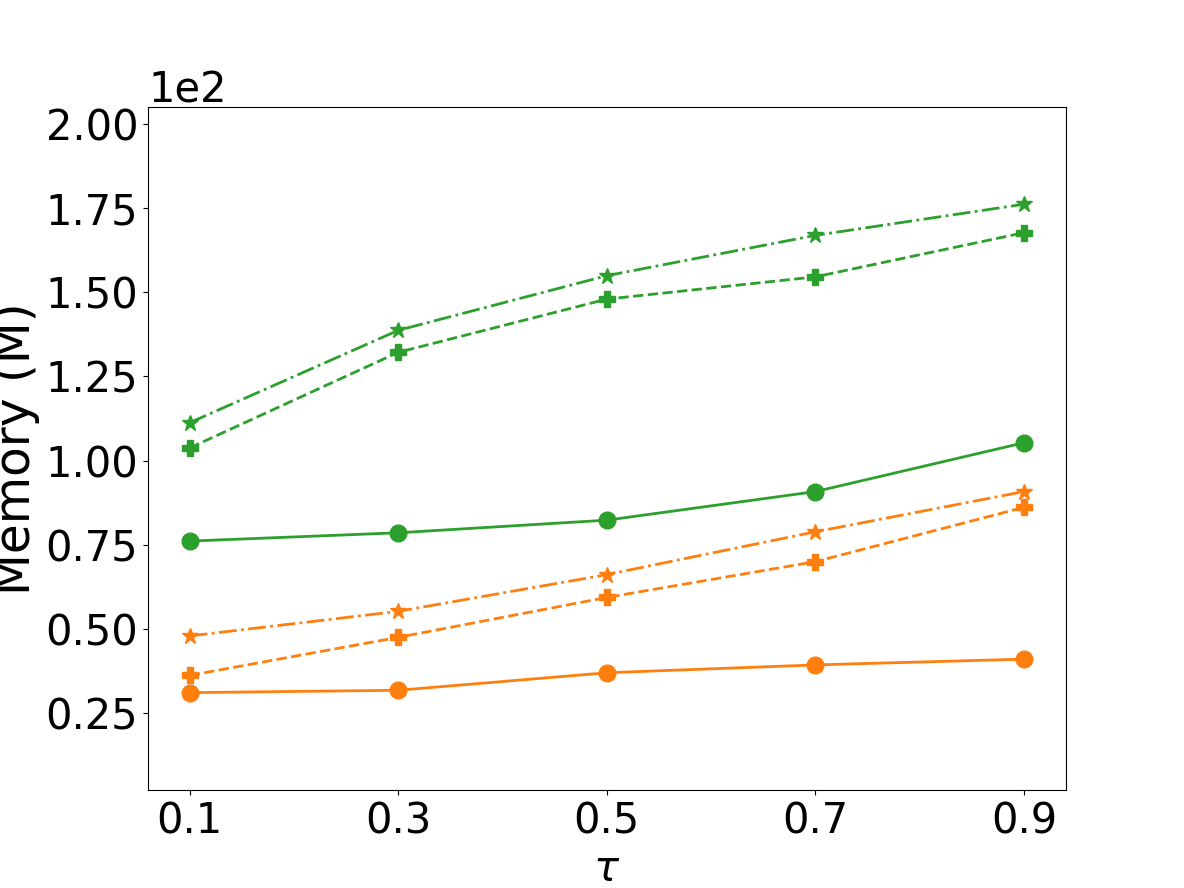}}
	\subfloat[amazon0302 \label{3c}]{\includegraphics[width=4.2cm, height=3.36cm]{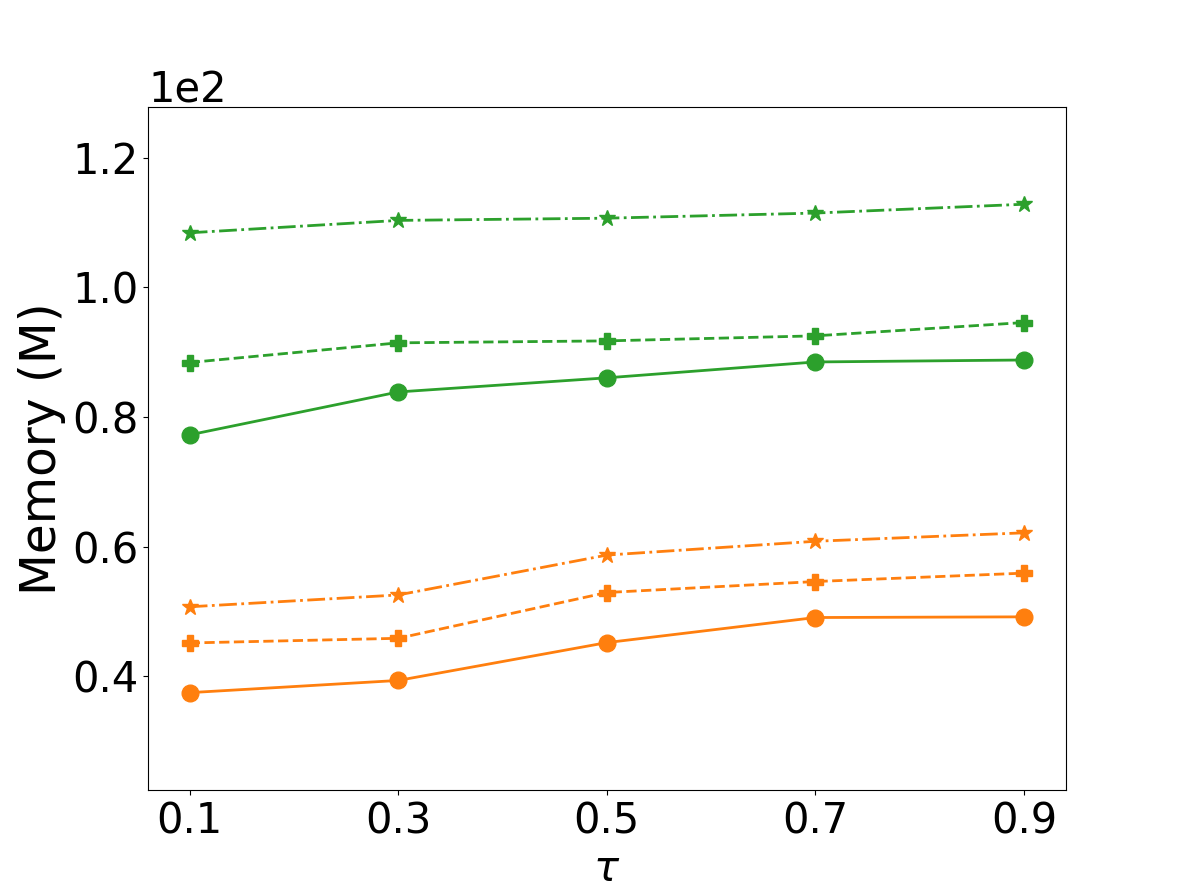}}
\subfloat[email-EuAll  \label{3d}]{\includegraphics[width=4.2cm, height=3.36cm]{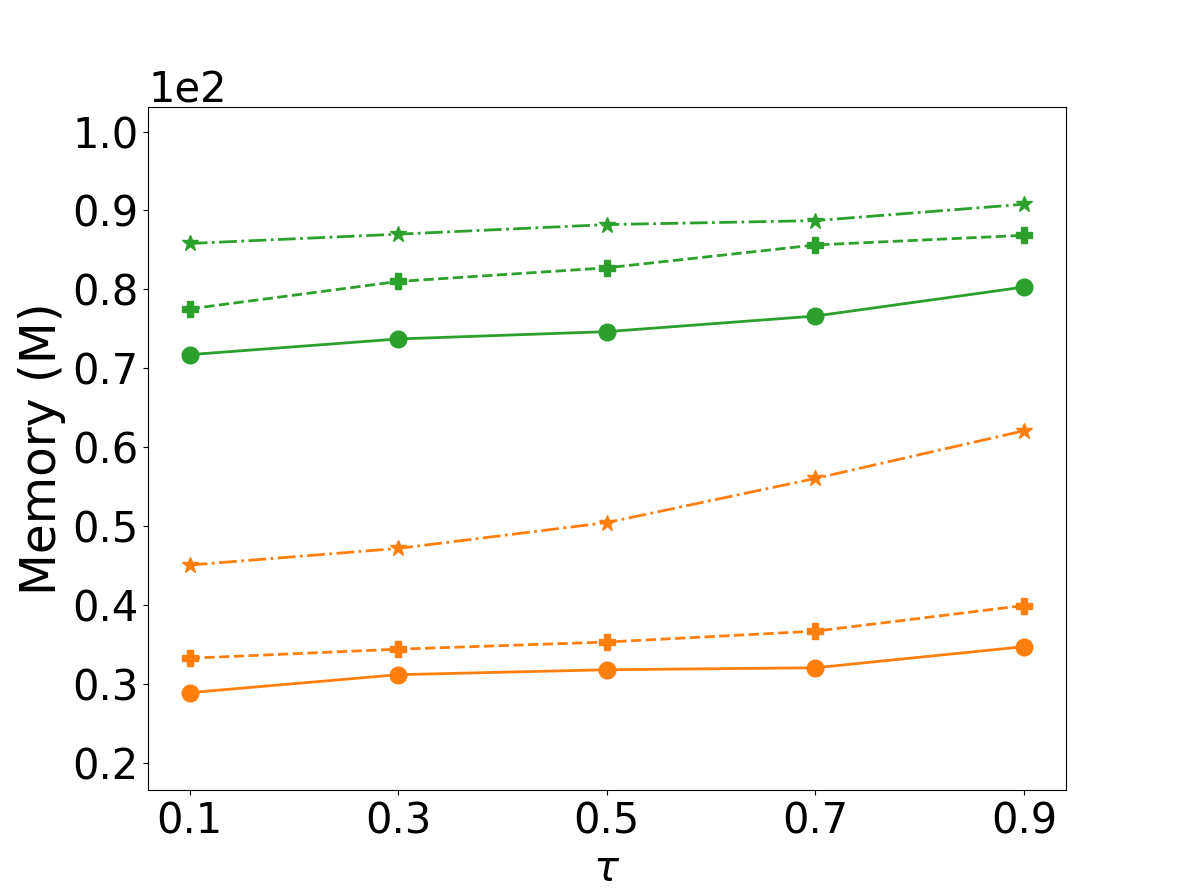}}	

\vspace{-13pt}

\subfloat[web-NotreDame \label{3f}]{\includegraphics[width=4.2cm, height=3.36cm]{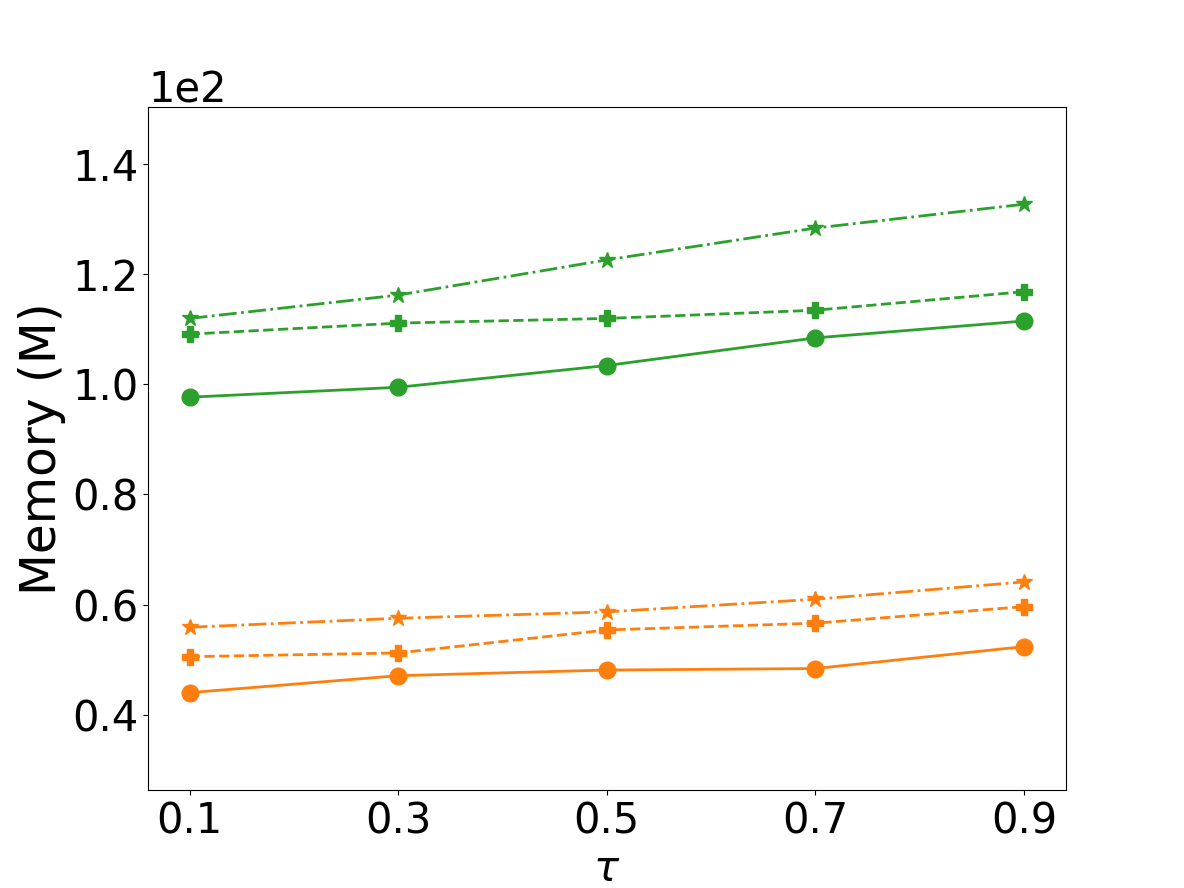}}
\subfloat[com-youtube \label{3g}]{\includegraphics[width=4.2cm, height=3.36cm]{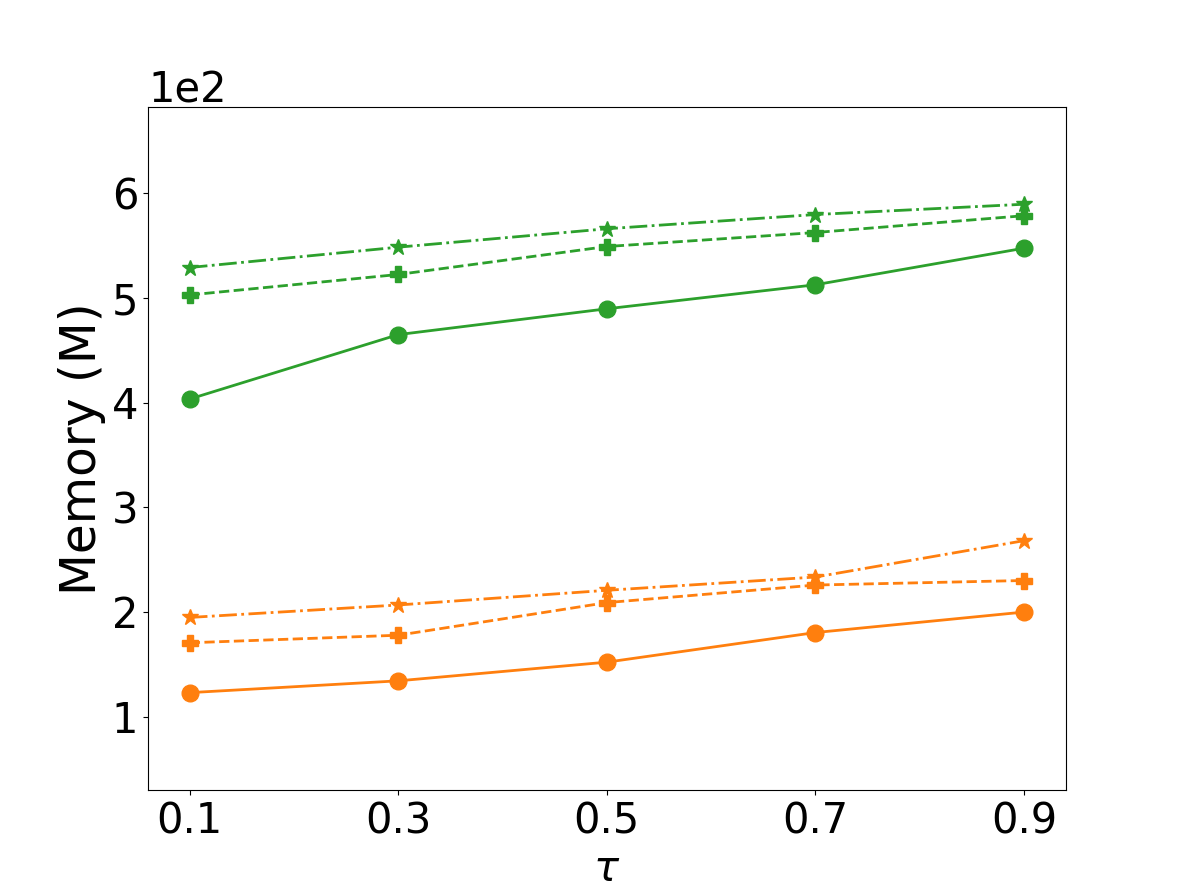}}
\subfloat[soc-pokec \label{3h}]{\includegraphics[width=4.2cm, height=3.36cm]{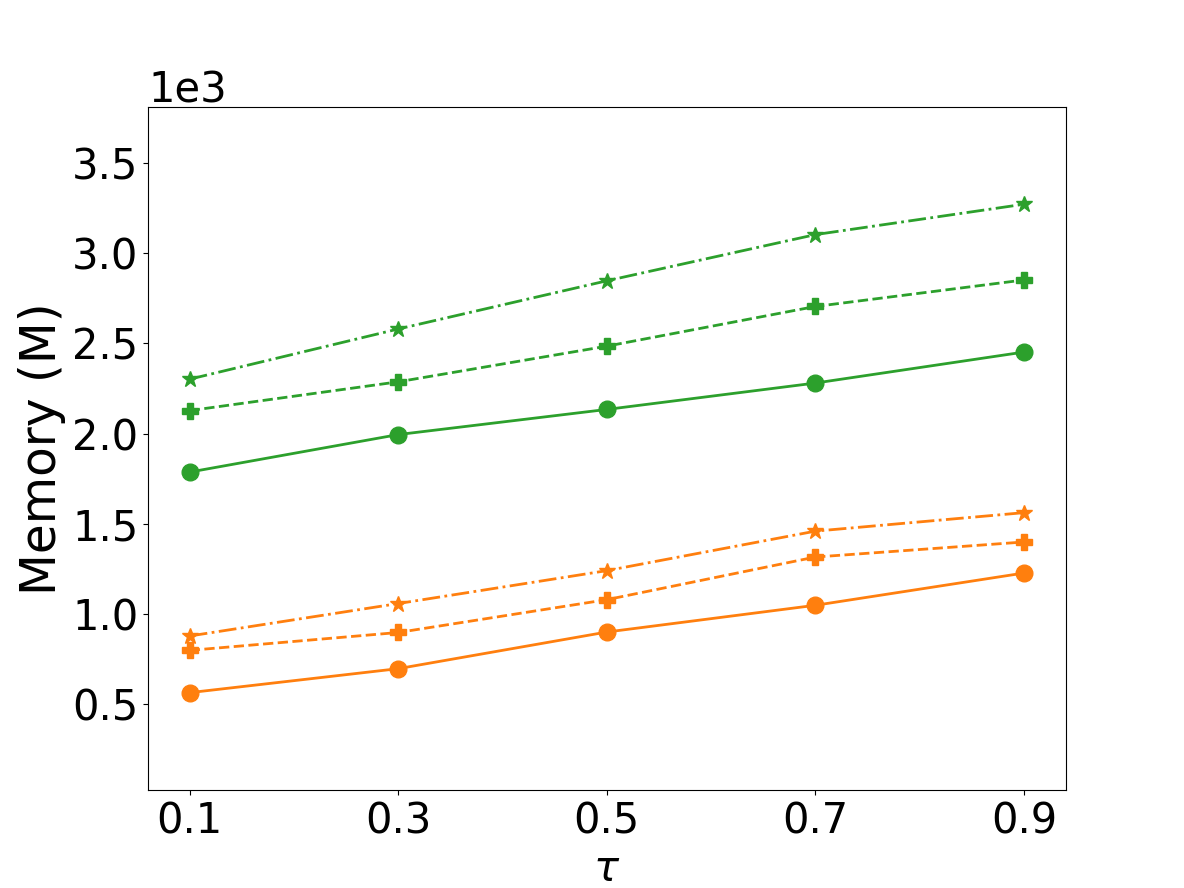}}
\subfloat[cit-Patents \label{3i}]{\includegraphics[width=4.2cm, height=3.36cm]{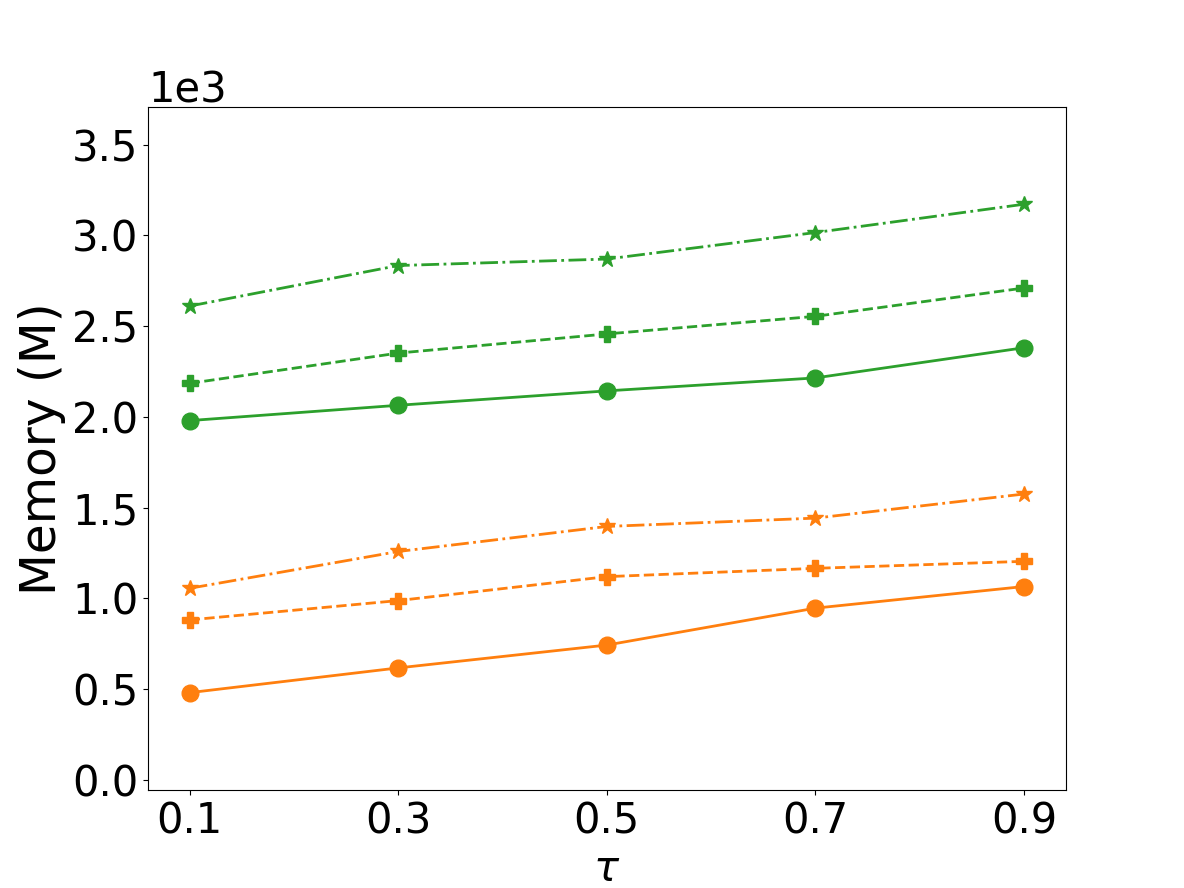}}

%
%
	\captionsetup{justification=centering}
	\caption{Memory requirement of $\tau$-R$^+$MCE and $\tau$-RMCE on eight datasets with different bounds, $\tau$ varies from 0.5 to 0.9, U order as default}\label{fig:mbound}
	\vspace{-10pt}
\end{figure*}
To test the efficiency of three types of  vertex orders, we implement $\tau$-RMCE and $\tau$-R$^+$MCE with orders U, I and R. 
The default bound is set  to T. 
We recorded both the total running time and  memory requirement for all experiments. 
The details are shown in Fig.~\ref{fig:torder} and Fig.~\ref{fig:morder}. 
\begin{table}
\renewcommand\arraystretch{1.2}
\small
\centering
\caption{First-Result Time (s) at $\tau = 0.9$ }
\begin{tabular}{p{60pt}p{60pt}p{60pt}}
\toprule
Name  & $\tau$-RMCE-TU &   $\tau$-R$^+$MCE-TU     \\
\midrule
soc-Epinions1     & 0.62  & 0.61   \\
loc-Gowalla      & 8.55 & 8.55  \\
amazon0302       & 0.21  & 0.21    \\
email-EuAll     & 1.22  & 1.20   \\
NotreDame     & 5.97  & 5.97    \\
com-youtube    & 12.33  & 12.32   \\
soc-pokec    & 40.11 & 39.57   \\ 
cit-Patents & 10.30 & 9.27\\
\bottomrule
\end{tabular}
\label{frt}
\vspace{-10pt}
\end{table}

%

\textbf{Running time: }
Fig.~\ref{fig:torder} shows that the results of $\tau$-R$^+$MCE and $\tau$-RMCE are very similar on each dataset, 
hence we focus on  curves of $\tau$-R$^+$MCE. 
We see that $\tau$-R$^+$MCE-TU shows the best performance on five out of eight datasets 
({\it soc-Epinions1 ($18\%\sim 29\%$), amazon0302 ($5\%\sim 7\%$), email-EuAll ($8\%\sim 23\%$), com-yotube ($15\%\sim 40\%$), cit-Patents ($3\%\sim 4\%$)}, 
where the percentages in parentheses are the range of reductions vs. $\tau$-R$^+$MCE-TI). 
It shows similar performances as $\tau$-R$^+$MCE-TI on two datasets ({\it web-NotreDame, soc-pokec}) since the two lines coincide with each other. 
$\tau$-R$^+$MCE-TU shows the worst performance on a special dataset {\it loc-Gowalla} because of its small degeneracy. 
This result implies that benefited from its summarization effectiveness, $\tau$-R$^+$MCE-TU shows a comparable or even better performance than the state-of-the-art order on a variety of real-world datasets. 
However, degeneracy order is still the best choice for graphs with small degeneracies that this order is initially designed for.

\textbf{Memory requirement:}
Fig.~\ref{fig:morder} shows the memory requirement for different orders. 
We see that the truss order U consistently outperforms the other two for both $\tau$-R$^+$MCE and $\tau$-RMCE. 
The memory reduction of $\tau$-R$^+$MCE-TU vs. $\tau$-R$^+$MCE-TI varies little when $\tau$ changes. 
The reduction is more than $10\%$ for all eight datasets, with two of which ({\it email-EuAll, com-youtube}) even achieving $30\%$.
The results of memory cost are much similar to the output size.  
This is because the memory requirement highly relies on the depth of recursion: more number of deep branches result in higher memory consumption. 
The strong locality thus early pruning power of $\tau$-R$^+$MCE-TU prevents some of the redundant branches from growing unnecessarily deep, hence the memory requirement can be reduced significantly, 
which has the same reason why the output summary size is reduced.

}

\subsubsection{Efficiency of bounds}

We test the efficiency of different bounds (T, C, H) with default vertex order U. 
Both the total running time (Fig.~\ref{fig:tbound}) and memory requirement (Fig.~\ref{fig:mbound}) are recorded. 
%

\textbf{Running time:}
Fig.~\ref{fig:tbound} shows that for both $\tau$-R$^+$MCE and $\tau$-RMCE,  H bound is the fastest choice on five out of eight datasets (except for {\it soc-Epinions1, amazon0302, email-EuAll}). 
U bound runs most slowly on seven out of eight datasets (except for {\it soc-Epinions1}).
However, we still notice that the time differences  between $\tau$-R$^+$MCE-TU and $\tau$-R$^+$MCE-CU are narrowed with $\tau$ decreasing for all datasets. 
This is consistent with Fig.~\ref{fig:sbound}:  
since the summary reduction increases with $\tau$ decreasing, 
the benefit of early pruning gradually offsets the cost of bound calculation. 
This explains why $\tau$-R$^+$MCE-TU shows the best performance when $\tau\leq 0.7$. 

\textbf{Memory requirement:}
As we explained in Section~\ref{sc:622}, the result of memory requirement is similar to that of output size. 
The performance of three bounds for both $\tau$-R$^+$MCE and $\tau$-RMCE are quite clear: U is better than C, and C is better than H. 
When we focus on $\tau$-R$^+$MCE, we see that the memory reduction of $\tau$-R$^+$MCE-TU vs. $\tau$-R$^+$MCE-CU  is more than $10\%$  on eight datasets for all $\tau$ values, 
among which four datasets ({\it soc-Epinions1, com-youtube, soc-pokec, cit-Patents}) even achieves $25\%$ at $\tau = 0.5$.  
This reduction is mainly caused by the fact that a tight bound thus early pruning helps to avoid redundant search branches from growing unnecessarily deep, 
which shows the superiority of the truss bound. 

\subsection{Algorithm Realization} 
{In this subsection, we discuss one detail of the algorithm realization: 
 during each recursion of Algorithm~\ref{al2},  to calculate the intersection $C\cap C'$ by merge join, we have to sort the vertices of $C$ according to a given order. 
This is caused by such an undetected fact: even if we strictly select each vertex in candidate set $T'$ by a given order to grow the current partial clique $C$, the order of vertices in $C$ is still disorganized due to  pivot selection.
The other intersection operations ($T\cap \mathcal{N}(v),D\cap \mathcal{N}(v)$) in Algorithm~\ref{al2} have no such  requirement of vertex sorting. 
Ignorance of such a  fact can lead to wrong outputs while the algorithm runs improperly fast. 
This explains why our results of $\tau$-RMCE differ from the work~\cite{wang_redundancy-aware_2013},  
where the results  show that the output size of $\tau$-RMCE changes dramatically   when $\tau$ decreases. 
However, with our appropriate implementation, we indeed find  that both output size and running time decrease much more slowly when $\tau$ varies from $0.9$ to $0.5$. 
Even if we set the default parameter as our newly proposed truss bound T  and truss order U, the decreases are  much less dramatic than the results shown in the work~\cite{wang_redundancy-aware_2013}. 
Besides, we test the algorithms on eight benchmark real-world datasets from various domains, all of which show very similar results.  
Hence we believe that our results are convincing. 
Moreover, we did not include two other datasets {\it Skitter} and {\it Wiki} used by~\cite{wang_redundancy-aware_2013} because  although the vertex numbers of them are less than {\it cit-Patents}, the running times of $\tau$-RMCE  with our appropriate  implementation on these two datasets are not as fast as anticipated. 
Because of the inherent nature of the clique enumeration problem, once the input graph gets large, the processing time becomes unbearable. 
Therefore, we chose datasets that are relatively manageable ($3600s$ for the running time limit). 
For extremely large datasets, we believe that parallel and distributed computing solutions are more preferred, which could be very interesting to explore~\cite{schmidt_scalable_2009,xu_distributed_2016,das2020shared}. 
Accordingly, we are excited to raise this problem in the future work section to attract more attentions from the research community.
}
\subsection{Summary}

After a comprehensive discussion of all  experiments, 
 we can now answer the three  questions  at the beginning of Section~\ref{sec:exp}: 

	(1) $\tau$-R$^+$MCE consistently outperforms $\tau$-RMCE for both effectiveness and efficiency on all datasets with all the $\tau$ values. 
The output reduction can be up to one order of magnitude, and time reduction is more than $35\%$ at $\tau = 0.5$. 
$\tau$-R$^+$MCE achieves relatively better performance on large graphs than $\tau$-RMCE. 

(2) When implemented with $\tau$-R$^+$MCE, the truss order reduces up to $83\%$ output size vs. the state-of-the-art degeneracy order at $\tau = 0.5$, 
and this reduction of truss bound vs. core bound can be up to $43\%$. 
The boost of vertex order is more significant than that of bounds.

(3) The running time of truss order  with $\tau$-R$^+$MCE has comparable or even better performance than the degeneracy order except when implemented on small degeneracy graphs. 
The memory requirement of truss order consistently shows the best performance, of which the reduction vs. degeneracy order is more than $10\%$. 
Although the efficiency of truss bound is surpassed by core bound and H bound, the difference is narrowed with $\tau$ decreasing. 
The memory requirement of truss bound still shows the best performance, which achieves more than $10\%$ reduction vs.  core bound. 

\section{Related Work}\label{sec:related}

The number of maximal cliques in an undirected graph is proved to be exponential~\cite{moon_cliques_1965}. 
Bron and Kerbosch~\cite{bron_algorithm_1973}, Akkoyunlu et al.~\cite{akkoyunlu_enumeration_1973} introduced backtracking algorithms to enumerate all maximal cliques in a graph. 
There are sufficient studies focusing on the efficiency of MCE. 
To effectively reduce the search space, pruning strategies were introduced  in~\cite{koch_enumerating_2001,cazals_note_2008,tomita_worst-case_2006} by selecting good pivots.  
The key idea is to avoid searching in some unnecessary branches which leads to duplicated results. 
Degeneracy vertex ordering was introduced by~\cite{eppstein_listing_2013} to bound the time complexity  
because with the degeneracy  order the size of candidate set $T$ in the first recursion level can be bounded by the degeneracy, thus all the candidate set at all depths of the search tree can be bounded. 
Pivot selection strategies were studied by~\cite{san_segundo_efficiently_2018,naude_refined_2016} to optimize the algorithms.  
Naudé~\cite{naude_refined_2016} relaxed the restriction of pivot selection while keeping the time complexity unchanged. 
Segundo et al.~\cite{san_segundo_efficiently_2018} improved the practical performance of the algorithm by avoiding too much time consumed by selecting the pivot. 
With distributed computing paradigms,  scalable and parallel algorithms were designed for MCE in~\cite{schmidt_scalable_2009,xu_distributed_2016}. 
Schmidt et al.~\cite{schmidt_scalable_2009} decomposed the search tree to enable  parallelization. 
Xu et al.~\cite{xu_distributed_2016} proposed a distributed MCE algorithm based on a share-nothing architecture. 
The I/O performance of MCE in massive networks was improved by~\cite{cheng_finding_2010,cheng_fast_2012}.  
The external-memory algorithm for MCE was first introduced by~\cite{cheng_finding_2010} to bound the memory consumption. 
A partition-based MCE algorithm is designed by~\cite{cheng_fast_2012} to reduce the memory used for processing large graphs. 
{The maximal spatial clique enumeration was studied by~\cite{zhang_efficient_2019}, in which some geometric properties were used to enhance the enumeration efficiency. }
Dynamic maximal clique enumeration was studied in~\cite{DBLP:journals/vldb/DasST19,stix_finding_2004,sun_mining_2017}, in which the graph structure can evolve mildly.  
All the three works considered the dynamic cases where edges can be added or deleted. 
When considering an  uncertain graph, which is a nondeterministic distribution on a set of deterministic graphs, the uncertain version of MCE was designed by~\cite{li_improved_2019,mukherjee_mining_2015}.  
Mukherjee et al.~\cite{mukherjee_mining_2015} designed an algorithm to enumerate all $\alpha$-maximal cliques in an uncertain graph.
The size of an uncertain graph can be reduced by core-based algorithms proposed by~\cite{li_improved_2019}. 
The top-$k$ maximal clique finding problem was also studied by~\cite{zou_finding_2010} on uncertain graphs. 
While these efficient approaches reduced the running time of MCE, the bottleneck in applications is the large output size, which is our main focus.

There exist a large volume of works~\cite{tomita_efficient_2007,tomita_simple_2010,tomita_much_2016,Lu:2017:FMC:3137628.3137660,chang_efficient_2019} studying the maximum clique  problem, which aimed to find a maximal clique with the largest size. 
An approximate coloring technique was employed by~\cite{tomita_efficient_2007} to bound the maximum clique size, which was further improved by~\cite{tomita_simple_2010} and~\cite{tomita_much_2016}. 
Lu et al.~\cite{Lu:2017:FMC:3137628.3137660} proposed a randomized algorithm with a binary search technique to find the maximum clique in massive graphs, 
while the work~\cite{chang_efficient_2019} studied this problem over sparse graphs by transforming the maximum clique in sparse graphs to the $k$-clique  over dense subgraphs. 
Although the concept of maximum clique is closely related to the maximal clique, 
the MCE and maximum clique finding are two distinguishable problems and 
 there is no need to employ a summary to summarize the output of this problem since the number of maximum cliques is typically  small. 

Summarizing has also been studied for frequent pattern mining~\cite{afrati_approximating_2004,yan_summarizing_2005,xin_extracting_2006}. 
Afrati et al.~\cite{afrati_approximating_2004} studied how to find 
at most $k$ patterns to span a collection of patterns  which 
is an approximation of the original pattern sets. 
Yan et al.~\cite{yan_summarizing_2005} proposed a profile-based approach to summarize all frequent patterns by $k$ representives. 
The  pattern redundancy was introduced by~\cite{xin_extracting_2006}, which studied how to extract redundancy-aware and top-$k$ significant patterns.
While cliques share great similarity with frequent patterns, these algorithms cannot be used to summarize maximal cliques efficiently due to their offline nature.
There are some studies focusing on online algorithms to do summarizing. 
Saha et al.~\cite{apte_maximum_2009}
and Ausiello et al.~\cite{owe_online_2011}
 studied how to find diversified $k$ sets to represent all sets with a streaming approach, based on which ~\cite{yuan_diversified_2016} introduced an online algorithm to give diversified top-$k$ maximal cliques. In these works, $k$ is normally small, and coverage is not the focus.  

Our work is close to the work~\cite{wang_redundancy-aware_2013} which introduced the $\tau$-visible summary of maximal cliques.
Other than giving a better sampling function in the earlier version~\cite{li2019mining}, we further discuss the optimality conditions and propose to approach the optimal by introducing the novel truss vertex order and truss bound. 
\section{Conclusion and Future Work}\label{sec:conclusion}

In this paper, we have 
studied how to report a summary of less overlapping maximal cliques 
  during the online maximal clique enumeration process. 
We have proposed so far the best sampling strategy, which can guarantee that
the summary expectedly represents all the maximal cliques
while keeping the summary sufficiently concise, i.e., each maximal clique can be expectedly covered 
by at least one maximal clique in the summary with a ratio of at least $\tau$ ($\tau$ is given by a user and reflects the user's tolerance of overlap). 
{We have proved the optimality  of this sampling approach under 
two conditions  (ideal bound estimation and sufficiently strong  locality), 
and proposed the novel truss order as well as the truss bound to approach the optimal.} 
Experimental studies have shown that 
the new strategy can  outperform the state-of-the-art approach  
in both effectiveness and efficiency
on eight  real-world datasets.
Future work could be conducted towards approaching the optimal conditions further.
It would also be interesting to solve the problem in parallel considering that maximal clique enumeration is expensive on large graphs.


%

\vspace{-10pt}
\section*{Acknowledgments}
The work was supported by Australia Research Council discovery projects DP170104747, DP180100212. 
We would like to thank Yujun Dai for her effort in the earlier version~\cite{li2019mining}. 
\vspace{-10pt}



%
%






\begin{thebibliography}{10}
\providecommand{\url}[1]{{#1}}
\providecommand{\urlprefix}{URL }
\expandafter\ifx\csname urlstyle\endcsname\relax
  \providecommand{\doi}[1]{DOI~\discretionary{}{}{}#1}\else
  \providecommand{\doi}{DOI~\discretionary{}{}{}\begingroup
  \urlstyle{rm}\Url}\fi

\bibitem{afrati_approximating_2004}
Afrati, F., Gionis, A., Mannila, H.: Approximating a collection of frequent
  sets.
\newblock In: Proceedings of the 2004 {ACM} {SIGKDD} International Conference
  on {Knowledge} Discovery and Data Mining, pp. 12--19. ACM Press, Seattle, WA,
  USA
\newblock  (2004)

\bibitem{akkoyunlu_enumeration_1973}
Akkoyunlu, E.: The {enumeration} of {maximal} {cliques} of {large} {graphs}.
\newblock SIAM Journal on Computing \textbf{2}(1), 1--6
\newblock  (1973)

\bibitem{owe_online_2011}
Ausiello, G., Boria, N., Giannakos, A., Lucarelli, G., Paschos, V.T.: Online
  {maximum} k-{coverage}.
\newblock In: O.~Owe, M.~Steffen, J.A. Telle (eds.) Fundamentals of
  {Computation} {Theory}, vol. 6914, pp. 181--192. Springer Berlin Heidelberg,
  Berlin, Heidelberg
\newblock  (2011)

\bibitem{berry2004emergent}
Berry, N., Ko, T., Moy, T., Smrcka, J., Turnley, J., Wu, B.: Emergent clique
  formation in terrorist recruitment.
\newblock In: AAAI-04 Workshop on Agent Organizations: Theory and Practice
\newblock  (2004)

\bibitem{bron_algorithm_1973}
Bron, C., Kerbosch, J.: Algorithm 457: finding all cliques of an undirected
  graph.
\newblock Communications of the ACM \textbf{16}(9), 575--577
\newblock  (1973)

\bibitem{cazals_note_2008}
Cazals, F., Karande, C.: A note on the problem of reporting maximal cliques.
\newblock Theoretical Computer Science \textbf{407}(1-3), 564--568
\newblock  (2008)

\bibitem{chang_efficient_2019}
Chang, L.: Efficient {maximum} {clique} {computation} over {large} {sparse}
  {graphs}.
\newblock In: Proceedings of the 25th {ACM} {SIGKDD} {International}
  {Conference} on {Knowledge} {Discovery} \& {Data} {Mining}, pp. 529--538.
  ACM, New York, USA
\newblock  (2019)

\bibitem{cheng_finding_2010}
Cheng, J., Ke, Y., Fu, A.W.C., Yu, J.X., Zhu, L.: Finding {maximal} {cliques}
  in {massive} {networks} by {H}*-graph.
\newblock In: Proceedings of the 2010 {ACM} {SIGMOD} {International}
  {Conference} on {Management} of {Data}, pp. 447--458. ACM, New York, USA
\newblock  (2010)

\bibitem{cheng_fast_2012}
Cheng, J., Zhu, L., Ke, Y., Chu, S.: Fast {algorithms} for {maximal} {clique}
  {enumeration} with {limited} {memory}.
\newblock In: Proceedings of the 18th {ACM} {SIGKDD} {International}
  {Conference} on {Knowledge} {Discovery} and {Data} {Mining}, pp. 1240--1248.
  ACM, New York, USA
\newblock  (2012)

\bibitem{das2020shared}
Das, A., Sanei-Mehri, S.V., Tirthapura, S.: Shared-memory parallel maximal
  clique enumeration from static and dynamic graphs.
\newblock ACM Transactions on Parallel Computing (TOPC) \textbf{7}(1), 1--28
\newblock  (2020)

\bibitem{DBLP:journals/vldb/DasST19}
Das, A., Svendsen, M., Tirthapura, S.: Incremental maintenance of maximal
  cliques in a dynamic graph.
\newblock The {VLDB} Journal \textbf{28}(3), 351--375
\newblock  (2019)

\bibitem{eppstein_listing_2013}
Eppstein, D., Löffler, M., Strash, D.: Listing {all} {maximal} {cliques} in
  {large} {sparse} {real}-{world} {graphs}.
\newblock Journal of Experimental Algorithmics \textbf{18}, 3.1--3.21
\newblock  (2013)

\bibitem{Khaouid:2015:KDL:2850469.2850471}
Khaouid, W., Barsky, M., Srinivasan, V., Thomo, A.: K-core decomposition of
  large networks on a single \mbox{PC}.
\newblock Proceedings of the VLDB Endowment \textbf{9}(1), 13--23
\newblock  (2015)

\bibitem{koch_enumerating_2001}
Koch, I.: Enumerating all connected maximal common subgraphs in two graphs.
\newblock Theoretical Computer Science \textbf{250}(1), 1--30
\newblock  (2001)

\bibitem{li_improved_2019}
Li, R., Dai, Q., Wang, G., Ming, Z., Qin, L., Yu, J.X.: Improved {algorithms}
  for {maximal} {clique} {search} in {uncertain} {networks}.
\newblock In: 2019 {IEEE} 35th {International} {Conference} on {Data}
  {Engineering} ({ICDE}), pp. 1178--1189
\newblock  (2019)

\bibitem{li2019mining}
Li, X., Zhou, R., Dai, Y., Chen, L., Liu, C., He, Q., Yang, Y.: Mining maximal
  clique summary with effective sampling.
\newblock In: 2019 IEEE International Conference on Data Mining (ICDM), pp.
  1198--1203. IEEE
\newblock  (2019)

\bibitem{Lu:2017:FMC:3137628.3137660}
Lu, C., Yu, J.X., Wei, H., Zhang, Y.: Finding the maximum clique in massive
  graphs.
\newblock Proceedings of the VLDB Endowment \textbf{10}(11), 1538--1549
\newblock  (2017)

\bibitem{lu2018community}
Lu, Z., Wahlstr{\"o}m, J., Nehorai, A.: Community detection in complex networks
  via clique conductance.
\newblock Scientific Reports \textbf{8}(1), 5982--5997
\newblock  (2018)

\bibitem{moon_cliques_1965}
Moon, J.W., Moser, L.: On cliques in graphs.
\newblock Israel Journal of Mathematics \textbf{3}(1), 23--28
\newblock  (1965)

\bibitem{mukherjee_mining_2015}
Mukherjee, A.P., Xu, P., Tirthapura, S.: Mining maximal cliques from an
  uncertain graph.
\newblock In: 2015 {IEEE} 31st {International} {Conference} on {Data}
  {Engineering} (ICDE), pp. 243--254
\newblock  (2015)

\bibitem{naude_refined_2016}
Naudé, K.A.: Refined pivot selection for maximal clique enumeration in graphs.
\newblock Theoretical Computer Science \textbf{613}, 28--37
\newblock  (2016)

\bibitem{rokhlenko2007similarities}
Rokhlenko, O., Wexler, Y., Yakhini, Z.: Similarities and differences of gene
  expression in yeast stress conditions.
\newblock Bioinformatics \textbf{23}(2), 184--190
\newblock  (2007)

\bibitem{apte_maximum_2009}
Saha, B., Getoor, L.: On {maximum} {coverage} in the {streaming} {model} \&
  {application} to {multi}-topic {Blog}-{Watch}.
\newblock In: C.~Apte, H.~Park, K.~Wang, M.J. Zaki (eds.) Proceedings of the
  2009 {SIAM} {International} {Conference} on {Data} {Mining}, pp. 697--708.
  Society for Industrial and Applied Mathematics, Philadelphia, PA
\newblock  (2009)

\bibitem{san_segundo_efficiently_2018}
San~Segundo, P., Artieda, J., Strash, D.: Efficiently enumerating all maximal
  cliques with bit-parallelism.
\newblock Computers \& Operations Research \textbf{92}, 37--46
\newblock  (2018)

\bibitem{schmidt_scalable_2009}
Schmidt, M.C., Samatova, N.F., Thomas, K., Park, B.H.: A scalable, parallel
  algorithm for maximal clique enumeration.
\newblock Journal of Parallel and Distributed Computing \textbf{69}(4),
  417--428
\newblock  (2009)

\bibitem{seidman1983network}
Seidman, S.B.: Network structure and minimum degree.
\newblock Social networks \textbf{5}(3), 269--287
\newblock  (1983)

\bibitem{stix_finding_2004}
Stix, V.: Finding {all} {maximal} {cliques} in {dynamic} {graphs}.
\newblock Computational Optimization and Applications \textbf{27}(2), 173--186
\newblock  (2004)

\bibitem{sun_mining_2017}
Sun, S., Wang, Y., Liao, W., Wang, W.: Mining {maximal} {cliques} on {dynamic}
  {graphs} {efficiently} by {local} {strategies}.
\newblock In: 2017 {IEEE} 33rd {International} {Conference} on {Data}
  {Engineering} ({ICDE}), pp. 115--118
\newblock  (2017)

\bibitem{Tandon:2018:ASH:3237383.3238081}
Tandon, A., Karlapalem, K.: Agent strategies for the hide-and-seek game.
\newblock In: Proceedings of the 17th International Conference on Autonomous
  Agents and MultiAgent Systems, pp. 2088--2090. International Foundation for
  Autonomous Agents and Multiagent Systems, Richland, SC
\newblock  (2018)

\bibitem{tomita_efficient_2007}
Tomita, E., Kameda, T.: An {efficient} {branch}-and-bound {algorithm} for
  {finding} a {maximum} {clique} with {computational} {experiments}.
\newblock Journal of Global Optimization \textbf{37}(1), 95--111
\newblock  (2007)

\bibitem{tomita_simple_2010}
Tomita, E., Sutani, Y., Higashi, T., Takahashi, S., Wakatsuki, M.: A {simple}
  and {faster} {branch}-and-{bound} {algorithm} for {finding} a {maximum}
  {clique}.
\newblock In: M.S. Rahman, S.~Fujita (eds.) {WALCOM}: {Algorithms} and
  {Computation}, Lecture {Notes} in {Computer} {Science}, pp. 191--203.
  Springer Berlin Heidelberg
\newblock  (2010)

\bibitem{tomita_worst-case_2006}
Tomita, E., Tanaka, A., Takahashi, H.: The worst-case time complexity for
  generating all maximal cliques and computational experiments.
\newblock Theoretical Computer Science \textbf{363}(1), 28--42
\newblock  (2006)

\bibitem{tomita_much_2016}
Tomita, E., Yoshida, K., Hatta, T., Nagao, A., Ito, H., Wakatsuki, M.: A {much}
  {faster} {branch}-and-{bound} {algorithm} for {finding} a {maximum} {clique}.
\newblock In: D.~Zhu, S.~Bereg (eds.) Frontiers in {Algorithmics}, Lecture
  {Notes} in {Computer} {Science}, pp. 215--226. Springer International
  Publishing
\newblock  (2016)

\bibitem{valiant1979complexity}
Valiant, L.G.: The complexity of enumeration and reliability problems.
\newblock SIAM Journal on Computing \textbf{8}(3), 410--421
\newblock  (1979)

\bibitem{wang2012truss}
Wang, J., Cheng, J.: Truss decomposition in massive networks.
\newblock arXiv preprint arXiv:1205.6693
\newblock  (2012)

\bibitem{wang_redundancy-aware_2013}
Wang, J., Cheng, J., Fu, A.W.C.: Redundancy-aware maximal cliques.
\newblock In: Proceedings of the 19th {ACM} {SIGKDD} international conference
  on {Knowledge} discovery and data mining, pp. 122--130. ACM Press, Chicago,
  Illinois, USA
\newblock  (2013)

\bibitem{xin_extracting_2006}
Xin, D., Cheng, H., Yan, X., Han, J.: Extracting redundancy-aware top-k
  patterns.
\newblock In: Proceedings of the 12th {ACM} {SIGKDD} International Conference
  on {Knowledge} Discovery and Data Mining, pp. 444--453. ACM Press,
  Philadelphia, PA, USA
\newblock  (2006)

\bibitem{xu_distributed_2016}
Xu, Y., Cheng, J., Fu, A.W.: Distributed {maximal} {clique} {computation} and
  {management}.
\newblock IEEE Transactions on Services Computing \textbf{9}(1), 110--122
\newblock  (2016)

\bibitem{yan_summarizing_2005}
Yan, X., Cheng, H., Han, J., Xin, D.: Summarizing itemset patterns: a
  profile-based approach.
\newblock In: Proceeding of the 11th {ACM} {SIGKDD} International Conference on
  {Knowledge} Discovery in Data Mining, pp. 314--323. ACM Press, Chicago,
  Illinois, USA
\newblock  (2005)

\bibitem{yuan_diversified_2016}
Yuan, L., Qin, L., Lin, X., Chang, L., Zhang, W.: Diversified top-k clique
  search.
\newblock The VLDB Journal \textbf{25}(2), 171--196
\newblock  (2016)

\bibitem{zhang2008pull}
Zhang, B., Park, B.H., Karpinets, T., Samatova, N.F.: From pull-down data to
  protein interaction networks and complexes with biological relevance.
\newblock Bioinformatics \textbf{24}(7), 979--986
\newblock  (2008)

\bibitem{zhang_efficient_2019}
Zhang, C., Zhang, Y., Zhang, W., Qin, L., Yang, J.: Efficient {maximal}
  {spatial} {clique} {enumeration}.
\newblock In: 2019 {IEEE} 35th {International} {Conference} on {Data}
  {Engineering} ({ICDE}), pp. 878--889
\newblock  (2019)

\bibitem{zou_finding_2010}
Zou, Z., Li, J., Gao, H., Zhang, S.: Finding top-k maximal cliques in an
  uncertain graph.
\newblock In: 2010 {IEEE} 26th {International} {Conference} on {Data}
  {Engineering} ({ICDE}), pp. 649--652
\newblock  (2010)

\end{thebibliography}

\end{document}